\documentclass[USenglish]{llncs}

\newcommand{\ExtendedVersion}[1]{#1} \newcommand{\PaperVersion}[1]{}

\usepackage{algorithm}
\usepackage[noend]{algorithmic} 
\usepackage{amsmath}
\usepackage{amssymb}
\usepackage[USenglish]{babel}
\usepackage[colorlinks=true,linkcolor=black,citecolor=black]{hyperref}
\usepackage{stmaryrd}
\usepackage{subfigure}
\usepackage{times}

\hypersetup{bookmarksopen=true}

\ExtendedVersion{\setcounter{tocdepth}{3}}

\newcommand{\enumA}{(i)}
\newcommand{\enumB}{(ii)}
\newcommand{\enumC}{(iii)}
\newcommand{\enumD}{(iv)}

\newcommand{\definedAs}  % the symbol for defining the thing on the LHS
	{=}

\newcommand{\fctsymDom}{\mathrm{dom}} % the function that maps a function/mapping to its domain
\newcommand{\fctDom}[1]{\fctsymDom(#1)}

\newcommand{\constant}[1]{\text{\footnotesize$\mathsf{#1}$}} % to be used for RDF terms in example patterns of example data

\newcommand{\symAllURIs}{\mathcal{I}} % the symbol for the set of all URIs/IRIs
\newcommand{\symAllLiterals}{\mathcal{L}} % the symbol for the set of all literals
\newcommand{\symAllBNodes}{\mathcal{B}} % the symbol for the set of all blank nodes
\newcommand{\symAllTriples}{\mathcal{T}} % the symbol for the set of all RDF triples
\newcommand{\symAllVariables}{\mathcal{V}} % the symbol for the set of all query variables
\newcommand{\symURI}{u} % the symbol for a URI
\newcommand{\symRDFgraph}{G} % the (base) symbol for a set of RDF triples

\newcommand{\fctsymTerms}{\mathrm{terms}} % the function that maps triples / RDF graphs to the set of RDF terms mentioned in them
\newcommand{\fctTerms}[1]{\fctsymTerms(#1)}
\newcommand{\fctsymURIs}{\mathrm{iris}} % the function that maps triples / SPARQL expressions to the set of URIs/IRIs mentioned in them
\newcommand{\fctURIs}[1]{\fctsymURIs(#1)}

\newcommand{\fctsymVars}{\texttt{V}}

\newcommand{\fctVars}[1]{\fctsymVars(#1)}

\newcommand{\fctsymCVars}{\texttt{SBV}}

\newcommand{\fctCVars}[1]{\fctsymCVars(#1)}
\newcommand{\symTP}{tp} % the (base) symbol for a triple pattern
\newcommand{\symPattern}{P} % the (base) symbol for a SPARQL expression

\newcommand{\OpAND}{\text{ \normalfont\scriptsize\textsf{AND} }}
\newcommand{\OpUNION}{\text{ \normalfont\scriptsize\textsf{UNION} }}
\newcommand{\OpOPT}{\text{ \normalfont\scriptsize\textsf{OPT} }}
\newcommand{\OpFILTER}{\text{ \normalfont\scriptsize\textsf{FILTER} }}
\newcommand{\OpSERVICE}{\text{ \normalfont\scriptsize\textsf{SERVICE} }}

\newcommand{\symPPE}{\texttt{path}} % the (base) symbol for a SPARQL property path expression
\newcommand{\invPPE}[1]{\!\,^\wedge\!#1}
\newcommand{\symPPEi}{\invPPE{\symPPE}} % the (base) symbol for an inverse SPARQL property path expression
\newcommand{\symPP}{{P}} % the (base) symbol for a SPARQL property path triple

\newcommand{\muEmpty}{\mu_\emptyset}
\newcommand{\compatible}{\sim}
\newcommand{\incompatible}{\not\sim}

\newcommand{\multicup}{\sqcup}    % union for multisets of solution mappings
\newcommand{\fctsymCard}{\mathit{card}} % the function that appears in multisets of valuations
\newcommand{\fctCard}[1]{\fctsymCard(#1)}

\newcommand{\fctsymCardGen}[1]{\mathsf{card1}^{(#1)}} % the generic cardinality function

\newcommand{\fctOrigQueryPG}[2]{[\![#1]\!]_{#2}} % the symbol for the evaluation function for SPARQL expression #1 over an RDF graph #2

\newcommand{\symWoD}{W} % the (base) symbol for a Web of Linked Data
\newcommand{\symDocs}{D} % the (base) symbol for the set of LD documents in a Web of Linked Data
\newcommand{\fctsymData}{data} % the (base) symbol for the data function of a Web of Linked Data
\newcommand{\fctData}[1]{\fctsymData(#1)}
\newcommand{\fctsymADoc}{adoc} % the (base) symbol for the data function of a Web of Linked Data
\newcommand{\fctADoc}[1]{\fctsymADoc(#1)}

\newcommand{\Reach}[1]%    % the symbol for a component #1 of a reachable part of a Web
	{#1^*}

\newcommand{\symUorVLeft}{\alpha} % the generic left hand side of a property path expression
\newcommand{\symUorVRight}{\beta} % the generic right hand side of a property path expression
\newcommand{\symLeftConst}{x_\mathrm{L}} % the constant left hand side of a property path expression
\newcommand{\symRightConst}{x_\mathrm{R}} % the constant right hand side of a property path expression
\newcommand{\symLeftURI}{\symURI_\mathrm{L}} % the URI left hand side of a property path expression
\newcommand{\symLeftLit}{l_\mathrm{L}} % the literal left hand side of a property path expression
\newcommand{\symLeftVar}{?v_\mathrm{L}} % the variable left hand side of a property path expression
\newcommand{\symRightVar}{?v_\mathrm{R}} % the variable right hand side of a property path expression

\newcommand{\EvalCtx}[2]{\llbracket #1 \rrbracket^{\symStandSemCtx}_{#2}}
\newcommand{\EvalFull}[2]{\llbracket #1 \rrbracket^{\symFullWebSem}_{#2}}

\newtheorem{fact}[proposition]{Fact}
\newtheorem{Note}[note]{\textbf{Note}}
\newtheorem{rmrk}[remark]{\textbf{Remark}}

\newtheorem{exmp}[example]{\textbf{Example}}
\newenvironment{proofidea}[1]{\paragraph{Proof idea (#1).}}{\qed\vspace{1ex}}

\newcommand{\fctsymContext}[1]%
	{\mathrm{C}^{#1}}
\newcommand{\fctContext}[2]{\fctsymContext{#1}\!(#2)}

\newcommand{\symFullWebSem}{\texttt{fw}}
\newcommand{\symStandSemCtx}%
	{\texttt{ctx}}

\newcommand{\Web}{WoLD}
\newcommand{\tupleD}[1]{\langle #1 \rangle}

\newcommand{\symWebSafeVar}{X}

\newcommand{\fctsymWebSafe}{\texttt{CBV}}

\newcommand{\fctCondWebSafe}[2]{\fctsymWebSafe(#1\,|\,#2)}

\newcommand{\captionspace}%
	{}

\begin{document}
\title{A Context-Based Semantics for\\SPARQL Property Paths over the Web}
\ExtendedVersion{%
\subtitle{(Extended Version)\footnote{This document is an extended version of a paper
	published in ESWC~2015~\cite{ProceedingsVersion}.%
	%submitted to ESWC~2015.\nocite{ExtendedVersion}%
	}}%
}

\author{Olaf Hartig\inst{1} \and Giuseppe Pirr\`o\inst{2}}
\authorrunning{O. Hartig \and G. Pirr\`o}   % abbreviated author list (for running head)

\institute{University of Waterloo, Canada\\
\email{ohartig@uwaterloo.ca}\\
\and
Institute for High Performance Computing and Networking, ICAR-CNR, Rende, Italy\\
\email{pirro@icar.cnr.it}
}

\tocauthor{
    Olaf Hartig (University of Waterloo)
    Giuseppe Pirr\`o (ICAR-CNR)
}

\maketitle

\ExtendedVersion{\vspace{-7mm}}

\begin{abstract}
As of today, there exists no standard language for querying Linked Data \textit{on the Web}, where navigation across distributed data sources is a key feature.
A natural candidate seems to be SPARQL, which recently has been enhanced with navigational capabilities thanks to the introduction of \textit{property paths} (PPs). However, the semantics of SPARQL restricts the scope of navigation via PPs to \textit{single} RDF graphs. This restriction limits the applicability of PPs on the Web. 
To fill this gap, in this paper we provide formal foundations for evaluating PPs on the Web, thus contributing to the definition of a query language for Linked Data. In particular, we introduce a
	%con\-text-based
query semantics for PPs that couples navigation at the data level with navigation on the Web graph. Given this semantics we find that for some PP-based SPARQL queries a complete evaluation on the Web is not feasible. To enable systems to identify queries that can be evaluated completely, we establish a decidable syntactic property of such queries.
\end{abstract}

\ExtendedVersion{\vspace{-11mm}}

% %
% %
% %

\section{Introduction}
\label{sec:introduction}
% % % % %%
The increasing trend in sharing and interlinking pieces of structured data on the World Wide Web (WWW) is evolving the classical Web---which is focused on hypertext documents and syntactic links among them---into a
%%% Ideally, the term "Web of Linked Data" should be avoided here because
%%% it is the name of a concept in the data model. "WoLD" should definitively
%%% be avoided for that reason.
%%% The "real" thing captured by this abstract concept is the WWW or, at
%%% least, the Linked Data portion of the WWW.
%%% However, I kept "Web of Linked Data" at this point because it is the
%%% best way to go into the next sentence.
	Web of Linked Data.
The Linked Data principles~\cite{BernersLee:LinkedData} present an approach to extend the scope of Uniform Resource Identifiers~(URIs) to new types of resources (e.g., people, places%
	%, football teams%
) and represent their descriptions and interlinks by using the Resource Description Framework (RDF)~\cite{klyne2006} as standard data format. RDF adopts a graph-based data model, which can be queried upon by using the SPARQL query language~\cite{Harris2013}.
 When it comes to Linked Data on the WWW, the common way to provide query-based access is via SPARQL endpoints, that is, services that usually answer SPARQL queries over a single dataset. Recently, the original core of SPARQL has been extended with features supporting query federation; it is now possible, within a single query, to target multiple endpoints (via the\OpSERVICE operator). However, such an extension is not enough to cope with an unbounded and a priori unknown space of data sources such as the WWW. Moreover,
	not all Linked Data on the WWW is accessible via SPARQL endpoints.
	%the assumption that all Linked Data on the WWW is \textit{reachable} via SPARQL endpoints does not hold in general.
% % EXPLAIN IN DETAIL WHY IT IS NOT ENOUGH
%
Hence, as of today, there exists no standard query language for Linked Data on the WWW%
	, although SPARQL is clearly~a~candidate.
	%. Toward the definition of such a language this paper studies the applicability of the navigational features introduced by SPARQL~1.1.

While earlier research
	on using SPARQL for Linked Data is limited to fragments of the first version of the language~\cite{Bouquet09:QueryingWebOfData,harth2012,Hartig12:TheoryPaper,Umbrich2014}, 
	%studies the applicability of fragments of the first version of SPARQL in the context of Linked Data~\cite{Bouquet09:QueryingWebOfData,harth2012,Hartig12:TheoryPaper,Umbrich2014},
the more recent version 1.1 introduces a feature that is particularly interesting in the context of queries over a graph-like environment such as Linked Data on the WWW. This feature is called \textit{property paths}~(PPs) and equips SPARQL with navigational capabilities~\cite{Harris2013}. However, the standard definition of PPs is limited to single, centralized RDF graphs and, thus, not directly applicable to Linked Data that is distributed over the WWW. Therefore, toward the definition of a language for accessing Linked Data live on the WWW, the
	following questions emerge naturally: \emph{``How can PPs be defined over the WWW?''} and \emph{``What are the implications of such a definition?''} Answering these questions
	%question of \emph{``How can PPs be defined over the WWW?''} emerges naturally. Answering this question
is the broad objective of this paper. To this end, we make the following main contributions:

\begin{enumerate}
\item We formalize a query semantics for PP-based SPARQL queries that are meant to be evaluated over Linked Data on the WWW. This semantics is \textit{con\-text-based}; it intertwines Web graph navigation with navigation at the level of data.
\item We study the feasibility of evaluating queries under this semantics. We assume that query engines do not have complete information about the queried Web of Linked Data~(as it is the case for the WWW). Our study shows that
	%query evaluation under the con\-text-based semantics is not feasible in general.
	there exist cases in which query evaluation under the con\-text-based semantics is not feasible.
\item We provide a
	%decidable syntactic property to identify queries whose evaluation is feasible under the con\-text-based semantics.
	decidable syntactic property of queries for which an evaluation under the con\-text-based semantics is feasible.
\end{enumerate}

\noindent
% We believe that the outcome of this research can be a starting point toward the definition of a
% 	%comprehensive
% language for
% 	querying and navigating over
% 	%accessing
% Linked Data on the WWW.
% 
% 
% \smallskip
% \medskip \noindent \textbf{Organization.}
%
The remainder of the paper is organized as follows. Section~\ref{sec:RelWork} provides an overview on related work. Section~\ref{sec:FormalFramework} introduces the formal framework for this paper, including a data
	%model that captures a notion of Linked Data on the WWW.
	model that captures a notion of Linked Data.
	%model of Linked Data on the WWW.
	%model of Linked Data.
	%model.
In Section~\ref{sec:the-semantics} we focus on PPs, independently from other SPARQL operators%
	%, and define different semantics for PPs over the Web%
. In Section~\ref{sec:sparql} we broaden our view to study PP-based SPARQL
	graph
patterns; we characterize a class of \textit{Web-safe} patterns and prove their feasibility. 
%Section~\ref{sec:related-work} provides an overview of related research. 
Finally, in Section~\ref{sec:conclusion} we conclude and sketch future~work.

\section{Related Work} \label{sec:RelWork}

The idea of querying the WWW as a database is not new~(see Florescu et al.'s survey~\cite{Florescu1998}). Perhaps the most notable early works
	in this context
are by Konopnicki and Shmueli~\cite{Konopnicki98:W3QLandW3QS}, Abiteboul and Vianu~\cite{Abiteboul00}, and Mendelzon et al.~\cite{Mendelzon97}, all of which tackled the problem of evaluating SQL-like queries on the traditional hypertext Web. While such queries
	included
	%allowed users to include
navigational features, the focus was on retrieving specific Web pages, particular attributes of specific pages, or content within them.

From a
	%more
graph-ori\-ent\-ed perspective, languages for the {\em navigation and specification} of vertices in graphs have a long tradition (see Wood's survey~\cite{Wood12}). In the
	RDF world,
	%context of RDF,
%
	%PSPARQL~\cite{Alkhateeb09}, nSPARQL~\cite{perez2010}, SPARQLeR~\cite{kochut2007sparqler} defined extensions of SPARQL with navigational features
	extensions of SPARQL such as PSPARQL~\cite{Alkhateeb09}, nSPARQL~\cite{perez2010}, and SPARQLeR~\cite{kochut2007sparqler} introduced navigational features
since those were missing in the first version of SPARQL. Only recently, with the addition of \textit{property paths}~(PPs) in version~1.1~\cite{Harris2013}, SPARQL has been enhanced officially with such features.
	%Recently, Fionda et. al~\cite{fionda2015aaai} introduced an extension of PPs.
	The final definition of PPs has been influenced by research that studied the computational complexity of an early draft version of PPs~\cite{arenas2012,Loseman12:SPARQLPropertyPaths}, and there also already exists a proposal to extend PPs with more expressive power~\cite{fionda2015aaai}.
%
	%The
	However, the
main assumption of all these navigational extensions of SPARQL is to work on a single, centralized RDF graph.
	%All these navigational extensions of SPARQL are meant to work on centralized data collections.
Our departure point is different: \emph{We aim at defining
	%query
semantics of SPARQL queries~(including property paths) over Linked Data on the WWW,} which involves dealing with two
	%types of graphs;
	graphs of different types;
namely, an RDF graph that is distributed over
	documents on the WWW
	%Web documents
and the Web graph of how these documents are interlinked with each other.

To express queries over Linked Data on the WWW, two main strands of research can be identified. The first studies how to extend the scope of SPARQL queries to the WWW%
	%~\cite{Bouquet09:QueryingWebOfData,harth2012,Hartig12:TheoryPaper,Umbrich2014}. %The idea is to retrieve Linked Data from the WWW and \textit{navigate} toward other Linked Data while executing a query.
	, with existing work focusing on basic graph patterns~\cite{Bouquet09:QueryingWebOfData,harth2012,Umbrich2014} or a more expressive fragment that includes \!\OpAND\!, \!\OpOPT\!, \!\OpUNION\!~and \!\OpFILTER\!~\cite{Hartig12:TheoryPaper}.
The second
	strand
focuses on
	%graph navigational features (e.g., NautiLOD~\cite{fionda2012}).
	navigational languages such as NautiLOD~\cite{fionda2012,FiondaTWEB2015}.
These two strands have different departure points. The former
%%% Comment: it is not uncontrolled - the navigation does not simply
%%% happen. Instead, the system has full control over the navigation,
%%% and even the user has some control, e.g., by choosing a specific
%%% reachability criterion.          Olaf
	%%% uses ``implicit'' (not-controlled)
	employs navigation over the WWW to collect data for answering a given SPARQL query; here navigation is a means to discover que\-ry-rel\-e\-vant data. The latter provides explicit navigational features and uses querying capabilities to filter data sources of interest; here navigation~(not querying) is the main~focus.
%%% Commented the following sentence because the previous sentences
%%% sufficiently cover this topic / already make the point.   Olaf
	%While \nautilod\ uses queries to enhance/control the navigation, LTQBE proceeds in the reverse direction.
% % % The COMBINATION OF The two is what is needed!
The con\-text-based query semantics proposed in this paper
	%can be understood as a combination of
	%present a combination of
	combines
both approaches. We believe that the outcome of this research can be a starting point toward the definition of a
	%comprehensive
language for querying and navigating over Linked Data on the WWW.

 \section{Formal Framework} \label{sec:FormalFramework}
This section provides a formal framework for studying semantics of PPs over Linked Data.
	%In particular, we
	We
first recall the definition of PPs as per the SPARQL standard~\cite{Harris2013}. Thereafter, we introduce a data model that captures the notion of Linked Data on the WWW.

\subsection{Preliminaries} \label{subsec:preliminaries}

Assume four pairwise disjoint, countably infinite sets $\symAllURIs$ (IRIs), $\symAllBNodes$ (blank nodes), $\symAllLiterals$~(literals), and $\symAllVariables$ (variables%
	%, denoted by a leading '?' symbol%
). An \emph{RDF triple} (or simply \emph{triple}) is a tuple from the set $\symAllTriples = (\symAllURIs \cup \symAllBNodes) \times \symAllURIs \times (\symAllURIs \cup \symAllBNodes \cup \symAllLiterals)$.
For any triple
	%$t = \tupleD{s,p,o}$
	$t \in \symAllTriples$
we write $\fctURIs{t}$ to denote the set of IRIs in that triple. % (i.e., $\fctURIs{t} = \lbrace s,p,o \rbrace \cap \symAllURIs$).
A set of triples is called an \emph{RDF graph}.

A \emph{property path pattern} (or \emph{PP pattern} for short) is a tuple $\symPP = \tupleD{\symUorVLeft, \symPPE, \symUorVRight}$ such that $\symUorVLeft,\symUorVRight \in (\symAllURIs \cup \symAllLiterals \cup \symAllVariables)$ and $\symPPE$ is a \emph{property path expression} (\emph{PP expression}) defined by the following grammar (where $\symURI, \symURI_1, \ldots, \symURI_n \in \symAllURIs$):
%
% \vspace{-3mm}
\begin{align*}
	\symPPE \, \definedAs \, \, &
		\symURI
	\,\mid\,
		\, !( \symURI_1 \,|\, \ldots \,|\, \symURI_n )
% 	\,\mid\,
% 		\, !\bigl(\symURIi_1 \,|\, \ldots \,|\, \symURIi_n \bigr)
% 	\,\mid\, \\ &
% 		\, !\bigl( \symURI_1 \,|\, \ldots \,|\, \symURI_k \,|\, \symURIi_{k+1} \,|\, \ldots \,|\, \symURIi_n  \bigr)
	\,\mid\,
		\symPPEi
	\,\mid\,
		\symPPE / \symPPE
	\,\mid\,
		(\symPPE \,|\, \symPPE)
	\,\mid\, 
%\\ &
		(\symPPE)^*
% 	\,\mid\,
% 		(\symPPE)^+
% 	\,\mid\,
% 		(\symPPE)?
\end{align*}

\enlargethispage{\baselineskip} % Layout Adjustment

\noindent
Note that the SPARQL standard introduces
	%a few
additional types of PP expressions~\cite{Harris2013}. Since these are merely syntactic sugar (they are defined in terms of expressions covered by the grammar given above), we ignore them in this paper.
As another slight deviation from the
	%SPARQL
standard, we do not permit blank nodes in PP patterns (i.e., $\symUorVLeft,\symUorVRight \notin \symAllBNodes$). However, standard
	%SPARQL
	PP
patterns with blank nodes can be simulated using fresh variables.

\begin{exmp}\label{ex:PPsyntax}\hspace{-1ex}\textbf{.}
		%As an example of a PP pattern consider
		An example of a PP pattern is
	$\tupleD{\constant{Tim},(\constant{knows})^*/\constant{name},?n}$, which retrieves the names of persons that can be reached from $\constant{Tim}$ by an arbitrarily long path of $\constant{knows}$ relationships~(which includes $\constant{Tim}$).
	Another example are the two PP patterns $\tupleD{?p,\constant{knows},\constant{Tim}}$ and $\tupleD{\constant{Tim},\invPPE{\constant{knows}},?p}$, both of which retrieve persons that know $\constant{Tim}$.
	%For further examples we refer to the SPARQL specification~\cite[Section~9.2]{Harris2013}.
\end{exmp}

\noindent
The (standard)
	query
semantics of PP patterns is defined by an evaluation function that returns multisets
	%(bags)
of
	%so called solution mappings; such a \emph{solution mapping}
	\emph{solution mappings} where a solution mapping $\mu$
is a partial function $\mu : \symAllVariables \rightarrow (\symAllURIs \cup \symAllBNodes \cup \symAllLiterals)$. Given a solution mapping $\mu$ and a PP pattern $\symPP$, we write $\mu[\symPP]$ to denote the PP pattern obtained by replacing the variables in $\symPP$ according to $\mu$ (unbound variables must not be replaced). Two solution mappings, say $\mu_1$ and $\mu_2$, are \emph{compatible}, denoted by $\mu_1 \compatible \mu_2$, if $\mu_1(?v)=\mu_2(?v)$ for all variables $?v \in \bigl( \fctDom{\mu_1} \cap \fctDom{\mu_2} \bigr)$.

We represent a \emph{multiset} of solution mappings by a pair $M = \tupleD{\Omega,\fctsymCard}$ where
	%$\Omega$ is
	$\Omega$~is
the underlying set (of solution mappings) and
	%$\fctsymCard$ is the corresponding \emph{cardinality function}; i.e., $\fctsymCard : \Omega \rightarrow \lbrace 1,2, ... \, \rbrace$.
	$\fctsymCard : \Omega \rightarrow \lbrace 1,2, ... \, \rbrace$ is the corresponding \emph{cardinality function}.
By abusing notation slightly, we write $\mu \in M$ for all $\mu \in \Omega$.
Furthermore, we introduce a family of special (parameterized) cardinality functions that shall simplify the definition of any multiset whose solution mappings all have a cardinality of 1. That is, for any set of solution mappings $\Omega$, let $\fctsymCardGen{\Omega}\! : \Omega \!\rightarrow\! \lbrace 1,2, ... \rbrace$ be the \emph{constant-1 cardinality function} that is defined by $\fctsymCardGen{\Omega}(\mu) = 1$ for all $\mu \in \Omega$.
% % %

To define the aforementioned evaluation function
	%for PP patterns
we also need to introduce several SPARQL algebra operators. Let $M_1 = \tupleD{\Omega_1,\fctsymCard_1}$ and $M_2 = \tupleD{\Omega_2,\fctsymCard_2}$ be multisets of~solution mappings and let $V \subseteq \symAllVariables$ be a finite set of variables. Then:
\begin{description}
	\itemsep1.5mm % Layout Adjustment
	%%% MULTISET UNION
	\item
		[$M_1 \multicup M_2 \definedAs \tupleD{\Omega,\fctsymCard}$] where $\Omega = \Omega_1 \cup \Omega_2$
			%and, \enumA~for each $\mu \!\in\! \Omega \setminus \Omega_2$, $\fctCard{\mu} \!=\! \fctsymCard_1(\mu)$, \enumB~for each $\mu \!\in\! \Omega \setminus \Omega_1$, $\fctCard{\mu} = \fctsymCard_2(\mu)$, and \enumC~for each $\mu \in \Omega_1 \cap \Omega_2$, $\fctCard{\mu} = \fctsymCard_1(\mu) + \fctsymCard_2(\mu)$.
			and \enumA~$\fctCard{\mu} = \fctsymCard_1(\mu)$ for all solution mappings $\mu \in \Omega \setminus \Omega_2$, \enumB~$\fctCard{\mu} = \fctsymCard_2(\mu)$ for all $\mu \in \Omega \setminus \Omega_1$, and \enumC~$\fctCard{\mu} = \fctsymCard_1(\mu) + \fctsymCard_2(\mu)$ for all $\mu \in \Omega_1 \cap \Omega_2$.
	%%% JOIN
	\item%
		[$M_1 \Join M_2 \definedAs \tupleD{\Omega,\fctsymCard}$] where $\Omega = \big\lbrace \, \mu_1 \!\cup \mu_2 \,|\, (\mu_1,\mu_2) \in \Omega_1\!\times \Omega_2 \text{ and } \mu_1 \compatible \mu_2 \big\rbrace$
			%and, for every solution mapping $\mu \in \Omega$,
			and, for every $\mu \in \Omega$,
		$\fctCard{\mu} = \sum_{(\mu_1\!,\mu_2)\in \Omega_1\!\times\Omega_2 \text{ s.t. } \mu = \mu_1 \cup \mu_2} \fctCard{\mu_1} \cdot \fctCard{\mu_2}$.
	%%% MINUS
	\item%
		[$M_1 \setminus M_2 \definedAs \tupleD{\Omega,\fctsymCard}$] where $\Omega = \big\lbrace \, \mu_1 \in \Omega_1 \,|\, \mu_1 \incompatible \mu_2 \text{ for all } \mu_2 \in \Omega_2 \big\rbrace$ and, for every $\mu \in \Omega$, $\fctCard{\mu} = \fctsymCard_1(\mu)$.
	%%% PROJECTION
	\item%
		[$\pi_V (M_1) \definedAs \tupleD{\Omega,\fctsymCard}$] where $\Omega = \big\lbrace \mu \,|\, \exists \mu' \!\in\! \Omega_1 \!: \mu \!\compatible\! \mu' \text{ and } \fctDom{\mu} \!=\! V \cap \fctDom{\mu'} \big\rbrace$ and, for every $\mu \in \Omega$, $\fctCard{\mu} = \sum_{\mu' \!\in \Omega_1 \text{ s.t. } \mu \compatible \mu'} \fctsymCard_1(\mu')$. 
\end{description}

\noindent
In addition to these algebra operators, the SPARQL standard introduces auxiliary functions to define the semantics of PP patterns of the form $\tupleD{\symUorVLeft, \symPPE^*\!, \symUorVRight}$. 
Figure~\ref{Figure:StdALP} provides these
	functions---which we call $\mathtt{ALP1}$ and $\mathtt{ALP2}$---adapted
	%functions adapted
to our
	formalism.\footnote{Variable $?x$ in line~\ref{line:StdALP:Loop} is necessary since PP patterns in our formalism do not have blank~nodes.}
	%formalism (e.g., variable $?x$ in line~\ref{line:StdALP:Loop} is necessary since PP patterns in our formalism do not have blank~nodes).

\begin{figure}[t]
{\small%
	\subfigure{%
		\begin{minipage}[t]{0.35\textwidth}%
		\underline{\textbf{Function} $\mathtt{ALP1}\bigl( \gamma, \symPPE, \symRDFgraph \bigr)$}
		\par
		\textbf{Input:} $\gamma \in ( \symAllURIs \cup \symAllBNodes \cup \symAllLiterals )$,
		\par \hspace{9mm} $\symPPE$ is a PP expression,
		\par \hspace{9mm} $\symRDFgraph$ is an RDF graph.
		\par
		\begin{algorithmic}[1]
			\STATE \textit{Visited} := $\emptyset$
			\STATE $\mathtt{ALP2}\bigl( \gamma, \symPPE, \mathit{Visited}, \symRDFgraph \bigr)$
			\RETURN \textit{Visited}
		\end{algorithmic}%
		\end{minipage}%
	}%
	\subfigure{%
		\begin{minipage}[t]{0.63\textwidth}%
		\underline{\textbf{Function} $\mathtt{ALP2}\bigl( \gamma, \symPPE, \mathit{Visited}, \symRDFgraph \bigr)$}
		\par
		\textbf{Input:} $\gamma \in ( \symAllURIs \cup \symAllBNodes \cup \symAllLiterals )$,
		~ $\symPPE$ is a PP expression,
		\par \hspace{8mm} $\mathit{Visited} \subseteq ( \symAllURIs \cup \symAllBNodes \cup \symAllLiterals )$,
		~ $\symRDFgraph$ is an RDF graph.
		\par
		\begin{algorithmic}[1]
		\setcounter{ALC@line}{3}
			\IF {$\gamma \notin \mathit{Visited}$}
				\STATE add $\gamma$ to \textit{Visited}
%%%
%%% The additional condition in the following for all statement is necessary
%%% to avoid that we end up with PP patterns that have bnodes (because \gamma
%%% can be a bnode). Nonetheless, to be compatible with the standard we want
%%% to support bnodes within paths that match a kleene star based PP expression.
%%%                                   Olaf, May 6, 2014
%%%
				\FORALL {$\mu \in \fctOrigQueryPG{\tupleD{?x, \symPPE, ?y}}{\symRDFgraph}$ \, s.t. \, $\mu(?x) = \gamma$} \label{line:StdALP:Loop}
					\STATE $\mathtt{ALP2}\bigl( \mu(?y), \symPPE, \mathit{Visited}, \symRDFgraph \bigr)$ \quad \, \COMMENT{$?x,?y \in \symAllVariables$} \label{line:StdALP:Recursion}
				\ENDFOR
			\ENDIF
		\end{algorithmic}%
		\end{minipage}%
	}%
\vspace{-3mm}
\caption{\textbf{Auxiliary functions
	%used
for defining the semantics of PP expressions of the form} $\symPPE^*$\!\textbf{.}}%
\label{Figure:StdALP}%
%
% \vspace{-5mm}
}%
\end{figure}

We are now ready to define the standard query semantics of PP patterns.

\begin{definition} \label{def:StdDefinition}
	The \emph{evaluation} of a PP pattern $\symPP$ over an RDF graph $\symRDFgraph$, denoted by $\fctOrigQueryPG{\symPP}{\symRDFgraph}$, is a multiset of solution mappings $\tupleD{\Omega,\fctsymCard}$ that is defined recursively as given in Figure~\ref{fig:w3c-standard} where $\symUorVLeft,\symUorVRight \in (\symAllURIs \cup \symAllLiterals \cup \symAllVariables)$, $\symLeftConst,\symRightConst \in ( \symAllURIs \cup \symAllLiterals )$, $\symLeftVar,\symRightVar \in \symAllVariables$, $\symURI, \symURI_1, ..., \symURI_n \in \symAllURIs$, $?v \in \symAllVariables$ is a fresh variable, and $\muEmpty$ denotes the empty solution mapping (%
		%i.e.,
	$\fctDom{\muEmpty} = \emptyset$).
\end{definition}

% \vspace{-5mm}

\begin{figure}[t]
{\scriptsize
\begin{gather}
\begin{align*}
	% BASE CASE (I.E., PREDICATE)
	\fctOrigQueryPG{\tupleD{\symUorVLeft, \symURI, \symUorVRight}}{\symRDFgraph} & \definedAs
		\Big\langle
			\big\lbrace\, \mu \,|\, \fctDom{\mu}=( \lbrace \symUorVLeft,\symUorVRight \rbrace \cap \symAllVariables ) \text{ and } \mu[\tupleD{\symUorVLeft, \symURI, \symUorVRight}] \in \symRDFgraph \,\big\rbrace  \textbf{ , }\,
			\fctsymCardGen{\Omega}
		\Big\rangle
	\\
	% NEGATION
	\fctOrigQueryPG{\tupleD{\symUorVLeft, !( \symURI_1 \,|\, \ldots \,|\, \symURI_n ), \symUorVRight}}{\symRDFgraph} & \definedAs
		\Big\langle
			\big\lbrace\, \mu \,|\, \fctDom{\mu} = \bigl( \lbrace \symUorVLeft,\symUorVRight \rbrace \cap \symAllVariables \bigr) \text{ and }  \\[-2mm]
			& \hspace{11.5mm} \exists\ \mu[\tupleD{\symUorVLeft, \symURI, \symUorVRight}] \in \symRDFgraph : \symURI \in \bigl( \symAllURIs \setminus \lbrace \symURI_1 , \ldots , \symURI_n \rbrace \bigr) \,\big\rbrace  \textbf{ , }\,
			\fctsymCardGen{\Omega}
		\Big\rangle
	\\
% 	\fctOrigQueryPG{\tupleD{\symUorVLeft, !(\symURIi_1 \,|\, \ldots \,|\, \symURIi_k ), \symUorVRight}}{\symRDFgraph} & \definedAs
% 		\Big\langle
% 			\todo{...}  \textbf{ , }\,
% 			\fctsymCardGen{\Omega}
% 		\Big\rangle
% 	\\
% 	\fctOrigQueryPG{\tupleD{\symUorVLeft, !( \symURI_1 \,|\, \ldots \,|\, \symURI_k \,|\, \symURIi_{k+1} \,|\, \ldots \,|\, \symURIi_n  ), \symUorVRight}}{\symRDFgraph} & \definedAs
% 		\Big\langle
% 			\todo{...}  \textbf{ , }\,
% 			\fctsymCardGen{\Omega}
% 		\Big\rangle
	%
	% INVERSE
	\fctOrigQueryPG{\tupleD{\symUorVLeft, \symPPEi, \symUorVRight}}{\symRDFgraph} & \definedAs
		\fctOrigQueryPG{\tupleD{\symUorVRight, \symPPE, \symUorVLeft}}{\symRDFgraph} 
	\\
	% CONCATENATION
	\fctOrigQueryPG{\tupleD{\symUorVLeft, \symPPE_1/\symPPE_2, \symUorVRight}}{\symRDFgraph} & \definedAs
		\pi_{ \lbrace \symUorVLeft,\symUorVRight \rbrace \cap \symAllVariables } \Bigl( \fctOrigQueryPG{\tupleD{\symUorVLeft, \symPPE_1, ?v}}{\symRDFgraph} \Join \fctOrigQueryPG{\tupleD{?v, \symPPE_2, \symUorVRight}}{\symRDFgraph} \Bigr)
	\\
	% ALTERNATIVE
	\fctOrigQueryPG{\tupleD{\symUorVLeft, (\symPPE_1 |\, \symPPE_2), \symUorVRight}}{\symRDFgraph} & \definedAs
		\fctOrigQueryPG{\tupleD{\symUorVLeft, \symPPE_1, \symUorVRight}}{\symRDFgraph} \multicup \fctOrigQueryPG{\tupleD{\symUorVLeft, \symPPE_2, \symUorVRight}}{\symRDFgraph}
	\\
	% KLEENE STAR
	\fctOrigQueryPG{\tupleD{\symLeftConst, (\symPPE)^*, \symRightVar}}{\symRDFgraph} & \definedAs
		\Big\langle
			\big\lbrace\, \mu \,|\, \fctDom{\mu} = \lbrace \symRightVar \rbrace \text{ and } \mu(\symRightVar) \in \mathtt{ALP1}(\symLeftConst,\symPPE,\symRDFgraph) \,\big\rbrace  \textbf{ , }\,
			\fctsymCardGen{\Omega}
		\Big\rangle
	\\
	\fctOrigQueryPG{\tupleD{\symLeftVar, (\symPPE)^*, \symRightVar}}{\symRDFgraph} & \definedAs
		\Big\langle
			\big\lbrace\, \mu \,|\, \fctDom{\mu} = \lbrace \symLeftVar,\symRightVar \rbrace \text{ and } \mu(\symLeftVar) \in \fctTerms{\symRDFgraph} \text{ and } 
			\\[-2mm]
						& \hspace{11.5mm}
			\mu(\symRightVar) \in \mathtt{ALP1}(\mu(\symLeftVar),\symPPE,\symRDFgraph) \,\big\rbrace  \textbf{ , }\,
			\fctsymCardGen{\Omega}
		\Big\rangle
	\\
	\fctOrigQueryPG{\tupleD{\symLeftVar, (\symPPE)^*, \symRightConst}}{\symRDFgraph} & \definedAs
		\fctOrigQueryPG{\tupleD{\symRightConst, (\symPPEi)^*, \symLeftVar}}{\symRDFgraph}
	\\
	\fctOrigQueryPG{\tupleD{\symLeftConst, (\symPPE)^*, \symRightConst}}{\symRDFgraph} & \definedAs
		\Big\langle
			\begin{cases} \lbrace \muEmpty \rbrace & \text{if } \exists\ \mu \in \fctOrigQueryPG{\tupleD{\symLeftConst, (\symPPE)^*, ?v}}{\symRDFgraph} : \mu(?v) = \symRightConst , \\ ~~\emptyset & \text{else} \end{cases} \textbf{ , }\,
			\fctsymCardGen{\Omega}
		\Big\rangle
\end{align*}%
\end{gather}%
}%
%
% \vspace{-1cm}
%
\caption{\textbf{SPARQL 1.1 W3C property paths semantics.}}
%
% \vspace{-5mm}
%
\label{fig:w3c-standard}
\end{figure}
\subsection{Data Model} \label{subsec:DataModel}
The standard SPARQL evaluation function for PP patterns~(cf.~Section~\ref{subsec:preliminaries}) defines the expected result of the evaluation of a pattern over a single RDF graph. Since the WWW is not an RDF graph, the standard definition is insufficient as a formal foundation for evaluating PP patterns over Linked Data on the WWW. To provide a suitable definition we need a data model that captures the notion of a Web of Linked Data. To this end, we adopt the data model proposed in our earlier work~\cite{Hartig12:TheoryPaper}. Here, a \emph{Web of Linked Data}~(\emph{\Web}) is a tuple $\symWoD=\tupleD{\symDocs,\fctsymData,\fctsymADoc}$ consisting of
\enumA~a set $\symDocs$ of so called \emph{Linked Data documents} (\emph{documents}),
\enumB~a mapping $\fctsymData : \symDocs \rightarrow 2^\symAllTriples$ that maps each document to a finite set of RDF triples (representing the data that can be obtained from the document),
and \enumC~a partial mapping $\fctsymADoc : \symAllURIs \rightarrow \symDocs$ that maps (some) IRIs to a document and, thus, captures
	%the URI-based retrieval of documents on the WWW.
	a IRI-based retrieval of documents.
In this paper we assume that the set of documents $\symDocs$ in any \Web\ $\symWoD = \tupleD{\symDocs,\fctsymData,\fctsymADoc}$ is finite, in which case we say $\symWoD$ is \emph{finite} (for a discussion of infiniteness refer to our earlier~work~\cite{Hartig12:TheoryPaper}).

%We emphasize that, in contrast to an RDF graph (the mathematical structure), the three elements of a \Web\ (i.e., $\symDocs$, $\fctsymData$, and $\fctsymADoc$) typically are not available directly to systems that aim to compute queries over the Web captured as a \Web. For instance, the complete domain of partial mapping $\fctsymADoc$ (i.e., all URIs that can be used to retrieve some document) is unknown to such systems and can only be disclosed partially (by trying to look up URIs).

A few other concepts are needed for the subsequent discussion. For any two documents $d,d' \in \symDocs$ in a \Web\ $\symWoD = \tupleD{\symDocs,\fctsymData,\fctsymADoc}$, document $d$ has a \emph{data link}~to $d'$ if
	%there exists an RDF triple $t = \tupleD{s,p,o}$ that \enumA~appears in the data of $d$, i.e., $t \in \fctData{d}$, and \enumB~mentions a URI $\symURI \in \symAllURIs$ that $\fctsymADoc$ maps to $d'$\!, i.e., $\symURI \in \lbrace s,p,o \rbrace$ and $\fctADoc{\symURI} = d'$.
	the data of $d$ mentions an IRI $\symURI \in \symAllURIs$ (i.e., there exists a triple $\tupleD{s,p,o} \in \fctData{d}$ with $\symURI \in \lbrace s,p,o \rbrace$) that can be used to retrieve $d'$ (i.e., $\fctADoc{\symURI} = d'$).
Such data links establish the \emph{link graph} of the \Web\ $\symWoD$\!, that is, a directed graph $\tupleD{\symDocs,E}$ in which the edges $E$ are all pairs $\tupleD{d,d'} \in \symDocs \times \symDocs$ for which $d$ has a data link to $d'$\!.
Note that this graph, as well as the tuple $\tupleD{\symDocs,\fctsymData,\fctsymADoc}$ typically are not available directly to systems that aim to compute queries over the Web captured by $\symWoD$\!. For instance, the complete domain of the partial mapping $\fctsymADoc$ (i.e., all IRIs that can be used to retrieve some document) is unknown to such systems and can only be disclosed partially (by trying to look up~IRIs).
Also note that the link graph of a \Web\ is a different type of graph than the RDF ``graph'' whose triples are distributed over the documents in the~\Web.
% % why? we need to say ine sentence.

% We are now ready to define query semantics that provide a formal foundation for using PP patterns as queries over a \Web\ (and, thus, over Linked Data on the WWW).

\section{Web-aware Query Semantics for Property Paths}
\label{sec:the-semantics}
% This section introduces three alternative query semantics, each of which defines an expected query result for any PP pattern over any \Web.
% 
% While the first of these semantics is more of theoretical interest, we shall show that the other two can be implemented in practice, using algorithms that are sound and complete.

% This section introduces query semantics that defines an expected query result for any PP pattern over any \Web.

We are now ready to introduce our framework, which does not deal with syntactic aspects of PPs but aims at defining query semantics that provide a formal foundation for using PP patterns as queries over a \Web\ (and, thus, over Linked Data on the WWW).

\subsection{Full-Web Query Semantics} \label{sec:fullweb-sem}
As a first approach we may assume a full-Web query semantics that is based on the standard evaluation function~(as introduced in Section~\ref{subsec:preliminaries}) and defines
	%any query result
	an expected query result for any PP pattern 
in terms of \emph{all data} on the queried \Web. Formally%
%, let $\symPP$ be a PP
% 	%pattern, let $\symWoD = \tupleD{\symDocs,\fctsymData,\fctsymADoc}$ be a \Web, and let $\symRDFgraph^*$ be an RDF graph such that $\symRDFgraph^* = \bigcup_{d \in \symDocs} \fctData{d}$, then the \emph{evaluation} of $\symPP$ over $\symWoD$ under \emph{full-Web semantics}, denoted by $\EvalFull{\symPP}{\symWoD}$, is defined by $\EvalFull{\symPP}{\symWoD} \definedAs \fctOrigQueryPG{\symPP}{\symRDFgraph^*}$.
% 	pattern and let $\symWoD = \tupleD{\symDocs,\fctsymData,\fctsymADoc}$ be a \Web, then the \emph{evaluation} of $\symPP$ over $\symWoD$ under \emph{full-Web semantics}, denoted by $\EvalFull{\symPP}{\symWoD}$, is defined by $\EvalFull{\symPP}{\symWoD} \definedAs \fctOrigQueryPG{\symPP}{\symRDFgraph^*}$, where $\symRDFgraph^*$ is the RDF graph that consists of all triples in $\symWoD$\!, i.e., $\symRDFgraph^* = \bigcup_{d \in \symDocs} \fctData{d}$.
:

\begin{definition} \label{def:Full}
	Let $\symPP$ be a PP pattern, let $\symWoD = \tupleD{\symDocs,\fctsymData,\fctsymADoc}$ be a \Web, and let $\symRDFgraph^*$ be an RDF graph such that $\symRDFgraph^* = \bigcup_{d \in \symDocs} \fctData{d}$, then the \emph{evaluation} of $\symPP$ over $\symWoD$ under \emph{full-Web semantics}, denoted by $\EvalFull{\symPP}{\symWoD}$, is defined by $\EvalFull{\symPP}{\symWoD} \definedAs \fctOrigQueryPG{\symPP}{\symRDFgraph^*}$.
\end{definition}
We emphasize that the full-Web query semantics is mostly of theoretical interest. In practice, that is, for a \Web\ $\symWoD$ that represents the ``real'' WWW (as it runs on the Internet), there cannot exist a system that guarantees to compute the given evaluation function $\EvalFull{\cdot}{\cdot}$ over $\symWoD$ using an algorithm that both terminates and returns complete query results.
	In earlier work, we
	%We
showed such a limitation for evaluating other types of SPARQL graph patterns---including triple patterns---under a corresponding full-Web query semantics defined for these patterns~\cite{Hartig12:TheoryPaper}. 
This result readily carries over to the full-Web query semantics for PP patterns because any PP pattern $\symPP = \tupleD{\symUorVLeft, \symPPE, \symUorVRight}$ with PP expression $\symPPE$ being an IRI $\symURI \in \symAllURIs$
	%is semantically equivalent to
	is, in fact,
a triple pattern $\tupleD{\symUorVLeft, \symURI, \symUorVRight}$. Informally, we explain this negative result by the fact that the three structures $\symDocs$, $\fctsymData$, and $\fctsymADoc$ that capture the queried Web formally, are not available in practice. Consequently, to enumerate the set of all triples on the Web (i.e., the
	RDF graph $\symRDFgraph^*$ in Definition~\ref{def:Full}),
	%aforementioned RDF graph $\symRDFgraph^*$),
a query execution system would have to enumerate all documents (the set $\symDocs$); given that such a system has limited access to mapping $\fctsymADoc$ (in particular, $\fctDom{\fctsymADoc}$---the set of all IRIs whose lookup retrieves a document---is, at best, partially known), the only guarantee to discover all documents is to look up any possible (HTTP-scheme) IRI. Since these are infinitely many~\cite{Fielding99:HTTP}, the enumeration process cannot terminate.

%% ignoring queries that can be computed trivially (such as unsatisfiable queries)
% \input{reach}
%
\subsection{Context-Based Query Semantics}
\label{sec:w3c-ctx-sem}
% % %

Given the limited practical applicability of full-Web query semantics for PPs, we propose an alternative query semantics that interprets
	PP patterns
	%PPs
as a language for navigation over Linked Data on the Web (i.e., along the lines of earlier navigational languages for Linked Data such as NautiLOD~\cite{fionda2012}). We refer to this semantics as \emph{con\-text-based}.

The main idea behind this query semantics is to
	%introduce a notion of a \emph{context} in which each part of a PP expression has to be evaluated
	restrict the scope of
		%searching for triples of paths that match a PP expression to specific data of specific documents on the queried \Web.
		%computing/determining any next step of paths that eventually match a PP expression 
		searching for any next triple of a potentially matching path to specific data within specific documents on the queried \Web.
As a basis for formalizing these restrictions we introduce the notion of a \emph{context selector}. Informally, for each IRI that can be used to retrieve a document, the context selector returns a specific subset of the data within that document; this subset contains only those RDF triples that have the given IRI as their subject~(such a set of triples resembles Harth and Speiser's notion of subject authoritative triples~\cite{harth2012}). Formally, for any \Web\ $\symWoD = \tupleD{\symDocs, \fctsymData,\fctsymADoc}$, the context selector of $\symWoD$ is a function $\fctsymContext{\symWoD\!}\!: \symAllURIs \cup \symAllBNodes \cup \symAllLiterals \cup \symAllVariables \rightarrow 2^{\symAllTriples}$ that, for each $\gamma \in ( \symAllURIs \cup \symAllBNodes \cup \symAllLiterals \cup \symAllVariables )$, is defined as follows:%
\footnote{To simplify the following formalization of con\-text-based semantics, context selectors are defined not only over IRIs, but also over blank nodes, literals, and variables.}
\begin{equation*}
	\fctContext{\symWoD}{\gamma} \definedAs
	\begin{cases}
		\big\lbrace \tupleD{s,p,o} \in \fctsymData\bigl( \fctADoc{\gamma} \bigr) \,\big|\, \gamma = s \big\rbrace &  \text{if $\gamma \in \symAllURIs$ and $\gamma \in \fctDom{\fctsymADoc}$,} \\
		\emptyset & \text{otherwise}.
	\end{cases}
\end{equation*}

\noindent
Informally, we explain how a context selector restricts the scope of PP patterns over a \Web\ as follows. Suppose a sequence of triples $\tupleD{s_1,p_1,o_1}, \,...\, , \tupleD{s_k,p_k,o_k}$ presents a path that already matches a sub-expression of a given PP expression. Under the previously defined
	%reach\-abil\-i\-ty-based query semantics~(cf.~Section~\ref{sec:reach-semantics}), the next triple for such a path can be searched for in any reachable
	full-Web query semantics~(cf.~Section~\ref{sec:fullweb-sem}), the next triple for such a path can be searched for in an arbitrary
document in the queried \Web\ $\symWoD$\!. By contrast, under the con\-text-based query
	%semantics that we define in this section, the next triple has to be searched for only in $\fctContext{\symWoD}{o_k}$. Given these preliminaries, we now define con\-text-based semantics:
	semantics, the next triple has to be searched for only in~$\fctContext{\symWoD}{o_k}$. Given these preliminaries, we now define con\-text-based semantics:
	%semantics that we formalize in Definition~\ref{def:Context}, the next triple has to be searched for only in $\fctContext{\symWoD}{o_k}$.

\begin{definition} \label{def:Context}
	Let $\symPP$ be a PP pattern and let $\symWoD = \tupleD{\symDocs,\fctsymData,\fctsymADoc}$ be a \Web. The \emph{evaluation} of $\symPP$ over $\symWoD$ under \emph{con\-text-based semantics}, denoted by $\EvalCtx{\symPP}{\symWoD}$, returns a multiset of solution mappings $\tupleD{\Omega,\fctsymCard}$ defined recursively as given in Figure~\ref{fig:stc-sem}, where $\symURI,.., \symURI_n \in \symAllURIs$; $\symLeftConst,\symRightConst \in ( \symAllURIs \cup \symAllLiterals )$; $\symLeftVar,\symRightVar \in \symAllVariables$; $\muEmpty$ is the empty solution mapping (i.e., $\fctDom{\muEmpty} = \emptyset$); function $\mathtt{ALPW1}$ is given in Figure~\ref{Figure:WebALP}; and $?v \in \symAllVariables$ is a fresh variable. % (i.e., $?v \notin \lbrace \symLeftVar,\symRightVar \rbrace$).
\end{definition}

\begin{figure}[t]
{%
\scriptsize%
\begin{gather}%
\begin{align*}
%
%
%PREDICATE alpha constant
\EvalCtx{\tupleD{\symLeftURI, p, \symUorVRight}}{\symWoD} & \definedAs
		\Big\langle\,
			\big\lbrace\, \mu \,|\, \fctDom{\mu}=( \lbrace \symUorVRight \rbrace \cap \symAllVariables ) \text{ and } \mu[\tupleD{\symLeftURI, p, \symUorVRight}] \in \fctContext{\symWoD}{\symLeftURI} \,\big\rbrace  \textbf{ , }\,
			\fctsymCardGen{\Omega}
		\,\Big\rangle\\ 
	% % %LITERAL
	\EvalCtx{\tupleD{\symLeftLit, p, \symUorVRight}}{\symWoD} &
	\definedAs
	\Big\langle \emptyset,\fctsymCardGen{\emptyset}
	\,\Big\rangle \\
% % % %PREDICATE alpha variable
	\EvalCtx{\tupleD{\symLeftVar, p, \symUorVRight}}{\symWoD} & \definedAs
		\Big\langle\,
			\big\lbrace\, \mu \,|\, \fctDom{\mu}=( \lbrace \symLeftVar,\symUorVRight \rbrace \cap \symAllVariables ) \text{ and }
	\\[-2mm]
			& \hspace{11.5mm} \mu[\tupleD{\symLeftVar, p, \symUorVRight}] \in \bigcup_{\symURI \in \symAllURIs}\fctContext{\symWoD}{\symURI} \,\big\rbrace  \textbf{ , }\,
			\fctsymCardGen{\Omega}
		\,\Big\rangle
	\\[-1mm]
% %NEGATION alpha constant
	\EvalCtx{\tupleD{\symLeftURI, !(\symURI_1\mid\dots\mid \symURI_n), \symUorVRight}}{\symWoD} & \definedAs
		\Big\langle\,
			\big\lbrace\, \mu \,|\, \fctDom{\mu}=( \lbrace \symUorVRight \rbrace \cap \symAllVariables ) \text{ and } 
			\\[-2mm]
						& \hspace{11.5mm} \exists\ \mu[\tupleD{\symLeftURI, p, \symUorVRight}] \in \fctContext{\symWoD}{\symLeftURI} : p \notin\lbrace \symURI_1 , \ldots , \symURI_n\rbrace \,\big\rbrace  \textbf{ , }\,
			\fctsymCardGen{\Omega}
		\,\Big\rangle \\
	%LIT
\EvalCtx{\tupleD{\symLeftLit, !(\symURI_1\mid\dots\mid \symURI_n), \symUorVRight}}{\symWoD} & \definedAs
\Big\langle \emptyset,\fctsymCardGen{\emptyset}
\,\Big\rangle
\\
% %NEGATION alpha variable
			\EvalCtx{\tupleD{\symLeftVar, !(\symURI_1\mid\dots\mid \symURI_n), \symUorVRight}}{\symWoD} & \definedAs
				\Big\langle\,
					\big\lbrace\, \mu \,|\, \fctDom{\mu}=( \lbrace \symLeftVar,\symUorVRight \rbrace \cap \symAllVariables ) \text{ and }
					\\[-2mm]
								& \hspace{11.5mm} 
					 \exists\ \mu[\tupleD{\symLeftVar, p, \symUorVRight}] \in \bigcup_{\symURI \in \symAllURIs}\fctContext{\symWoD}{\symURI} : p \not\in\lbrace \symURI_1 , \ldots , \symURI_n\rbrace \,\big\rbrace  \textbf{ , }\,
					\fctsymCardGen{\Omega}
				\,\Big\rangle \\
		%INVERSE
	\EvalCtx{\tupleD{\symUorVLeft, \symPPEi, \symUorVRight}}{\symWoD} & \definedAs
		\EvalCtx{\tupleD{\symUorVRight, \symPPE, \symUorVLeft}}{\symWoD}\\ 
%CONCATENATION
		\EvalCtx{\tupleD{\symUorVLeft, \symPPE_1/\symPPE_2, \symUorVRight}}{\symWoD} & \definedAs
			\pi_{ \lbrace \symUorVLeft,\symUorVRight \rbrace \cap \symAllVariables } \Bigl( \EvalCtx{\tupleD{\symUorVLeft, \symPPE_1, ?v}}{\symWoD} \Join \EvalCtx{\tupleD{?v, \symPPE_2, \symUorVRight}}{\symWoD} \Bigr) \\
%%%
%%%
	%UNION
		\EvalCtx{\tupleD{\symUorVLeft,\symPPE_1 \,|\, \symPPE_2,\symUorVRight}}{\symWoD} & \definedAs
			\EvalCtx{\tupleD{\symUorVLeft, \symPPE_1, \symUorVRight}}{\symWoD} \multicup \EvalCtx{\tupleD{\symUorVLeft, \symPPE_2, \symUorVRight}}{\symWoD}
	\\
	% KLEENE STAR
	\EvalCtx{\tupleD{\symLeftConst, (\symPPE)^*, \symRightVar}}{\symWoD} & \definedAs
		\Big\langle\,
			\big\lbrace\, \mu \,|\, \fctDom{\mu} = \lbrace \symRightVar \rbrace \text{ and }
			 \mu(\symRightVar) \!\in\! \mathtt{ALPW1}(\symLeftConst,\symPPE,\symWoD) \big\rbrace  \textbf{ , }\,
			\fctsymCardGen{\Omega}
		\,\Big\rangle
	\\
%%%
%%% An earlier version of the following case contained the additional
%%% constraint that $\mu(\symLeftVar) \in \symAllURIs$. I do not see
%%% why this constraint is necessary. In fact, with this constraint it
%%% would not be possible to have a PP expression  path1 / (path2)*
%%% and a matching path for sub-expression path1 ending up in a literal
%%% or bnode. Consequently, I removed the constraint.
%%% However, I added the following constraint which is similar:
%%%           $\mu(\symLeftVar) \in \fctTerms{\symWoD}$.
%%% This constraint is along the lines of the corresponding constraint
%%% in the standard semantics and it is necessary to avoid infinite
%%% query results.
%%%                                   Olaf, May 6, 2014
%%%
	\EvalCtx{\tupleD{\symLeftVar, (\symPPE)^*, \symRightVar}}{\symWoD} & \definedAs
		\Big\langle\,
			\big\lbrace\, \mu \,|\, \fctDom{\mu} = \lbrace \symLeftVar,\symRightVar \rbrace \text{ and } \mu(\symLeftVar) \in \fctTerms{\symWoD} \text{ and }
			\\[-2mm]
						& \hspace{11.5mm}
			 \mu(\symRightVar) \in \mathtt{ALWP1}(\mu(\symLeftVar),\symPPE,\symWoD) \,\big\rbrace  \textbf{ , }\,
			\fctsymCardGen{\Omega}
		\,\Big\rangle
	\\
	\EvalCtx{\tupleD{\symLeftVar, (\symPPE)^*, \symRightConst}}{\symWoD} & \definedAs
		\EvalCtx{\tupleD{\symRightConst, (\symPPEi)^*, \symLeftVar}}{\symWoD}
	\\
	\EvalCtx{\tupleD{\symLeftConst, (\symPPE)^*, \symRightConst}}{\symWoD} & \definedAs
		\Big\langle\,
			\begin{cases} \lbrace \muEmpty \rbrace & \text{if } \exists\ \mu\in \EvalCtx{\tupleD{\symLeftConst, (\symPPE)^*, ?v}}{\symWoD} : \mu(?v) = \symRightConst , \\ ~~\emptyset & \text{else} \end{cases} \textbf{ , }\,
			\fctsymCardGen{\Omega}
		\,\Big\rangle
\end{align*}%
\end{gather}%
}%
%
% \vspace{-.3cm}
%
\caption{\textbf{Context-based query semantics for SPARQL property paths over the Web.}}
\label{fig:stc-sem}
\end{figure}

\begin{figure}[t]
{\small%
	\subfigure{%
		\begin{minipage}[t]{0.37\textwidth}%
		\underline{\textbf{Function} $\mathtt{ALPW1}\bigl( \gamma, \symPPE, \symWoD \bigr)$}
		\par
		\textbf{Input:} $\gamma \in ( \symAllURIs \cup \symAllBNodes \cup \symAllLiterals )$,
		\par \hspace{9mm} $\symPPE$ is a PP expression,
		\par \hspace{9mm} $\symWoD$ is a \Web.
		\par
		\begin{algorithmic}[1]
			\STATE \textit{Visited} := $\emptyset$
%%%
				\STATE $\mathtt{ALPW2}\bigl( \gamma, \symPPE, \mathit{Visited}, \symWoD \bigr)$ \label{line:WebALP:FirstRecursion}
%%%			\ENDIF
			\RETURN \textit{Visited}
		\end{algorithmic}%
		\end{minipage}%
	}%
	\subfigure{%
		\begin{minipage}[t]{0.63\textwidth}%
		\underline{\textbf{Function} $\mathtt{ALPW2}\bigl( \gamma, \symPPE, \mathit{Visited}, \symWoD \bigr)$}
		\par
		\textbf{Input:} $\gamma \in ( \symAllURIs \cup \symAllBNodes \cup \symAllLiterals )$,
		~ $\symPPE$ is a PP expression,
		\par \hspace{8mm} $\mathit{Visited} \subseteq ( \symAllURIs \cup \symAllBNodes \cup \symAllLiterals )$,
		~ $\symWoD$ is a \Web.
		\par
		\begin{algorithmic}[1]
		\setcounter{ALC@line}{3}
			\IF {$\gamma \notin \mathit{Visited}$}
				\STATE add $\gamma$ to \textit{Visited}
%%%
%%% I replaced the following line (and everything else that is
%%% necessary to have $\symWoD$ available here) because trying
%%% to find the whole path in the context document cannot be
%%% the intention of evaluating Kleene star expressions under
%%% context-based semantics.
%%%                           Olaf, May 16, 2014
%%%
%%%				\FORALL {$\mu \in \fctOrigQueryPG{\tupleD{?x, \symPPE, ?y}}{\symRDFgraph}$ \, s.t. \, $\mu(?x) = \gamma$}
				\FORALL {$\mu \in \EvalCtx{\tupleD{?x, \symPPE, ?y}}{\symWoD}$ \, s.t. \, $\mu(?x) = \gamma$} \label{line:WebALP:ForBegin}
%%%
%%%                           Olaf, May 16, 2014
%%%
					\STATE $\mathtt{ALPW2}\bigl( \mu(?y), \symPPE, \mathit{Visited}, \symWoD \bigr)$ \, \COMMENT{$?x,?y \in \symAllVariables$} \label{line:WebALP:Recursion}
				\ENDFOR \label{line:WebALP:ForEnd}
			\ENDIF
		\end{algorithmic}%
		\end{minipage}%
	}%
% \vspace{-1mm}
\caption{\textbf{Auxiliary functions used for defining context-based query semantics.}}
\label{Figure:WebALP}
}
\end{figure}

% \noindent
% \textit{Discussion.}
% % % % % % %
% There are some points worth to mention in the formalization of the $\symStandSemCtx$ semantics.
% %%% First, for path concatenation, we use the auxiliary function $\fctSigma{\mu}{?v}$:
% %
% %%% <not needed at the moment, moved to the explanation above>
% %%%                                   Olaf, May 6, 2014
% %
% Second, patterns having variables on the left-hand side (e.g., line 3) are not meant to be considered from an implementation point of view. A pattern like $\symPP = \tupleD{\symLeftVar, p, \symUorVRight}$ should consider an \textit{infinite} set of URIs and retrieve their associated documents, which is practically infeasible. We are going to deepen this aspect in Section~\ref{sec:sparql}.
% %
% Third, to deal with arbitrary length paths, we use a new function \texttt{ALP1'}.
% 
% This function instead of recursively calling the evaluation over a fixed graph $G$ (e.g., line 8 $\mathtt{ALP1}$ in Fig.~\ref{Figure:StdALP}) considers the document $\fctData{\symURI}$ (line~12 Fig.~\ref{Figure:WebALP}). Moreover, \texttt{ALP1'} treats blank nodes separately (lines 2-3 Fig.~\ref{Figure:WebALP}); conceptually a blank node in input amounts at instantiating it with each URI. We have also formalized the evaluation property paths over the Web in terms of an \textit{existential} con\-text-based semantics, which is not reported here for lack of space.

\noindent
There are three points worth mentioning w.r.t.~Definition~\ref{def:Context}:
First, note how the context selector restricts the data that has to be searched to find matching triples (e.g., consider the first line in Figure~\ref{fig:stc-sem}).
Second, we emphasize that con\-text-based query semantics is defined such that it resembles the standard semantics of PP patterns as close as possible~(cf.~Section~\ref{subsec:preliminaries}).
Therefore, for the part of our definition that covers PP patterns of the form $\tupleD{\symUorVLeft, \symPPE^*\!, \symUorVRight}$, we also use auxiliary functions%
	%, which we call
	---%
$\mathtt{ALPW1}$ and $\mathtt{ALPW2}$ (cf.~Figure~\ref{Figure:WebALP}).
%%%
%%% Replaced the following sentence after changing the two
%%% functions $\mathtt{ALPW1}$ and $\mathtt{ALPW2}$. For an
%%% explanation of this change, see my comment within these
%%% functions.
%%%	These functions evaluate the sub-expression $\symPPE$ recursively by adjusting the scope of the evaluation via the context selector (cf.~lines \ref{line:WebALP:FirstRecursion} and \ref{line:WebALP:Recursion} in Figure~\ref{Figure:WebALP}). %---instead of using a fixed RDF graph~(as done in the standard semantics, cf.~line~\ref{line:StdALP:Recursion} in Figure~\ref{Figure:StdALP}).
	These functions evaluate the sub-expression $\symPPE$ recursively over the queried \Web%
		%, instead of using a fixed RDF graph as done in the standard semantics (see W3C spec. section 18.4~\cite{Harris2013}).
		~(instead of using a fixed RDF graph as done in the standard semantics in Figure~\ref{Figure:StdALP}).
%%%
%%%                                Olaf, May 16, 2014
%%%
%
%
%
%%%Third, as a result of adopting the standard semantics of inverse PP expressions, PP patterns of the form $\tupleD{\symUorVLeft,\symPPEi,\symUorVRight}$ do \emph{not} mean to navigate the Web backwards (i.e., from $\symUorVLeft$ ``back'' to $\symUorVRight$); instead, such a pattern simply presents an inverted approach to ask for paths from $\symUorVRight$ to $\symUorVLeft$.
%
%
Third, the two base cases with a variable in the subject position (i.e., the third and the sixth line in Figure~\ref{fig:stc-sem}) require an enumeration of all IRIs. Such a requirement is necessary to preserve consistency with the standard semantics, as well as to preserve commutativity of operators that can be defined on top of PP patterns (such as the \!\OpAND operator in SPARQL; cf.~Section~\ref{sec:sparql}).
However, due to this requirement there exist PP patterns whose~(complete) evaluation under con\-text-based semantics is infeasible when querying the WWW. The following example describes such a case.

\begin{exmp}\label{ex:unsafe-pattern}\hspace{-1ex}\textbf{.}
 	Consider the PP pattern $\symPattern_\mathsf{E\ref{ex:unsafe-pattern}} = \tupleD{?v,\constant{knows},\constant{Tim}}$, which asks for the IRIs of people that know Tim. Under con\-text-based semantics, any IRI~$\symURI'$ can be used to generate a correct solution mapping for the pattern as long as a lookup of that IRI
		%(in the queried \Web)
	results in retrieving a document whose data includes the triple $\tupleD{\symURI'\!,\constant{knows},\constant{Tim}}$. While, for any \Web\ that is finite, there exists only a finite number of such IRIs, determining these IRIs and guaranteeing completeness requires to enumerate the infinite set of all IRIs and to check each of them (unless one knows the complete---and finite---subset of all IRIs that can be used to retrieve some document, which, due to the infiniteness of possible HTTP IRIs, cannot be achieved for the WWW).
\end{exmp}
 
\noindent
It is not difficult to see that the issue illustrated in the example exists for any triple pattern
	that has
	%with
a variable in the subject position. On the other hand, triple patterns whose subject is an IRI do not have this issue. However, having an IRI in the subject position is not a sufficient condition in general. % for a feasible evaluation.  %%% Comment b/c it reads awkward. Olaf
For instance, the PP pattern $\tupleD{\constant{Tim},\invPPE{\constant{knows}},?v}$ has the same issue as the pattern in Example~\ref{ex:unsafe-pattern}~(in fact, both patterns are semantically equivalent under con\-text-based semantics).
A question that arises is whether there exists a
	%(decidable)
	%syntactic
property of PP patterns that can be used to distinguish between patterns that do not have this issue (i.e., evaluating them over any \Web\ is feasible) and those that do. We shall discuss this question
	%in the following section
	%in Section~\ref{sec:sparql}
for the more general case of PP-based SPARQL~queries.
% \input{ctxt_exist}
%\input{discussion}
% \input{comparison}
% % % % %
\section{SPARQL with Property Paths on the Web}
\label{sec:sparql}

After considering PP patterns in separation, we now turn to a more expressive fragment of SPARQL that embeds PP patterns as the basic building block and uses additional operators on top.
	%In this section, we
	We
define the resulting PP-based SPARQL queries, discuss the feasibility of evaluating these queries over the Web, and introduce a syntactic property to identify queries for which an evaluation under con\-text-based semantics is feasible.

\subsection{Definition} \label{subsec:DefinitionSPARQL}
By
	%adopting
	using
the algebraic
	%syntax of SPARQL as introduced by P\'erez et al.~\cite{perez2009},
	syntax of SPARQL~\cite{perez2009},
	%SPARQL syntax~\cite{perez2009},
we define a \emph{graph pattern}
	recursively as follows:
	%as follows:
% \begin{itemize}
% 	\item
% 		Any PP pattern $\tupleD{\symUorVLeft,\symPPE,\symUorVRight}$ is a graph pattern.
% 	\item
% 		If $\symPattern_1$ and $\symPattern_2$ are graph patterns, then $(\symPattern_1 \OpAND \symPattern_2)$, $(\symPattern_1 \OpUNION \symPattern_2)$, and $(\symPattern_1 \OpOPT \symPattern_2)$ are graph patterns.
% \end{itemize}
%
\enumA~Any PP pattern $\tupleD{\symUorVLeft,\symPPE,\symUorVRight}$ is a graph pattern; and \enumB~if $\symPattern_1$ and $\symPattern_2$ are graph patterns, then $(\symPattern_1 \OpAND \symPattern_2)$, $(\symPattern_1 \OpUNION \symPattern_2)$, and $(\symPattern_1 \OpOPT \symPattern_2)$ are graph patterns.%
\footnote{For this paper we leave out other types of SPARQL graph
	%patterns such as filter conditions or query federation.
	%patterns such as filter conditions.
	patterns such as filters.
Adding them is an exercise that would not have any significant implication on the following discussion.}
For any graph pattern $\symPattern$, we write $\fctVars{\symPattern}$ to denote the set of \emph{all variables} in $\symPattern$.

By using PP patterns as the basic building block of graph patterns, we can readily carry over our
	%notion of a
con\-text-based
	%query
semantics to graph patterns: For any graph pattern $\symPattern$ and any \Web\ $\symWoD$\!, the \emph{evaluation} of $\symPattern$ over $\symWoD$ under con\-text-based semantics is a multiset of solution mappings, denoted by $\EvalCtx{\symPattern}{\symWoD}$\!, that is defined recursively as~follows:\footnote{Note that the definition uses the algebra operators introduced in Section~\ref{subsec:preliminaries}.}

\begin{itemize}
	\item
		If $\symPattern$ is a PP pattern, then $\EvalCtx{\symPattern}{\symWoD}$ is defined in Definition~\ref{def:Context}.
	\item
		If $\symPattern$ is $(\symPattern_1 \OpAND \symPattern_2)$, then $\EvalCtx{\symPattern}{\symWoD} \definedAs \EvalCtx{\symPattern_1}{\symWoD} \Join \EvalCtx{\symPattern_2}{\symWoD}$.
	\item
		If $\symPattern$ is $(\symPattern_1 \OpUNION \symPattern_2)$, then $\EvalCtx{\symPattern}{\symWoD} \definedAs \EvalCtx{\symPattern_1}{\symWoD} \multicup \EvalCtx{\symPattern_2}{\symWoD}$.
	\item
		If $\symPattern$ is $(\symPattern_1 \OpOPT \symPattern_2)$, then $\EvalCtx{\symPattern}{\symWoD} \definedAs \bigl( \EvalCtx{\symPattern_1}{\symWoD} \Join \EvalCtx{\symPattern_2}{\symWoD} \bigr) \multicup \bigl( \EvalCtx{\symPattern_1}{\symWoD} \setminus \EvalCtx{\symPattern_2}{\symWoD} \bigr)$.
\end{itemize}

% %
% % % % % % % % % % % % % % %
% As it can be observed, the evaluation presented above can be parametrized according to the different semantics introduced in Section~\ref{sec:the-semantics}. 
%Consider the following example.
%As stated in the W3C specification, if patterns of the form $\symPPT$ do not contain the closure operator (* and +) they can be translated into equivalent patterns with the introduction of variables~\cite{}. Hence, such kind of rewritten patterns have alrready been given a sematics for the case $\symFullWebSem,$ and $\symReachSem(\cdot,\!\cdot$).
% % %
% % %
%\noindent
%We denote by $\fctVars{\symPattern}$ the set of variables in $\symPattern$.
\subsection{Discussion}
Given
	a
query semantics for evaluating PP-based graph patterns over a \Web, we now discuss the feasibility of such evaluation.
%
% 	%Unsurprisingly,
% 	First,
% for reach\-abil\-i\-ty-based semantics, the positive result in Theorem~\ref{thm:Reach} carries over directly to the more expressive fragment:
% 
% \begin{corollary}
% 	Let $\symPattern$ be a graph pattern, let $\symSeedURIs \subseteq \symAllURIs$ be a finite set of URIs, and let $\symReachCrit$ be a reachability criterion.
% 	There exists an algorithm that, for any finite \Web\ $\symWoD$\!, computes $\EvalReach{\symPattern}{\symReachCrit}{\symSeedURIs}{\symWoD}$ by looking up a finite number of URIs only.
% \end{corollary}
% 
% \noindent
% Now we come back to our discussion of con\-text-based semantics, for which we
% 
	To this end, we
introduce the notion of \emph{Web-safeness} of graph patterns. Informally, graph patterns are Web-safe if evaluating them completely under con\-text-based semantics is possible. Formally:

\begin{definition}
	A graph pattern $\symPattern$ is \emph{Web-safe} if there exists an algorithm that, for any finite \Web\ $\symWoD=\tupleD{\symDocs,\fctsymData,\fctsymADoc}$, computes $\EvalCtx{\symPattern}{\symWoD}$ by looking up only a finite number of IRIs without assuming direct access to the sets $\symDocs$ and $\fctDom{\fctsymADoc}$.
\end{definition}

\begin{exmp} \label{ex:safeness}\hspace{-1ex}\textbf{.}
	Consider graph pattern $\symPattern_\mathsf{E\ref{ex:safeness}} = \bigl( \tupleD{\constant{Bob},\constant{knows},?v} \OpAND \tupleD{?v,\constant{knows},\constant{Tim}} \bigr)$. The right sub-pat\-tern $\symPattern_\mathsf{E\ref{ex:unsafe-pattern}} = \tupleD{?v,\constant{knows},\constant{Tim}}$ is \emph{not} Web-safe because evaluating it completely over the WWW is not feasible under con\-text-based semantics (cf.~Example~\ref{ex:unsafe-pattern}). However, the larger pattern $\symPattern_\mathsf{E\ref{ex:safeness}}$ is Web-safe; it can be evaluated completely~under con\-text-based semantics. 
	For instance, a possible algorithm may first evaluate the left
		%sub-pat\-tern $\tupleD{\constant{Bob},\constant{knows},?v}$,
		sub-pat\-tern,
	which is feasible because it requires the lookup of a single IRI only~(the IRI $\constant{Bob}$). Thereafter, the evaluation of the right sub-pat\-tern $\symPattern_\mathsf{E\ref{ex:unsafe-pattern}}$ can be reduced to looking up a finite number of IRIs only, namely the IRIs bound to variable $?v$ in solution mappings obtained
		%in the first step
	for the left sub-pat\-tern. Although any other IRI $\symURI^*$ might also be used to discover matching triples for $\symPattern_\mathsf{E\ref{ex:unsafe-pattern}}$, each of these triples has IRI $\symURI^*$
		%in its subject position%
		as its subject%
	~(which is a consequence of restricting retrieved data based on the context selector introduced in Section~\ref{sec:w3c-ctx-sem}). Therefore, the solution mappings resulting from such matching triples cannot be compatible with any solution for the left sub-pat\-tern and, thus, do not satisfy the join condition established by the semantics of\OpAND in pattern~$\symPattern_\mathsf{E\ref{ex:safeness}}$.
\end{exmp}

\noindent
The example illustrates that some graph patterns are Web-safe even if some of their sub-pat\-terns are not. Consequently, we are interested in a \emph{decidable} property that enables to identify Web-safe patterns, including those whose sub-patterns are not Web-safe.
%
	%By recalling Example~\ref{ex:safeness}, we note that the usage of variables in a complex pattern can (positively) affect the feasibility of its evaluation over the Web and, thus, its Web-safeness.

Buil-Aranda et al.~study a similar problem in the context of SPARQL federation where graph patterns of the form $\symPattern_S = \bigl(\!\OpSERVICE ?v \, \symPattern \bigr)$ are allowed~\cite{buil2012}. % (with $?v \in \symAllVariables$ and $\symPattern$ being another arbitrary graph pattern).
Here, variable~$?v$ ranges over a possibly large set of IRIs, each of which represents the address of a~(remote) SPARQL service that needs to be called to assemble the complete result of~$\symPattern_S$. However, many service calls may be avoided if~$\symPattern_S$ is embedded in a larger graph pattern that allows for an evaluation during which~$?v$ can be bound before evaluating~$\symPattern_S$. To tackle this problem, Buil-Aranda et al.~introduce a notion of
	\textit{strong boundedness} of variables in graph patterns
	%\textit{strongly bound} variables
and use it to show a notion of safeness for the evaluation of patterns like $\symPattern_S$ within larger graph patterns. 
% % % %
%For any graph pattern $\symPattern$, we write $\fctVars{\symPattern}$ to denote the set of all variables in $\symPattern$.
%
The set of \emph{strongly bound variables} in a graph pattern $\symPattern$, denoted by $\fctCVars{\symPattern}$, is defined recursively as follows:
\begin{itemize}
	\item
		If $\symPattern$ is a PP pattern, then $\fctCVars{\symPattern}=\fctVars{\symPattern}$ (recall that $\fctVars{\symPattern}$ are all variables in $\symPattern$).
	\item
		If $\symPattern$ is of the form $(\symPattern_1 \OpAND \symPattern_2)$, then $\fctCVars{\symPattern}=\fctCVars{\symPattern_1} \cup \fctCVars{\symPattern_2}$.
	\item
		If $\symPattern$ is of the form $(\symPattern_1 \OpUNION \symPattern_2)$, then $\fctCVars{\symPattern} = \fctCVars{\symPattern_1} \cap \fctCVars{\symPattern_2}$.
	\item
		If $\symPattern$ is of the form $(\symPattern_1 \OpOPT \symPattern_2)$, then $\fctCVars{\symPattern} = \fctCVars{\symPattern_1}$.
\end{itemize}

\noindent
The idea behind the notion of strongly bound variables has already been used in earlier work (e.g., \textit{``certain variables''}~\cite{Schmidt10:FoundationsOfSPARQLOptimization}, \textit{``output variables''}~\cite{TomanAndWeddellBook}), and it is tempting to adopt it for our problem. However, we note that one cannot identify Web-safe graph patterns by using strong boundedness in a manner similar to its use in Buil-Aranda et al.'s work alone.
For instance, consider graph pattern
	%$\symPattern_\mathsf{E\ref{ex:safeness}} = \bigl( \tupleD{\constant{Bob},\constant{knows},?v} \OpAND \tupleD{?v,\constant{knows},\constant{Tim}} \bigr)$
	$\symPattern_\mathsf{E\ref{ex:safeness}}$
from Example~\ref{ex:safeness}.
	We know that \enumA~$\symPattern_\mathsf{E\ref{ex:safeness}}$ is Web-safe and that \enumB~$\fctVars{\symPattern_\mathsf{E\ref{ex:safeness}}} = \lbrace ?v \rbrace$ and also $\fctCVars{\symPattern_\mathsf{E\ref{ex:safeness}}} = \lbrace ?v \rbrace$. Then, one
	%Given that $\fctVars{\symPattern_\mathsf{E\ref{ex:safeness}}} = \lbrace ?v \rbrace$, $\fctCVars{\symPattern_\mathsf{E\ref{ex:safeness}}} = \lbrace ?v \rbrace$, and we know that $\symPattern_\mathsf{E\ref{ex:safeness}}$ is Web-safe, we
might hypothesize that for every graph pattern $\symPattern$,
	%$\fctCVars{\symPattern}=\fctVars{\symPattern}$ implies Web-safeness.
	if $\fctCVars{\symPattern}=\fctVars{\symPattern}$, then $\symPattern$ is Web-safe.
However, the PP pattern $\symPattern_\mathsf{E\ref{ex:unsafe-pattern}} = \tupleD{?v,\mathrm{knows},\mathrm{Tim}}$ disproves such a hypothesis because, even if $\fctCVars{\symPattern_\mathsf{E\ref{ex:unsafe-pattern}}}=\fctVars{\symPattern_\mathsf{E\ref{ex:unsafe-pattern}}}$, pattern $\symPattern_\mathsf{E\ref{ex:unsafe-pattern}}$ is not Web-safe~(cf.~Example~\ref{ex:unsafe-pattern}).
%
% Alternatively, one might also hypothesize that if a graph pattern $\symPattern$ is Web-safe, then $\fctCVars{\symPattern}=\fctVars{\symPattern}$. However, this hypothesis can be disproved by using pattern $\symPattern = \bigl( \tupleD{\constant{Bob},\constant{knows},?x} \OpOPT \tupleD{?x,\constant{knows},?y} \bigr)$. It can easily be verified that $\symPattern$ is Web-safe~(e.g., it is not difficult to adjust the algorithm for pattern $\symPattern_\mathsf{E\ref{ex:safeness}}$ in Example~\ref{ex:safeness} accordingly). However, in contradiction to the hypothesis we have $\fctCVars{\symPattern} \neq \fctVars{\symPattern}$.

We conjecture the following reason why strong boundedness cannot be used directly for our problem. For complex patterns (i.e., patterns that are not PP patterns), the sets of strongly bound variables of all sub-patterns are defined \emph{independent} from each other, whereas the algorithm outlined in Example~\ref{ex:safeness} leverages a specific relationship between sub-patterns. More precisely, the algorithm leverages the fact that the same variable that is the subject of the right sub-pattern is also the object of the left sub-pattern. 

Based on this observation, we introduce the notion of \emph{conditionally Web-bounded variables}, the definition of which, for complex graph patterns, is based on specific relationships between sub-patterns. This notion shall turn out to be suitable for our~case.

\begin{definition} %(Conditionally Web-safe variables).
\label{def:cond-web-safeness}
The \emph{conditionally Web-bounded variables} of a graph pattern $\symPattern$ w.r.t.~a set of variables
	%$\symWebSafeVar \subseteq \symAllVariables$
	$\symWebSafeVar$
is the subset $\fctCondWebSafe{\symPattern}{\symWebSafeVar} \subseteq \fctVars{\symPattern}$
	%of variables in $\symPattern$
that is defined recursively as follows:
\begin{center}%
\scriptsize
\begin{tabular}{|rp{85mm}|c|}
	\hline
	& \hspace{30mm} If $\symPattern$  is: %$\symPattern$ (including condition)
	& then $\fctCondWebSafe{\symPattern}{\symWebSafeVar}$ is:
	\\ \hline
	1) & $\tupleD{\symUorVLeft,\symURI,\symUorVRight}$
	\, or \, $\tupleD{\symUorVLeft, !( \symURI_1 \,|\, ... \,|\, \symURI_n) ,\symUorVRight}$ \, such that  $\symUorVLeft \in (\symAllURIs \cup \symAllLiterals)$ or $\symUorVLeft \in \symWebSafeVar$
	& $\fctVars{\symPattern}$
	\\
	2) & $\tupleD{\symUorVLeft,\symURI,\symUorVRight}$
	\, or \, $\tupleD{\symUorVLeft, !( \symURI_1 \,|\, ... \,|\, \symURI_n) ,\symUorVRight}$ \, such that $\symUorVLeft \notin (\symAllURIs \cup \symAllLiterals)$ and $\symUorVLeft \notin \symWebSafeVar$
	& $\emptyset$
	\\[1mm]
3) & $\tupleD{\symUorVLeft,(\symPPE)^*\!,\symUorVRight}$
s.t.~$\symUorVLeft \in \symAllVariables$ and $\symUorVRight \notin \symAllVariables$
& $\fctsymWebSafe\bigl( \tupleD{\symUorVRight,(\symPPEi)^*\!,\symUorVLeft} \,|\, \symWebSafeVar \bigr)$
\\
4) & $\tupleD{\symUorVLeft,(\symPPE)^*\!,\symUorVRight}$
s.t.~\enumA~$\symUorVLeft \notin \symAllVariables$ or $\symUorVRight \in \symAllVariables$, and \enumB~for any two variables $?x,?y \in \symAllVariables$
& $\fctsymWebSafe\bigl( \tupleD{\symUorVLeft,\symPPE,\symUorVRight} \,|\, \symWebSafeVar \bigr)$
\\
& \hspace{27.5mm} it holds that $\fctsymWebSafe\bigl( \tupleD{?x,\symPPE,?y} \,|\, \lbrace ?x \rbrace \bigr) = \lbrace ?x,?y \rbrace$ &
\\
5) & $\tupleD{\symUorVLeft,(\symPPE)^*\!,\symUorVRight}$
such that none of the above
& $\emptyset$
\\[1mm]
	6) & $\tupleD{\symUorVLeft,\symPPEi,\symUorVRight}$
	with $\symPattern' = \tupleD{\symUorVRight,\symPPE,\symUorVLeft}$
	& $\fctCondWebSafe{\symPattern'}{\symWebSafeVar}$
	\\
% 	4) & $\tupleD{\symUorVLeft,(\symPPE)^*\!,\symUorVRight}$
% 	with $\symPattern' = \tupleD{\symUorVLeft,\symPPE,\symUorVRight}$
% 	& $\fctCondWebSafe{\symPattern'}{\symWebSafeVar}$
% 	\\
	7) & $\tupleD{\symUorVLeft,(\symPPE_1|\symPPE_2),\symUorVRight}$
	with $\symPattern' = \bigl( \tupleD{\symUorVLeft,\symPPE_1,\symUorVRight} \OpUNION \tupleD{\symUorVLeft,\symPPE_2,\symUorVRight} \bigr)$
	& $\fctCondWebSafe{\symPattern'}{\symWebSafeVar}$
	\\[1mm]
	8) & $\tupleD{\symUorVLeft,\symPPE_1/\symPPE_2,\symUorVRight}$
	s.t., for any $?v \in \symAllVariables \setminus \bigl( \symWebSafeVar \cup \lbrace \symUorVLeft,\symUorVRight \rbrace \bigr)$, $?v \in \fctCondWebSafe{\symPattern'}{\symWebSafeVar}$
	& $\fctCondWebSafe{\symPattern'}{\symWebSafeVar} \setminus \lbrace ?v \rbrace$
	\\
	& \hspace{28mm} where $\symPattern' = \bigl( \tupleD{\symUorVLeft,\symPPE_1,?v} \OpAND \tupleD{?v,\symPPE_2,\symUorVRight} \bigr)$ &
	\\
	9) & $\tupleD{\symUorVLeft,\symPPE_1/\symPPE_2,\symUorVRight}$
	such that none of the above
	& $\emptyset$
	\\[1mm]
	10) & $(\symPattern_1 \OpAND \symPattern_2)$
	s.t.~$\fctCondWebSafe{\symPattern_1}{\symWebSafeVar} = \fctVars{\symPattern_1}$ and $\fctCondWebSafe{\symPattern_2}{\symWebSafeVar} = \fctVars{\symPattern_2}$
	& $\fctVars{\symPattern}$
	\\
	11) & $(\symPattern_1 \OpAND \symPattern_2)$
	s.t.~$\fctCondWebSafe{\symPattern_1}{\symWebSafeVar} = \fctVars{\symPattern_1}$ and $\fctCondWebSafe{\symPattern_2}{\symWebSafeVar \cup \fctCVars{\symPattern_1}} = \fctVars{\symPattern_2}$
	& $\fctVars{\symPattern}$
	\\
	12) & $(\symPattern_1 \OpAND \symPattern_2)$
	s.t.~$\fctCondWebSafe{\symPattern_2}{\symWebSafeVar} = \fctVars{\symPattern_2}$ and $\fctCondWebSafe{\symPattern_1}{\symWebSafeVar \cup \fctCVars{\symPattern_2}} = \fctVars{\symPattern_1}$
	& $\fctVars{\symPattern}$
	\\
	13) & $(\symPattern_1 \OpAND \symPattern_2)$
	such that none of the above
	& $\emptyset$
	\\[1mm]
	14) & $(\symPattern_1 \OpUNION \symPattern_2)$
	& $\fctCondWebSafe{\symPattern_1}{\symWebSafeVar} \!\cap\! \fctCondWebSafe{\symPattern_2}{\symWebSafeVar}$
	\\[1mm]
	15) & $(\symPattern_1 \OpOPT \symPattern_2)$
	s.t.~$\fctCondWebSafe{\symPattern_1}{\symWebSafeVar} = \fctVars{\symPattern_1}$ and $\fctCondWebSafe{\symPattern_2}{\symWebSafeVar} = \fctVars{\symPattern_2}$
	& $\fctVars{\symPattern}$
	\\
	16) & $(\symPattern_1 \OpOPT \symPattern_2)$
	s.t.~$\fctCondWebSafe{\symPattern_1}{\symWebSafeVar} = \fctVars{\symPattern_1}$ and $\fctCondWebSafe{\symPattern_2}{\symWebSafeVar \cup \fctCVars{\symPattern_1}} = \fctVars{\symPattern_2}$
	& $\fctVars{\symPattern}$
	\\
	17) & $(\symPattern_1 \OpOPT \symPattern_2)$
	such that none of the above
	& $\emptyset$
	\\ \hline
\end{tabular}
\end{center}
\end{definition}
% % % %
%\end{landscape}

\smallskip
\begin{exmp}\label{ex:CondWebBoundedVars}\hspace{-1ex}\textbf{.}
	For the PP pattern $\symPattern_\mathsf{E\ref{ex:unsafe-pattern}} = \tupleD{?v,\constant{knows},\constant{Tim}}$---which is \emph{not} Web-safe~(as discussed in Example~\ref{ex:unsafe-pattern})---if we use the set $\lbrace ?v \rbrace$ as condition, then, by line~1 in Definition~\ref{def:cond-web-safeness}, it holds that $\fctsymWebSafe\bigl( \symPattern_\mathsf{E\ref{ex:unsafe-pattern}} \,\big|\, \lbrace ?v \rbrace \bigr) = \lbrace ?v \rbrace$. However, if we use the empty set instead, we obtain $\fctCondWebSafe{ \symPattern_\mathsf{E\ref{ex:unsafe-pattern}} }{\emptyset} = \emptyset$~(cf.~line~2 in Definition~\ref{def:cond-web-safeness}).

	While for the non-Web-safe pattern $\symPattern_\mathsf{E\ref{ex:unsafe-pattern}}$ we thus observe $\fctCondWebSafe{ \symPattern_\mathsf{E\ref{ex:unsafe-pattern}} }{\emptyset} \neq \fctVars{\symPattern_\mathsf{E\ref{ex:unsafe-pattern}}}$, for graph pattern $\symPattern_\mathsf{E\ref{ex:safeness}} \!=\! \bigl( \tupleD{\constant{Bob},\constant{knows},?v} \OpAND \tupleD{?v,\constant{knows},\constant{Tim}} \bigr)$---which is Web-safe~(cf. Example~\ref{ex:safeness})---we have $\fctCondWebSafe{ \symPattern_\mathsf{E\ref{ex:safeness}} }{ \emptyset } = \fctVars{ \symPattern_\mathsf{E\ref{ex:safeness}} }$. The fact that $\fctCondWebSafe{ \symPattern_\mathsf{E\ref{ex:safeness}} }{ \emptyset } = \lbrace ?v \rbrace$ follows from \enumA~$\fctsymWebSafe\bigl( \tupleD{\constant{Bob},\constant{knows},?v} \,\big|\, \emptyset \bigr) = \lbrace ?v \rbrace$, \enumB~$\fctCVars{ \tupleD{\constant{Bob},\constant{knows},?v} } = \lbrace ?v \rbrace$, \enumC~$\fctsymWebSafe\bigl( \tupleD{?v,\constant{knows},\constant{Tim}} \,\big|\, \lbrace ?v \rbrace \bigr) = \lbrace ?v \rbrace$, and \enumD~line~11 in Definition~\ref{def:cond-web-safeness}.

% 	Consider another PP pattern $\symPattern' = \tupleD{?x,\constant{knows},?y}$. It is easy to see that $\symPattern'$ is \emph{not} Web-safe, and we have $\fctsymWebSafe\bigl( \symPattern' \,\big|\, \lbrace ?x \rbrace \bigr) \!=\! \fctVars{\symPattern'}$ and $\fctCondWebSafe{\symPattern'}{\emptyset} \!\neq\! \fctVars{\symPattern'}$.
%
% 	Note also that $\symPattern'$ is the right sub-pattern of $\symPattern = \bigl( \tupleD{\constant{Bob},\constant{knows},?x} \OpOPT \tupleD{?x,\constant{knows},?y} \bigr)$, which is the other Web-safe graph pattern mentioned above. For $\symPattern$ we have $\fctCondWebSafe{\symPattern}{\emptyset} = \fctVars{\symPattern}$.
\end{exmp}

\noindent
The example seems to suggest that, if \emph{all} variables of a graph pattern are conditionally Web-bounded w.r.t.~the empty set of variables, then the graph pattern is Web-safe. The following result verifies this hypothesis.
\begin{theorem}\label{thm:MainResult}
	A graph pattern $\symPattern$ is Web-safe if $\fctCondWebSafe{\symPattern}{\emptyset} = \fctVars{\symPattern}$.
\end{theorem}

\begin{Note}\hspace{-1ex}\textbf{.}
	Due to the recursive nature of Definition~\ref{def:cond-web-safeness}, the condition
		$\fctCondWebSafe{\symPattern}{\emptyset} \!=\! \fctVars{\symPattern}$~(as used in Theorem~\ref{thm:MainResult})
		%used in Theorem~\ref{thm:MainResult}---i.e., whether $\fctCondWebSafe{\symPattern}{\emptyset} = \fctVars{\symPattern}$---%
	is decidable for any graph pattern $\symPattern$.
\end{Note}

\noindent
We prove Theorem~\ref{thm:MainResult} based on an algorithm that evaluates graph patterns recursively by passing (intermediate) solution mappings to recursive calls. To capture the desired results of each recursive call formally, we introduce a special evaluation function for a graph pattern $\symPattern$ over a \Web\ $\symWoD$ that takes a solution mapping $\mu$ as input and returns only the solutions for $\symPattern$ over $\symWoD$ that are compatible with $\mu$.

\begin{definition} \label{def:ContextWithInput}
	Let $\symPattern$ be a graph pattern, let $\symWoD$ be a \Web, and let $\tupleD{\Omega,\fctsymCard} = \EvalCtx{\symPattern}{\symWoD}$. Given a solution mapping $\mu$, the \emph{$\mu$-restricted evaluation} of $\symPattern$ over $\symWoD$ under \emph{con\-text-based semantics}, denoted by $\EvalCtx{\symPattern \,|\, \mu\,}{\symWoD}$, is the multiset of solution mappings $\tupleD{\Omega'\!,\fctsymCard'}$ with $\Omega' = \big\lbrace \mu' \in \Omega \,\big|\, \mu' \compatible \mu \big\rbrace$ and
		%$\fctsymCard'$ is the restriction of $\fctsymCard$ to~$\Omega'$ (i.e., $\fctsymCard'(\mu') = \fctCard{\mu'}$ for all $\mu' \in \Omega'$).
		$\fctsymCard'(\mu') = \fctCard{\mu'}$ for all $\mu'\! \in \Omega'$.
\end{definition}

\noindent
The following lemma shows the existence of the aforementioned recursive algorithm.

\begin{lemma}\label{lem:Proxy}
	Let $\symPattern$ be a graph pattern and let $\mu_\mathsf{in}$ be a solution mapping.
	If
		it holds that
	$\fctsymWebSafe\bigl( \symPattern \,\big|\, \fctDom{\mu_\mathsf{in}} \bigr) = \fctVars{\symPattern}$, there exists an algorithm that, for any finite \Web\ $\symWoD$\!, computes $\EvalCtx{\symPattern \,|\, \mu_\mathsf{in}\,}{\symWoD}$ by looking up a finite number of IRIs only.
\end{lemma}

\noindent
Before providing the proof of the lemma (and of Theorem~\ref{thm:MainResult}),
	%prove the lemma (and Theorem~\ref{thm:MainResult}),
we point out two important properties of Definition~\ref{def:ContextWithInput}.
First, it is easily seen that, for any graph pattern~$\symPattern$ and \Web\ $\symWoD$\!, $\EvalCtx{\symPattern \,|\, \muEmpty\,}{\symWoD} = \EvalCtx{\symPattern}{\symWoD}$, where $\muEmpty$ is the empty solution mapping (i.e., $\fctDom{\muEmpty} = \emptyset$). Consequently, given an algorithm, say $A$, that has the properties of the algorithm described by Lemma~\ref{lem:Proxy}, a trivial algorithm that can be used to prove Theorem~\ref{thm:MainResult} may simply call algorithm $A$ with the empty solution mapping and return the result of this call (%
	we shall elaborate more on this approach in the proof of
	%a more detailed discussion of this approach follows in the proof of
Theorem~\ref{thm:MainResult} below).
Second, for any PP pattern $\tupleD{\symUorVLeft, \symPPE, \symUorVRight}$ and \Web\ $\symWoD$\!, if $\symUorVLeft$ is a variable and $\symPPE$ is a
	%PP expression that corresponds to one of the first two cases in the grammar in Section~\ref{subsec:preliminaries} (i.e., the two base cases),
	base PP expression (i.e., one of the first two cases in the grammar in Section~\ref{subsec:preliminaries}),
then $\EvalCtx{\symPattern \,|\, \mu\,}{\symWoD}$ is empty for every solution mapping $\mu$ that binds (variable) $\symUorVLeft$ to a literal or a blank node. Formally, we show the latter as follows.

\begin{lemma} \label{lem:ContextWithInputBaseCases}
	Let $\symPattern$ be a PP pattern of the form $\tupleD{?v, \symURI, \symUorVRight}$ or $\tupleD{?v, !(\symURI_1\mid\dots\mid \symURI_n) , \symUorVRight}$ with $?v \in \symAllVariables$ and $\symURI, \symURI_1, \ldots, \symURI_n \in \symAllURIs$, and let $\mu$ be a solution mapping.
	If $?v \in \fctDom{\mu}$ and $\mu(?v) \in ( \symAllBNodes \cup \symAllLiterals )$, then, for any \Web\ $\symWoD$\!, $\EvalCtx{\symPattern \,|\, \mu\,}{\symWoD}$ is the empty multiset. % (of solution mappings).
\end{lemma}

\begin{proof}[Lemma~\ref{lem:ContextWithInputBaseCases}]
	Recall that, for any IRI
		%$\symURI \in \symAllURIs$
		$\symURI$
	and any \Web~$\symWoD$\!, context $\fctContext{\symWoD}{\symURI}$ contains only triples that have IRI~$\symURI$ as their subject. As a consequence, for any \Web~$\symWoD$\!,
		%by Definition~\ref{def:Context},
	every solution mapping $\mu' \in \EvalCtx{\symPattern}{\symWoD}$ binds variable~$?v$ to some IRI (and never to a literal or blank node); i.e., $\mu'(?v) \in \symAllURIs$. Therefore, if $?v \in \fctDom{\mu}$ and $\mu(?v) \in ( \symAllBNodes \cup \symAllLiterals )$,
		%solution mappings~$\mu'$ and $\mu$ cannot be compatible
		then $\mu$ cannot be compatible with any $\mu' \in \EvalCtx{\symPattern}{\symWoD}$
	and, thus, $\EvalCtx{\symPattern \,|\, \mu\,}{\symWoD}$ is empty.
	\qed
\end{proof}

\noindent
We use Lemma~\ref{lem:ContextWithInputBaseCases} to prove Lemma~\ref{lem:Proxy} as follows.

\begin{proofidea}{Lemma~\ref{lem:Proxy}}
	We prove the lemma by induction on the possible structure of graph pattern $\symPattern$.
	%The induction has two base cases, namely, either $\symPattern$ is a PP pattern of the form $\tupleD{\symUorVLeft, \symURI, \symUorVRight}$ or $\tupleD{\symUorVLeft, !(\symURI_1\mid\dots\mid \symURI_n) , \symUorVRight}$, and the induction step covers any other type of graph pattern, including all other types of PP patterns.
%
% 	Due to space limitations, in this paper we only discuss
% 		%the first of both base cases
% 		one base case
% 	and one pivotal case of the induction step. For a complete discussion of all cases refer to
% 		\PaperVersion{the extended version of this paper~\cite{ExtendedVersion}.}%
% 		\ExtendedVersion{the Appendix (cf.~Section~\ref{proofsec:lem:Proxy}).}
% 
% 	Algorithm~\ref{algo:Proxy} is the algorithm that we use for our proof (Algorithm~\ref{algo:Proxy} only presents the two cases discussed in this paper). We first focus on the base case $\symPattern = \tupleD{\symUorVLeft, \symURI, \symUorVRight}$. Suppose $\fctsymWebSafe\bigl( \symPattern \,\big|\, \fctDom{\mu} \bigr) = \fctVars{\symPattern}$. Then, by Definition~\ref{def:cond-web-safeness}, \todo{...}
%
	For the proof, we provide Algorithm~\ref{algo:Proxy} and show that this (recursive) algorithm has the desired properties for any possible graph pattern (i.e., any case of the induction, including the base case). Due to space limitations, in this paper we only present a fragment of the algorithm and highlight essential properties thereof. The given fragment covers the base case (lines \ref{line:Proxy:BaseCases:Begin}-\ref{line:Proxy:BaseCases:End}) and one pivotal case of the induction step, namely, graph patterns of the form $(\symPattern_1 \OpAND \symPattern_2)$ (lines \ref{line:Proxy:AND:Begin}-\ref{line:Proxy:AND:End}). The complete version of the algorithm and the full proof can be found in
		\PaperVersion{an extended version of this paper~\cite{ExtendedVersion}.}%
		\ExtendedVersion{the Appendix.} % (cf.~Section~\ref{proofsec:lem:Proxy}).}

	For the base
		%case (i.e., PP patterns of the form $\tupleD{\symUorVLeft, \symURI, \symUorVRight}$ or $\tupleD{\symUorVLeft, !(\symURI_1\mid\dots\mid \symURI_n) , \symUorVRight}$),
		case,
	Algorithm~\ref{algo:Proxy} looks up at most one IRI (cf.~lines \ref{line:Proxy:BaseCases:SelectingLookupURI:1}-\ref{line:Proxy:BaseCases:Lookup}). The crux of showing that the returned result is sound and complete is Lemma~\ref{lem:ContextWithInputBaseCases} and the fact that the only possible \textit{context} in which a triple $\tupleD{s,p,o}$ with $s \in \symAllURIs$ can be found is~$\fctContext{\symWoD}{s}$.

	For PP patterns of the form $(\symPattern_1 \OpAND \symPattern_2)$ consider lines \ref{line:Proxy:AND:Begin}-\ref{line:Proxy:AND:End}. By using Definition~\ref{def:cond-web-safeness}, we show $\fctsymWebSafe\bigl( \symPattern_i \,|\, \fctDom{\mu_\mathsf{in}} \bigr) = \fctVars{\symPattern_i}$
		%and, for each $\mu \in \Omega^{\symPattern_i}$\!, $\fctsymWebSafe\bigl( \symPattern_j \,\big|\, \fctDom{\mu_\mathsf{in}} \cup \fctDom{\mu} \bigr) = \fctVars{\symPattern_j}$.
		and $\fctsymWebSafe\bigl( \symPattern_j \,\big|\, \fctDom{\mu_\mathsf{in}} \cup \fctDom{\mu} \bigr) = \fctVars{\symPattern_j}$ for all $\mu \in \Omega^{\symPattern_i}$\!.
		Therefore,
		%Thus,
	by induction, all recursive calls
		(lines \ref{line:Proxy:AND:Recursion1} and \ref{line:Proxy:AND:Recursion2})
		%(lines~\ref{line:Proxy:AND:Recursion1},\ref{line:Proxy:AND:Recursion2})
	look up a finite number of IRIs and return correct
		%(intermediate)
	results; i.e., $\tupleD{\Omega^{\symPattern_i},\fctsymCard^{\symPattern_i}} = \EvalCtx{\symPattern_i \,|\, \mu_\mathsf{in}\,}{\symWoD}$ and $\tupleD{\Omega^{\mu},\fctsymCard^{\mu}} = \EvalCtx{\symPattern_j \,|\, \mu_\mathsf{in} \cup \mu\,}{\symWoD}$ for all $\mu \in \Omega^{\symPattern_i}$\!.
		Then, since each $\mu \in \Omega^{\symPattern_i}$ is compatible with all $\mu' \in \Omega^{\mu}$ and all processed solution mappings are compatible with $\mu_\mathsf{in}$,
		%Then,
	it is easily verified that the computed result is $\EvalCtx{(\symPattern_1 \OpAND \symPattern_2) \,|\, \mu_\mathsf{in}\,}{\symWoD}$.
\end{proofidea}

\begin{algorithm}[t]
{\footnotesize
\begin{algorithmic}[1]
	\IF {$\symPattern$ is of the form $\tupleD{\symUorVLeft, \symURI, \symUorVRight}$ \OR $\symPattern$ is of the form $\tupleD{\symUorVLeft, !(\symURI_1\mid\dots\mid \symURI_n) , \symUorVRight}$} \label{line:Proxy:BaseCases:Begin}
		\STATE \textbf{if} $\symUorVLeft \in \symAllURIs$ \textbf{then} $\symURI'$ := $\symUorVLeft$ \label{line:Proxy:BaseCases:SelectingLookupURI:1}
		\STATE \textbf{else if} $\symUorVLeft \in \symAllVariables$ \AND $\symUorVLeft \in \fctDom{\mu_\mathsf{in}}$ \AND $\mu_\mathsf{in}(\symUorVLeft) \in \symAllURIs$ \textbf{then} $\symURI'$ := $\mu_\mathsf{in}(\symUorVLeft)$ \label{line:Proxy:BaseCases:SelectingLookupURI:2}
		\STATE \textbf{else} $\symURI'$ := \texttt{null} \label{line:Proxy:BaseCases:SelectingLookupURI:3}
		\vspace{2mm}

		\IF {$\symURI'$ is an IRI and looking it up results in retrieving a document, say $d$} \label{line:Proxy:BaseCases:Lookup}
			\STATE $\symRDFgraph$ := the set of triples in $d$ (use a fresh set of blank node identifiers when parsing $d$) \label{line:Proxy:BaseCases:RetrieveData}
			\STATE $\symRDFgraph'$\! := $\big\lbrace \tupleD{s,p,o} \in \symRDFgraph \,\big|\, s = \symURI' \big\rbrace$ \label{line:Proxy:BaseCases:GenerateContext}
			\STATE $\tupleD{\Omega,\fctsymCard}$ := $\fctOrigQueryPG{\symPattern}{\symRDFgraph'}$ \hspace{1mm} ($\fctOrigQueryPG{\symPattern}{\symRDFgraph'}$ can be computed by using any algorithm that\par \hspace{28.5mm} implements the standard SPARQL evaluation function) \label{line:Proxy:BaseCases:GenerateQueryResult}
			\RETURN a new multiset $\tupleD{\Omega',\fctsymCard'}$ with $\Omega' = \big\lbrace \mu' \in \Omega \,\big|\, \mu' \compatible \mu_\mathsf{in} \big\rbrace$ and \par \hspace{50mm} $\fctsymCard'(\mu') = \fctCard{\mu'}$ for all $\mu'\! \in \Omega'$ \label{line:Proxy:BaseCases:ReturnComputedResult}
		\ELSE
			\RETURN a new empty multiset $\tupleD{\Omega,\fctsymCard}$ with $\Omega = \emptyset$ and $\fctDom{\fctsymCard} = \emptyset$ \label{line:Proxy:BaseCases:ReturnEmptyResult}
		\ENDIF \vspace{2ex} \label{line:Proxy:BaseCases:End}
\par \vspace{-2ex} \hspace{-3.7mm} $\ldots$ \vspace{1ex}

% 	\STATE \hspace{-3.7mm}\textbf{else if} $\symPattern$ is $\ldots$ \vspace{2ex}

% 	\STATE $\ldots$ \vspace{2ex}

\setcounter{ALC@line}{56}
	\ELSIF {$\symPattern$ is of the form $(\symPattern_1 \OpAND \symPattern_2)$} \label{line:Proxy:AND:Begin}
		\STATE \textbf{if} $\fctsymWebSafe\bigl( \symPattern_1 \,|\, \fctDom{\mu_\mathsf{in}} \bigr) = \fctVars{\symPattern_1}$ \textbf{then} $i$ := 1; $j$ := 2 \textbf{else} $i$ := 2; $j$ := 1
		\STATE Create a new empty multiset $M = \tupleD{\Omega,\fctsymCard}$ with $\Omega = \emptyset$ and $\fctDom{\fctsymCard} = \emptyset$
		\STATE $\tupleD{\Omega^{\symPattern_i},\fctsymCard^{\symPattern_i}}$ := \textit{EvalCtxBased}$( \symPattern_i, \mu_\mathsf{in})$ \label{line:Proxy:AND:Recursion1}
		\FORALL {$\mu \in \Omega^{\symPattern_i}$}
			\STATE $\tupleD{\Omega^{\mu},\fctsymCard^{\mu}}$ := \textit{EvalCtxBased}$( \symPattern_j, \mu_\mathsf{in} \cup \mu )$ \label{line:Proxy:AND:Recursion2}
			\FORALL {$\mu' \in \Omega^{\mu}$}
				\STATE $\mu^*$ := $\mu \cup \mu'$
				\STATE $k$ := $\fctsymCard^{\symPattern_i}\!(\mu) \cdot \fctsymCard^\mu\!(\mu')$
				\IF {$\mu^*\! \in \Omega$}
					\STATE $\mathit{old}$ := $\fctCard{\mu^*}$
					\STATE Adjust $\fctsymCard$ such that $\fctCard{\mu^*} = k + \mathit{old}$
				\ELSE
					\STATE Adjust $\fctsymCard$ such that $\fctCard{\mu^*} = k$
					\STATE Add $\mu^*$ to $\Omega$
				\ENDIF
			\ENDFOR
		\ENDFOR
		\RETURN $M$
	\ENDIF \label{line:Proxy:AND:End}
	%\vspace{2ex}

% 	\STATE \textbf{else if} $\symPattern$ is $\ldots$
\end{algorithmic}
}
	\caption{ ~ \textit{EvalCtxBased}$(\symPattern, \mu_\mathsf{in} )$, which computes $\EvalCtx{\symPattern \,|\, \mu_\mathsf{in}}{\symWoD}$.} %  for some \Web\ $\symWoD$\!.}
	\label{algo:Proxy}
\end{algorithm}

\noindent
We are now ready to prove Theorem~\ref{thm:MainResult}, for which we use Lemma~\ref{lem:Proxy}, or more precisely the algorithm that we introduce in the proof of the lemma.

\begin{proof}[Theorem~\ref{thm:MainResult}]
		%Suppose $\symPattern$ is
		Let $\symPattern$ be
	a graph pattern
		%such that
		s.t.~%
	$\fctCondWebSafe{\symPattern}{\emptyset} = \fctVars{\symPattern}$. Then,
		%by using
		given
	the empty solution mapping $\muEmpty$ with $\fctDom{\muEmpty} = \emptyset$, we have $\fctsymWebSafe\bigl( \symPattern \,\big|\, \fctDom{\muEmpty} \bigr) = \fctVars{\symPattern}$. Therefore,
		%by Lemma~\ref{lem:Proxy}, there exists an algorithm that, for any finite \Web\ $\symWoD$\!, computes $\EvalCtx{\symPattern \,|\, \mu\,}{\symWoD}$ by looking up a finite number of URIs only, and by our proof of the lemma we know that Algorithm~\ref{algo:Proxy} is such an algorithm.
		by our proof of Lemma~\ref{lem:Proxy} we know that, for any finite \Web\ $\symWoD$\!, Algorithm~\ref{algo:Proxy} computes $\EvalCtx{\symPattern \,|\, \muEmpty\,}{\symWoD}$ by looking up a finite number of IRIs.
	We also know that
		%the empty solution mapping $\muEmpty$
		the empty solution mapping 
		%$\muEmpty$
	is compatible with any solution mapping. Consequently, by Definition~\ref{def:ContextWithInput}, $\EvalCtx{\symPattern \,|\, \muEmpty\,}{\symWoD} \!=\! \EvalCtx{\symPattern}{\symWoD}$ for any \Web\ $\symWoD$\!. Hence, by passing
		the empty solution mapping
		%$\muEmpty$
	to it, Algorithm~\ref{algo:Proxy} can be used to compute $\EvalCtx{\symPattern}{\symWoD}$ for any finite \Web~$\symWoD$\!, and during this computation the algorithm looks up a finite number of IRIs only.
	\qed
\end{proof}

\noindent
While the condition
	%given
in Theorem~\ref{thm:MainResult} is sufficient to
	%verify that evaluating a graph pattern under con\-text-based semantics is feasible,
	%verify the Web-safeness of a given graph pattern,
	identify Web-safe graph patterns,
the question that remains is whether it is a necessary condition (in which case it could be used to decide Web-safeness of \emph{all} graph patterns).  %%% be very careful if you want to change the sentence in parenthesis!
	%The following example shows that this is not the case.
	Unfortunately, the answer is no.
\begin{exmp}\hspace{-1ex}\textbf{.}
	Consider the graph pattern $\symPattern = (\symPattern_1 \OpUNION \symPattern_2)$ with $\symPattern_1 = \tupleD{\symURI_1,p_1,?x}$ and $\symPattern_2 = \tupleD{\symURI_2,p_2,?y}$. We note that $\fctCondWebSafe{\symPattern_1}{\emptyset} = \lbrace ?x \rbrace$ and $\fctCondWebSafe{\symPattern_2}{\emptyset} = \lbrace ?y \rbrace$, and, thus, $\fctCondWebSafe{\symPattern}{\emptyset} = \emptyset$. Hence, the pattern does not satisfy the condition in Theorem~\ref{thm:MainResult}. Nonetheless, it is easy to see that there exists a (sound and complete) algorithm that, for any \Web\ $\symWoD$\!, computes $\EvalCtx{\symPattern}{\symWoD}$ by looking up a finite number of IRIs only. For instance, such an algorithm, say $A$, may first use two other algorithms that compute $\EvalCtx{\symPattern_1}{\symWoD}$ and $\EvalCtx{\symPattern_2}{\symWoD}$ by looking up a finite number of IRIs, respectively. Such algorithms exist by Theorem~\ref{thm:MainResult}, because $\fctCondWebSafe{\symPattern_1}{\emptyset} = \fctVars{\symPattern_1}$ and $\fctCondWebSafe{\symPattern_2}{\emptyset} = \fctVars{\symPattern_2}$. Finally, algorithm $A$ can generate the (sound and complete) query result $\EvalCtx{\symPattern}{\symWoD}$ by computing the multiset union $\EvalCtx{\symPattern_1}{\symWoD} \multicup \EvalCtx{\symPattern_2}{\symWoD}$, which requires no additional IRI~lookups.
\end{exmp}
\begin{rmrk}\hspace{-1ex}\textbf{.}
	The example illustrates that ``only if'' cannot be shown in Theorem~\ref{thm:MainResult}. It remains an open question whether there exists an alternative condition
		for Web-safeness that is both sufficient and necessary (and decidable). % and, thus, can be used to decide Web-safeness of all graph patterns.
		%that can be used to decide Web-safeness of all graph patterns.
\end{rmrk}

\section{Concluding Remarks and Future Work}
\label{sec:conclusion}
% % % %
%
	%This is the first paper that studies
	This paper studies
	%We have studied
the problem of extending the scope of SPARQL property paths to query Linked Data that is distributed on the WWW. We have
	proposed a con\-text-based query semantics and analyzed its peculiarities. Our perhaps most interesting finding is that
	%defined different possible query semantics and analyzed their peculiarities. Our perhaps most interesting finding is that, for a con\-text-based semantics, which is comparably restrictive in most cases,
there exist queries whose evaluation over the WWW is not feasible. We studied this aspect and introduced a decidable syntactic property for identifying feasible queries.

We believe that the presented work provides valuable input to a wider discussion about defining a language for
	accessing
	%querying and navigating over
Linked Data on the WWW. In this context, there are several directions for future research such as the following three. First,
	%identifying and
studying a more expressive navigational core for property paths over the Web; e.g., along the lines of other navigational languages such as nSPARQL~\cite{perez2010} or NautiLOD~\cite{fionda2012}. Second, investigating relationships between navigational queries and SPARQL federation. Third, while the aim of this paper was to introduce a formal foundation for answering SPARQL queries with PPs over Linked Data on the WWW, an investigation of how systems may implement efficiently the machinery developed in this paper is certainly~interesting.
% % % %
%% % % %
\bibliographystyle{splncs03}
\bibliography{biblio}

\ExtendedVersion{%
\newpage
\appendix
% \section{Proofs} \label{Appendix:Proofs}
% 
% \subsection{Proof of Lemma~\ref{lem:Proxy}}
% \label{proofsec:lem:Proxy}

\section{Proof of Lemma~\ref{lem:Proxy}} \label{Appendix:Proofs}

Suppose $\symPattern$ is a graph pattern and $\mu_\mathsf{in}$ is a solution mapping such that $$\fctsymWebSafe\bigl( \symPattern \,\big|\, \fctDom{\mu_\mathsf{in}} \bigr) = \fctVars{\symPattern}.$$ We have to show that there exists a (sound and complete) algorithm that, for any finite \Web\ $\symWoD$\!, computes $\EvalCtx{\symPattern \,|\, \mu_\mathsf{in}\,}{\symWoD}$ by looking up a finite number of IRIs only. For the proof we provide Algorithm~\ref{algo:Proxy} and show by induction on the possible structure of graph pattern $\symPattern$ that this (recursive) algorithm has the desired properties.

For the proof we use the following fact, which is easily verified by Definition~\ref{def:cond-web-safeness}.

\begin{fact}\hspace{-1ex}\textbf{.} \label{fact:MonotonicityOfWSV}
	Let $\symPattern$ be a graph pattern, and let $\symWebSafeVar \subseteq \symAllVariables$ and $\symWebSafeVar' \subseteq \symAllVariables$ be two (nonempty) sets of variables.
	Then, $\fctCondWebSafe{\symPattern}{\symWebSafeVar} \subseteq \fctCondWebSafe{\symPattern}{\symWebSafeVar \cup \symWebSafeVar'}$.
\end{fact}

% \subsubsection{Base case:}
\subsection{Base Case}
Suppose $\symPattern$ is either a PP pattern $\tupleD{\symUorVLeft,\symURI,\symUorVRight}$
	%with $\symURI \in \symAllURIs$ or a PP pattern $\tupleD{\symUorVLeft,!( \symURI_1 \,|\, ... \,|\, \symURI_n),\symUorVRight}$ with $\symURI_1, ... , \symURI_n \in \symAllURIs$.
	or a PP pattern $\tupleD{\symUorVLeft,!( \symURI_1 \,|\, ... \,|\, \symURI_n),\symUorVRight}$ (with $\symURI,\symURI_1, ... , \symURI_n \in \symAllURIs$).
The corresponding fragment of Algorithm~\ref{algo:Proxy} for this case is given as follows.

\vspace{1ex}
\begin{small}
\begin{algorithmic}[1]
	\IF {$\symPattern$ is of the form $\tupleD{\symUorVLeft, \symURI, \symUorVRight}$ \OR $\symPattern$ is of the form $\tupleD{\symUorVLeft, !(\symURI_1\mid\dots\mid \symURI_n) , \symUorVRight}$}
		\STATE \textbf{if} $\symUorVLeft \in \symAllURIs$ \textbf{then} $\symURI'$ := $\symUorVLeft$
		\STATE \textbf{else if} $\symUorVLeft \in \symAllVariables$ \AND $\symUorVLeft \in \fctDom{\mu_\mathsf{in}}$ \AND $\mu_\mathsf{in}(\symUorVLeft) \in \symAllURIs$ \textbf{then} $\symURI'$ := $\mu_\mathsf{in}(\symUorVLeft)$
		\STATE \textbf{else} $\symURI'$ := \texttt{null}
		\vspace{2mm}

		\IF {$\symURI'$ is an IRI and looking it up results in retrieving a document, say $d$}
			\STATE $\symRDFgraph$ := the set of triples in $d$ (use a fresh set of blank node identifiers when parsing $d$)
			\STATE $\symRDFgraph'$\! := $\big\lbrace \tupleD{s,p,o} \in \symRDFgraph \,\big|\, s = \symURI' \big\rbrace$
			\STATE $\tupleD{\Omega,\fctsymCard}$ := $\fctOrigQueryPG{\symPattern}{\symRDFgraph'}$ \hspace{1mm} ($\fctOrigQueryPG{\symPattern}{\symRDFgraph'}$ can be computed by using any algorithm that\par \hspace{28.5mm} implements the standard SPARQL evaluation function)
			\RETURN a new multiset $\tupleD{\Omega',\fctsymCard'}$ with $\Omega' = \big\lbrace \mu' \in \Omega \,\big|\, \mu' \compatible \mu_\mathsf{in} \big\rbrace$ and \par \hspace{50mm} $\fctsymCard'(\mu') = \fctCard{\mu'}$ for all $\mu'\! \in \Omega'$
		\ELSE
			\RETURN a new empty multiset $\tupleD{\Omega,\fctsymCard}$ with $\Omega = \emptyset$ and $\fctDom{\fctsymCard} = \emptyset$		
		\ENDIF
	\ENDIF
\end{algorithmic}
\end{small}
\vspace{1ex}

\noindent
We distinguish three cases (which correspond to the three cases in lines \ref{line:Proxy:BaseCases:SelectingLookupURI:1}-\ref{line:Proxy:BaseCases:SelectingLookupURI:3}):

\begin{enumerate}
	\itemsep2mm
	\item \label{item:proof:Base:1}
		If $\symUorVLeft$ is an IRI (i.e., $\symUorVLeft \in \symAllURIs$), Algorithm~\ref{algo:Proxy} looks up this IRI, which either may result in retrieving a document or not. In the following, we consider both cases:

		\smallskip
		\begin{enumerate}
			\itemsep1mm
			\item \label{item:proof:Base:1:1}
				If the lookup results in retrieving a document $d$, Algorithm~\ref{algo:Proxy} executes lines \ref{line:Proxy:BaseCases:RetrieveData} to~\ref{line:Proxy:BaseCases:ReturnComputedResult}, and we know that $d \in \symDocs$ and $\fctADoc{\symUorVLeft} = d$ hold for the queried \Web\ $\symWoD = \tupleD{\symDocs,\fctsymData,\fctsymADoc}$. In this case the algorithm selects specific triples from document $d$ to obtain an RDF graph $\symRDFgraph'$ (cf.~line~\ref{line:Proxy:BaseCases:GenerateContext}). Since this selection resembles the application of the context selector $\fctsymContext{\symWoD}$ (cf.~Section~\ref{sec:w3c-ctx-sem}), it holds that $\symRDFgraph' = \fctContext{\symWoD}{\symUorVLeft}$. Then, it is easily seen that, by using a standard evaluation algorithm for the computation in line~\ref{line:Proxy:BaseCases:GenerateQueryResult}, multiset $\tupleD{\Omega,\fctsymCard}$ is equivalent to query result $\EvalCtx{\symPattern}{\symWoD}$ (cf.~Figure~\ref{fig:stc-sem}) and $\tupleD{\Omega',\fctsymCard'}$ is equivalent to $\EvalCtx{\symPattern \,|\, \mu_\mathsf{in}\,}{\symWoD}$~(cf. Definition~\ref{def:ContextWithInput}).

			\item \label{item:proof:Base:1:2}
				If the lookup of IRI $\symUorVLeft$ does not result in retrieving a document, Algorithm~\ref{algo:Proxy} executes line~\ref{line:Proxy:BaseCases:ReturnEmptyResult}, and we know that $\symUorVLeft \notin \fctDom{\fctsymADoc}$ holds for the queried \Web\ $\symWoD = \tupleD{\symDocs,\fctsymData,\fctsymADoc}$. As a consequence, $\fctContext{\symWoD}{\symUorVLeft} = \emptyset$ (cf.~Section~\ref{sec:w3c-ctx-sem}). Then, by Definition~\ref{def:Context}, $\EvalCtx{\symPattern}{\symWoD}$ is the empty multiset of solution mappings, and so is $\EvalCtx{\symPattern \,|\, \mu_\mathsf{in}\,}{\symWoD}$ (cf.~Definition~\ref{def:ContextWithInput}). Hence, the empty multiset of solution mappings returned by Algorithm~\ref{algo:Proxy} (line~\ref{line:Proxy:BaseCases:ReturnEmptyResult}) is the correct result in this case.
		\end{enumerate}

	\item \label{item:proof:Base:2}
		If $\symUorVLeft$ is a variable and solution mapping $\mu_\mathsf{in}$ binds this variable to an IRI (i.e., $\symUorVLeft \in \symAllVariables$ and $\mu_\mathsf{in}(\symUorVLeft) \in \symAllURIs$), then Algorithm~\ref{algo:Proxy} looks up this IRI, which either may result in retrieving a document or not. In the following, we consider both cases:

		\smallskip
		\begin{enumerate}
			\itemsep1mm
			\item \label{item:proof:Base:2:1}
				If the lookup results in retrieving a document $d$, Algorithm~\ref{algo:Proxy} executes lines \ref{line:Proxy:BaseCases:RetrieveData} to~\ref{line:Proxy:BaseCases:ReturnComputedResult}, and we know that $d \in \symDocs$ and $\fctADoc{\mu_\mathsf{in}(\symUorVLeft)} = d$ hold for the queried \Web\ $\symWoD = \tupleD{\symDocs,\fctsymData,\fctsymADoc}$. Similar to case~\ref{item:proof:Base:1:1} before, we can show for the RDF graph $\symRDFgraph'$ constructed in line~\ref{line:Proxy:BaseCases:GenerateContext}, that $\symRDFgraph' = \fctContext{\symWoD}{\mu_\mathsf{in}(\symUorVLeft)}$ holds. Since $\symUorVLeft$ is a variable, by Definition~\ref{def:Context}, we would have to search for triples that match triple pattern $\symTP = \mu_\mathsf{in}[\tupleD{\symUorVLeft, p, \symUorVRight}]$ (with $p = \symURI$; resp.~$p \in \symAllURIs \setminus \lbrace \symURI_1, ...\,, \symURI_n \rbrace)$ in the context $\fctContext{\symWoD}{\symURI^*}$ of \emph{all} IRIs $\symURI^*\! \in \symAllURIs$. However, since $\mu_\mathsf{in}(\symUorVLeft)$ is an IRI, the only context that can contain such matching triples is $\symRDFgraph' = \fctContext{\symWoD}{\mu_\mathsf{in}(\symUorVLeft)}$~(cf. Section~\ref{sec:w3c-ctx-sem}). As a consequence, $\EvalCtx{\symPattern}{\symWoD} = \fctOrigQueryPG{\symPattern}{\symRDFgraph'}$ and, thus, the multiset of solution mappings $\tupleD{\Omega',\fctsymCard'}$ returned in line~\ref{line:Proxy:BaseCases:ReturnComputedResult} is equivalent to $\EvalCtx{\symPattern \,|\, \mu_\mathsf{in}\,}{\symWoD}$~(cf. Definition~\ref{def:ContextWithInput}).

			\item \label{item:proof:Base:2:2}
				If the lookup of IRI $\symUorVLeft$ does not result in retrieving a document, Algorithm~\ref{algo:Proxy} executes line~\ref{line:Proxy:BaseCases:ReturnEmptyResult}, and we know that $\mu_\mathsf{in}(\symUorVLeft) \notin \fctDom{\fctsymADoc}$ holds for the queried \Web\ $\symWoD = \tupleD{\symDocs,\fctsymData,\fctsymADoc}$. As in case~\ref{item:proof:Base:2:1}, the only context that can contain matching triples for triple pattern $\mu_\mathsf{in}[\tupleD{\symUorVLeft, p, \symUorVRight}]$ is $\fctContext{\symWoD}{\mu_\mathsf{in}(\symUorVLeft)}$. However, $\fctContext{\symWoD}{\mu_\mathsf{in}(\symUorVLeft)} = \emptyset$ because $\mu_\mathsf{in}(\symUorVLeft) \notin \fctDom{\fctsymADoc}$. Thus, $\EvalCtx{\symPattern}{\symWoD}$ is the empty multiset of solution mappings (cf.~Definition~\ref{def:Context}), and so is $\EvalCtx{\symPattern \,|\, \mu_\mathsf{in}\,}{\symWoD}$~(cf. Definition~\ref{def:ContextWithInput}). Hence, the empty multiset of solution mappings returned by Algorithm~\ref{algo:Proxy} (line~\ref{line:Proxy:BaseCases:ReturnEmptyResult}) is the correct result in this case.
		\end{enumerate}

	\item \label{item:proof:Base:3}
		If none of the other two cases holds, then either \enumA~$\symUorVLeft$ is a variable and solution mapping $\mu_\mathsf{in}$ binds this variable to a blank node or a to literal (i.e., $\symUorVLeft \in \symAllVariables$ and $\mu_\mathsf{in}(\symUorVLeft) \in \symAllBNodes \cup \symAllLiterals$) or \enumB~$\symUorVLeft$ is a literal. Note that, due to $\fctsymWebSafe\bigl( \symPattern \,\big|\, \fctDom{\mu_\mathsf{in}} \bigr) = \fctVars{\symPattern}$, by Definition~\ref{def:cond-web-safeness}, we can rule out a third possibility of $\symUorVLeft$ being a variable that is not bound at all by solution mapping $\mu_\mathsf{in}$ (i.e., $\symUorVLeft \in \symAllVariables$ and $\symUorVLeft \notin \fctDom{\mu_\mathsf{in}}$). Algorithm~\ref{algo:Proxy} executes line~\ref{line:Proxy:BaseCases:ReturnEmptyResult} and returns the empty multiset of solution mappings. In the following, we show that this is the correct result for each of the two (possible)~sub-cases:

		\smallskip
		\begin{enumerate}
			\itemsep1mm
			\item \label{item:proof:Base:3:1}
				If $\symUorVLeft \in \symAllVariables$ and $\mu_\mathsf{in}(\symUorVLeft) \in \symAllBNodes \cup \symAllLiterals$, then, by Lemma~\ref{lem:ContextWithInputBaseCases}, query result $\EvalCtx{\symPattern \,|\, \mu_\mathsf{in}\,}{\symWoD}$ is the empty multiset.

			\item \label{item:proof:Base:3:2}
				If $\symUorVLeft \in \symAllLiterals$, then, by Definition~\ref{def:Context}, query result $\EvalCtx{\symPattern}{\symWoD}$ is the empty multiset of solution mappings, and so is $\EvalCtx{\symPattern \,|\, \mu_\mathsf{in}\,}{\symWoD}$.
		\end{enumerate}
\end{enumerate}

\noindent
Our discussion shows that, for each of the three cases, Algorithm~\ref{algo:Proxy} looks up a finite number of IRIs (that is, one in the first and in the second case, respectively, and none in the third case) and returns the correct result.

% \bigskip\noindent
\subsection{Induction Step}
We now discuss the induction step, for which we distinguish ten cases.

\subsubsection{Case 1:}
Suppose $\symPattern$ is a PP pattern $\tupleD{\symUorVLeft,\symPPEi,\symUorVRight}$.

The fragment of Algorithm~\ref{algo:Proxy} that covers this case is given as follows.

\vspace{1ex}
\begin{small}
\begin{algorithmic}[1]
\setcounter{ALC@line}{11}
	\IF {$\symPattern$ is of the form $\tupleD{\symUorVLeft,\symPPEi,\symUorVRight}$}
		\STATE Create a PP pattern $\symPP' = \tupleD{\symUorVRight,\symPPE,\symUorVLeft}$ \label{line:Proxy:Inverse:ConstructReversePattern}
		\RETURN \textit{EvalCtxBased}$\bigl( \symPattern'\!, \mu_\mathsf{in} \bigr)$
	\ENDIF
\end{algorithmic}
\end{small}
\vspace{1ex}

\noindent
Let $\symPP' = \tupleD{\symUorVRight,\symPPE,\symUorVLeft}$ be the PP pattern created in line~\ref{line:Proxy:Inverse:ConstructReversePattern}. To show that, for any finite \Web\ $\symWoD$\!, Algorithm~\ref{algo:Proxy} computes $\EvalCtx{\symPattern \,|\, \mu_\mathsf{in}\,}{\symWoD}$ by looking up a finite number of IRIs only, it suffices to prove the following two claims:
\begin{itemize}
	\itemsep2mm
	\item[] \textit{Claim 1}: $\EvalCtx{\symPattern \,|\, \mu_\mathsf{in}\,}{\symWoD} = \EvalCtx{ \symPP' \,|\, \mu_\mathsf{in}\,}{\symWoD}$ for any \Web\ $\symWoD$\!.
	\item[] \textit{Claim 2}: $\fctsymWebSafe\bigl( \symPP' \,\big|\, \fctDom{\mu_\mathsf{in}} \bigr) = \fctVars{\symPP'}$.
\end{itemize}

\noindent
Then, by induction it follows that Algorithm~\ref{algo:Proxy} has the desired properties for pattern $\symPattern$.

To verify the first claim we recall that $\EvalCtx{\symPattern}{\symWoD} \!=\! \EvalCtx{\symPP'}{\symWoD}$ holds for any \Web~$\symWoD$~(cf. Definition~\ref{def:Context}). By using this equivalence and Definition~\ref{def:ContextWithInput}, we obtain Claim~1.

To prove Claim~2 we use the fact that
\begin{align*}
	\fctsymWebSafe\bigl( \symPattern \,\big|\, \fctDom{\mu_\mathsf{in}} \bigr) &= \fctVars{\symPattern} .
\intertext{Since, $\fctsymWebSafe\bigl( \symPattern \,\big|\, \fctDom{\mu_\mathsf{in}} \bigr) = \fctsymWebSafe\bigl( \symPP' \,\big|\, \fctDom{\mu_\mathsf{in}} \bigr)$ (cf.~Definition~\ref{def:cond-web-safeness}), we thus have}
	\fctsymWebSafe\bigl( \symPP' \,\big|\, \fctDom{\mu_\mathsf{in}} \bigr) &= \fctVars{\symPattern} ,
\intertext{Then, by using $\fctVars{\symPattern} = \fctVars{\symPP'}$, we obtain}
	\fctsymWebSafe\bigl( \symPP' \,\big|\, \fctDom{\mu_\mathsf{in}} \bigr) &= \fctVars{\symPP'} .
\end{align*}

\subsubsection{Case 2:}
Suppose $\symPattern$ is a PP pattern $\tupleD{\symUorVLeft,\symPPE_1/\symPPE_2,\symUorVRight}$.

The fragment of Algorithm~\ref{algo:Proxy} that covers this case is given as follows.

\vspace{1ex}
\begin{small}
\begin{algorithmic}[1]
\setcounter{ALC@line}{14}
	\IF {$\symPattern$ is of the form $\tupleD{\symUorVLeft,\symPPE_1/\symPPE_2,\symUorVRight}$}
		\STATE Create a graph pattern $\symPattern' = \bigl( \tupleD{\symUorVLeft,\symPPE_1,?v} \OpAND \tupleD{?v,\symPPE_2,\symUorVRight} \bigr)$ \par \hspace{28mm} such that $?v \in \symAllVariables \setminus \bigl( \fctDom{\mu_\mathsf{in}} \cup \lbrace \symUorVLeft,\symUorVRight \rbrace \bigr)$ \label{line:Proxy:Concat:ConstructANDPattern}
		\STATE $M$ := \textit{EvalCtxBased}$\bigl( \symPattern'\!, \mu_\mathsf{in} \bigr)$
		\STATE $M'$ := $\pi_{ \lbrace \symUorVLeft,\symUorVRight \rbrace \cap \symAllVariables } ( M )$ \, (this multiset projection is defined in Section~\ref{subsec:preliminaries} and
		\par \hspace{29.5mm} can be computed by using a standard algorithm)
		\RETURN $M'$
	\ENDIF
\end{algorithmic}
\end{small}
\vspace{1ex}

\noindent
Let $\symPattern' = \bigl( \tupleD{\symUorVLeft,\symPPE_1,?v} \OpAND \tupleD{?v,\symPPE_2,\symUorVRight} \bigr)$ be the graph pattern created in line~\ref{line:Proxy:Concat:ConstructANDPattern}; i.e., $?v \in \symAllVariables \setminus \bigl( \fctDom{\mu_\mathsf{in}} \cup \lbrace \symUorVLeft,\symUorVRight \rbrace \bigr)$ and, thus, $?v \notin \fctDom{\mu_\mathsf{in}}$. To show that, for any finite \Web\ $\symWoD$\!, Algorithm~\ref{algo:Proxy} computes $\EvalCtx{\symPattern \,|\, \mu_\mathsf{in}\,}{\symWoD}$ by looking up a finite number of IRIs only, it suffices to prove the following two claims:
\begin{itemize}
	\itemsep2mm
	\item[] \textit{Claim 1}: $\EvalCtx{\symPattern \,|\, \mu_\mathsf{in}\,}{\symWoD} = \pi_{ \lbrace \symUorVLeft,\symUorVRight \rbrace \cap \symAllVariables } \bigl( \EvalCtx{\symPattern' \,|\, \mu_\mathsf{in}\,}{\symWoD} \bigr)$ for any \Web\ $\symWoD$\!.
	\item[] \textit{Claim 2}: $\fctsymWebSafe\bigl( \symPattern' \,\big|\, \fctDom{\mu_\mathsf{in}} \bigr) = \fctVars{\symPattern'}$.
\end{itemize}

\noindent
Then, by induction it follows that Algorithm~\ref{algo:Proxy} has the desired properties for pattern $\symPattern$.

To verify the first claim we recall that $\EvalCtx{\symPattern}{\symWoD} = \pi_{ \lbrace \symUorVLeft,\symUorVRight \rbrace \cap \symAllVariables } \bigl( \EvalCtx{\symPattern'}{\symWoD} \bigr)$ holds for any \Web\ $\symWoD$ (cf.~Definition~\ref{def:Context}). By using this equivalence, the fact that $?v \notin \fctDom{\mu_\mathsf{in}}$, and Definition~\ref{def:ContextWithInput}, we obtain Claim 1.

To prove Claim 2 we recall that $\fctsymWebSafe\bigl( \symPattern \,\big|\, \fctDom{\mu_\mathsf{in}} \bigr) = \fctVars{\symPattern}$. Therefore, by Definition~\ref{def:cond-web-safeness}, it holds that $?v \in \fctsymWebSafe\bigl( \symPattern' \,\big|\, \fctDom{\mu_\mathsf{in}} \bigr)$ and:
\begin{align*}
	\fctsymWebSafe\bigl( \symPattern \,\big|\, \fctDom{\mu_\mathsf{in}} \bigr) &= \fctsymWebSafe\bigl( \symPattern' \,\big|\, \fctDom{\mu_\mathsf{in}} \bigr) \setminus \lbrace ?v \rbrace .
\intertext{Due to the former, we can rewrite the latter to obtain:}
	\fctsymWebSafe\bigl( \symPattern \,\big|\, \fctDom{\mu_\mathsf{in}} \bigr) \cup \lbrace ?v \rbrace &= \fctsymWebSafe\bigl( \symPattern' \,\big|\, \fctDom{\mu_\mathsf{in}} \bigr) .
\intertext{By using $\fctsymWebSafe\bigl( \symPattern \,\big|\, \fctDom{\mu_\mathsf{in}} \bigr) = \fctVars{\symPattern}$ again, we rewrite to:}
	\fctVars{\symPattern} \cup \lbrace ?v \rbrace &= \fctsymWebSafe\bigl( \symPattern' \,\big|\, \fctDom{\mu_\mathsf{in}} \bigr) ,
\intertext{and, with $\fctVars{\symPattern} \cup \lbrace ?v \rbrace = \fctVars{\symPattern'}$,}
	\fctVars{\symPattern'} &= \fctsymWebSafe\bigl( \symPattern' \,\big|\, \fctDom{\mu_\mathsf{in}} \bigr) .
\end{align*}

\subsubsection{Case 3:}
Suppose $\symPattern$ is a PP pattern $\tupleD{\symUorVLeft,(\symPPE_1|\symPPE_2),\symUorVRight}$.

This case is covered by the following fragment of Algorithm~\ref{algo:Proxy}.

\vspace{1ex}
\begin{small}
\begin{algorithmic}[1]
\setcounter{ALC@line}{19}
	\IF {$\symPattern$ is of the form $\tupleD{\symUorVLeft,\symPPE_1|\symPPE_2,\symUorVRight}$}
		\STATE Create graph pattern $\symPattern' = \bigl( \tupleD{\symUorVLeft,\symPPE_1,\symUorVRight} \OpUNION \tupleD{\symUorVLeft,\symPPE_2,\symUorVRight} \bigr)$ \label{line:Proxy:AlternativePaths:ConstructEquivPattern}
		\STATE $M$ := \textit{EvalCtxBased}$\bigl( \symPattern'\!, \mu_\mathsf{in} \bigr)$
		\RETURN $M$
	\ENDIF
\end{algorithmic}
\end{small}
\vspace{1ex}

\noindent
Due to the semantics of the operator\OpUNION (as given in Section~\ref{subsec:DefinitionSPARQL}), for the graph pattern 
$\symPattern'$ constructed in line~\ref{line:Proxy:AlternativePaths:ConstructEquivPattern} of Algorithm~\ref{algo:Proxy} and any \Web\ $\symWoD$\!, it holds that
\begin{align*}
	\EvalCtx{\symPattern'}{\symWoD} &= \EvalCtx{\tupleD{\symUorVLeft, \symPPE_1, \symUorVRight}}{\symWoD} \multicup \EvalCtx{\tupleD{\symUorVLeft, \symPPE_2, \symUorVRight}}{\symWoD} .
\intertext{Furthermore, by Definition~\ref{def:Context}, for any \Web\ $\symWoD$\!, it holds that}
	\EvalCtx{\tupleD{\symUorVLeft,\symPPE_1 \,|\, \symPPE_2,\symUorVRight}}{\symWoD} &= \EvalCtx{\tupleD{\symUorVLeft, \symPPE_1, \symUorVRight}}{\symWoD} \multicup \EvalCtx{\tupleD{\symUorVLeft, \symPPE_2, \symUorVRight}}{\symWoD} .
\end{align*}
Hence, for any \Web\ $\symWoD$\!, $\EvalCtx{\symPattern}{\symWoD} = 
\EvalCtx{\symPattern'}{\symWoD}$ and, thus,
\begin{equation} \label{eq:proof:AlternativePaths:1}
	\EvalCtx{\symPattern \,|\, \mu_\mathsf{in}\,}{\symWoD} = \EvalCtx{\symPattern' \,|\, \mu_\mathsf{in}\,}{\symWoD} .
\end{equation}

\noindent
Moreover, by using \enumA~the fact that $\fctsymWebSafe\bigl( \symPattern \,|\, \fctDom{\mu_\mathsf{in}} \bigr) = \fctVars{\symPattern}$, \enumB~$\fctVars{\symPattern} = \fctVars{\symPattern'}$, and \enumC~$\fctsymWebSafe\bigl( \symPattern \,|\, \fctDom{\mu_\mathsf{in}} \bigr) = \fctsymWebSafe\bigl( \symPattern' \,|\, \fctDom{\mu_\mathsf{in}} \bigr)$ (cf.~Definition~\ref{def:cond-web-safeness}), we can show
\begin{equation} \label{eq:proof:AlternativePaths:2}
	\fctsymWebSafe\bigl( \symPattern' \,|\, \fctDom{\mu_\mathsf{in}} \bigr) = \fctVars{\symPattern'} .
\end{equation}

\noindent
Due to (\ref{eq:proof:AlternativePaths:1}) and (\ref{eq:proof:AlternativePaths:2}), we may use the same argument as for case 6 below---which is the case that covers patterns of the form $(\symPattern_1 \OpUNION \symPattern_2)$---to show that, for any finite \Web~$\symWoD$\!, Algorithm~\ref{algo:Proxy} computes query result $\EvalCtx{\tupleD{\symUorVLeft,(\symPPE_1|\symPPE_2),\symUorVRight} \,|\, \mu_\mathsf{in}\,}{\symWoD}$ by looking up a finite number of IRIs only.

\subsubsection{Case 4:}
Suppose $\symPattern$ is a PP pattern $\tupleD{\symLeftConst, (\symPPE)^*, \symRightVar}$ s.t.~$\symLeftConst \in ( \symAllURIs \cup \symAllLiterals )$ and $\symRightVar \in \symAllVariables$.

We have to show that, for any finite \Web\ $\symWoD$\!, Algorithm~\ref{algo:Proxy} computes
	query result
$\EvalCtx{\tupleD{\symLeftConst, (\symPPE)^*, \symRightVar} \,|\, \mu_\mathsf{in}\,}{\symWoD}$ by looking up a finite number of IRIs only. The corresponding fragment of Algorithm~\ref{algo:Proxy} that covers this case is given as follows.

\vspace{1ex}
\begin{small}
\begin{algorithmic}[1]
\setcounter{ALC@line}{23}
	\IF {$\symPattern$ is of the form $\tupleD{\symLeftConst, (\symPPE)^*, \symRightVar}$ such that $\symLeftConst \in ( \symAllURIs \cup \symAllLiterals )$ and $\symRightVar \in \symAllVariables$} \label{line:Proxy:Star1:BeginOfFragment}
		\STATE Create a new empty multiset $M = \tupleD{\Omega,\fctsymCard}$ with $\Omega = \emptyset$ and $\fctDom{\fctsymCard} = \emptyset$ \label{line:Proxy:Star1:InitializeOutput}
		\STATE $X$ := \textit{ExecALPW1}$(\symLeftConst, \symPPE)$ \label{line:Proxy:Star1:CallExecALPW1}
		\FORALL {$x \in X$} \label{line:Proxy:Star1:BeginForLoop}
			\IF {$\symRightVar \notin \fctDom{\mu_\mathsf{in}}$ \OR $\mu_\mathsf{in}(\symRightVar) = x$}
				\STATE Create a new solution mapping $\mu$ such that $\fctDom{\mu} = \lbrace \symRightVar \rbrace$ and $\mu(\symRightVar) = x$
				\STATE Add $\mu$ to $\Omega$
				\STATE Adjust $\fctsymCard$ such that $\fctCard{\mu} = 1$
			\ENDIF
		\ENDFOR \label{line:Proxy:Star1:EndForLoop}
		\RETURN $M$ \label{line:Proxy:Star1:Return}
	\ENDIF
\end{algorithmic}
\end{small}
\vspace{1ex}

\noindent
Line~\ref{line:Proxy:Star1:CallExecALPW1} of the given fragment of Algorithm~\ref{algo:Proxy} calls a function \textit{ExecALPW1}. This function is given by Algorithm~\ref{algo:ExecALPW1}; it calls another function, named \textit{ExecALPW2} (cf.~Algorithm~\ref{algo:ExecALPW2}). It is easily seen that function \textit{ExecALPW1} implements the auxiliary function $\mathtt{ALPW1}$ as used in Definition~\ref{def:Context}~(cf.~Figure~\ref{Figure:WebALP}). Before we discuss Algorithm~\ref{algo:Proxy}, we prove the following two claims:
\begin{itemize}
	\itemsep2mm
	\item[] \textit{Claim 1}: Function \textit{ExecALPW2} implements the other auxiliary function, $\mathtt{ALPW2}$.
	\item[] \textit{Claim 2}: During any execution of \textit{ExecALPW2}, the execution of Algorithm~\ref{algo:Proxy} in \par \hspace{12.5mm} line~\ref{line:ExecALPW2:CallAlgo2} looks up a finite number of IRIs only.
\end{itemize}

\noindent
To prove these claims we use the fact that $\fctsymWebSafe\bigl( \symPattern \,\big|\, \fctDom{\mu_\mathsf{in}} \bigr) = \fctVars{\symPattern}$. Therefore, by Definition~\ref{def:cond-web-safeness}, we know that, for any two variables $?v \in \symAllVariables$ and $?w \in \symAllVariables$, it holds that $\fctsymWebSafe\bigl( \tupleD{?v,\symPPE,?w} \,|\, \lbrace ?v \rbrace \bigr) = \lbrace ?v,?w \rbrace$. Hence, $\fctsymWebSafe\bigl( \symPP' \,|\, \fctDom{\mu'} \bigr) = \fctVars{\symPP'}$ where $\symPP' \!= \tupleD{?x, \symPPE, ?y}$ is the PP pattern created in line~\ref{line:ExecALPW2:CreatePattern} of function \textit{ExecALPW2}~(cf. Algorithm~\ref{algo:ExecALPW2}) and $\mu'$ is the solution mapping created in line~\ref{line:ExecALPW2:CreateSolMapping}. Therefore, by induction we can assume that the execution of Algorithm~\ref{algo:Proxy} in line~\ref{line:ExecALPW2:CallAlgo2} has two properties: \enumA~it returns $\EvalCtx{\symPP' |\, \mu' \,}{\symWoD}$ and \enumB~it looks up a finite number of IRIs only. While the latter directly verifies Claim~2, we use the former to show Claim~1; in particular, we use $\tupleD{\Omega,\fctsymCard} = \EvalCtx{\symPP' |\, \mu' \,}{\symWoD}$, where $\tupleD{\Omega,\fctsymCard}$ is the multiset initialized in line~\ref{line:ExecALPW2:CallAlgo2}. Then, due to the properties of solution mapping $\mu'$ (cf.~line~\ref{line:ExecALPW2:CreateSolMapping}), for each solution mapping $\mu \in \Omega$, it holds that $\mu(?x) = \gamma$. Consequently, function \textit{ExecALPW2} implements the auxiliary function $\mathtt{ALPW2}$, where $\Omega$ in function \textit{ExecALPW2} corresponds to the set of all solution mappings that are considered by the loop in $\mathtt{ALPW2}$~(cf.~lines \ref{line:WebALP:ForBegin}-\ref{line:WebALP:ForEnd} in Figure~\ref{Figure:WebALP}).

After proving Claims 1 and 2, we now come back to Algorithm~\ref{algo:Proxy}. For the multiset $M$ that is populated by lines \ref{line:Proxy:Star1:BeginForLoop}-\ref{line:Proxy:Star1:EndForLoop} in Algorithm~\ref{algo:Proxy}, let $M^*$ denote the fully populated version of $M$ (i.e., before executing the return statement in line~\ref{line:Proxy:Star1:Return}). Since functions \textit{ExecALPW1} and \textit{ExecALPW2} implement $\mathtt{ALPW1}$ and $\mathtt{ALPW2}$, respectively, it can be easily seen that $M^* = \EvalCtx{\symPattern \,|\, \mu_\mathsf{in}\,}{\symWoD}$ (i.e., Algorithm~\ref{algo:Proxy} returns the expected result for PP pattern $\symPattern$). It remains to show that, during the computation of this result over a finite \Web, Algorithm~\ref{algo:Proxy} looks up a finite number of IRIs only: Due to the use of set \textit{Visited} in function \textit{ExecALPW2}, none of the IRIs that recursive calls of this function discover is considered more than once. As a consequence of this observation and of Claim~2, it follows that, if the queried \Web\ $\symWoD = \tupleD{\symDocs,\fctsymData,\fctsymADoc}$ is finite, then $\fctDom{\fctsymADoc}$ is finite and, thus, any execution of function \textit{ExecALPW2} (including all recursive calls in line~\ref{line:ExecALPW2:Recursion}) looks up a finite number of IRIs only, and so does the execution of~\textit{ExecALPW1} in line~\ref{line:Proxy:Star1:CallExecALPW1} of Algorithm~\ref{algo:Proxy}. Since none of the other lines of the corresponding fragment of Algorithm~\ref{algo:Proxy} (i.e., lines \ref{line:Proxy:Star1:BeginOfFragment}-\ref{line:Proxy:Star1:Return}) involves IRI lookups, the algorithm looks up a finite number of IRIs to compute $\EvalCtx{\symPattern \,|\, \mu_\mathsf{in}\,}{\symWoD}$ for any finite \Web~$\symWoD$\!.

\begin{algorithm}[t]
{\footnotesize
\begin{algorithmic}[1]
	\STATE \textit{Visited} := $\emptyset$
	\STATE \textit{Visited} := \textit{ExecALPW2}( $\gamma$, $\symPPE$, \textit{Visited} )
	\RETURN \textit{Visited}
\end{algorithmic}
}
	\caption{ ~ \textit{ExecALPW1}$(\gamma, \symPPE)$, which computes $\mathtt{ALPW1}(\gamma,\symPPE,\symWoD)$ (as given in Figure~\ref{Figure:WebALP}) for the queried \Web\ $\symWoD$\!.}
	\label{algo:ExecALPW1}
\end{algorithm}

\begin{algorithm}[t]
{\footnotesize
\begin{algorithmic}[1]
	\IF {$\gamma \notin$ \textit{Visited}}
		\STATE Add $\gamma$ to \textit{Visited}
		\STATE Create a PP pattern $\symPP' \!= \tupleD{?x, \symPPE, ?y}$ with $?x,?y \in \symAllVariables$ \label{line:ExecALPW2:CreatePattern}
		\STATE Create a new solution mapping $\mu'$ such that $\fctDom{\mu'} = \lbrace ?x \rbrace$ and $\mu'(?x) = \gamma$ \label{line:ExecALPW2:CreateSolMapping}
		\STATE $\tupleD{\Omega,\fctsymCard}$ := \textit{EvalCtxBased}$\bigl( \symPP'\!, \mu' \,\bigr)$ \quad (i.e., call Algorithm~\ref{algo:Proxy} to compute $\EvalCtx{\symPP' |\, \mu' \,}{\symWoD}$) \label{line:ExecALPW2:CallAlgo2}

		\FORALL {$\mu \in \Omega$}
			\STATE \textit{Visited} := \textit{ExecALPW2}$\bigl( \mu(?y), \symPPE, \mathit{Visited} \bigr)$\label{line:ExecALPW2:Recursion}
		\ENDFOR

	\ENDIF
	\RETURN \textit{Visited}
\end{algorithmic}
}
	\caption{ ~ \textit{ExecALPW2}$(\gamma, \symPPE,\mathit{Visited})$, which computes auxiliary function $\mathtt{ALPW2}(\gamma,\symPPE,\mathit{Visited},\symWoD)$ (as given in Figure~\ref{Figure:WebALP}) for the queried \Web\ $\symWoD$\!.}
	\label{algo:ExecALPW2}
\end{algorithm}

\subsubsection{Case 5:}
Suppose $\symPattern$ is a PP pattern $\tupleD{\symLeftVar, (\symPPE)^*, \symRightVar}$ such that $\symLeftVar, \symRightVar \in \symAllVariables$.

The fragment of Algorithm~\ref{algo:Proxy} that covers this case is given as follows:

\vspace{1ex}
\begin{small}
\begin{algorithmic}[1]
\setcounter{ALC@line}{32}
	\IF {$\symPattern$ is of the form $\tupleD{\symLeftVar, (\symPPE)^*, \symRightVar}$ such that $\symLeftVar \in \symAllVariables$ and $\symRightVar \in \symAllVariables$}
		\IF {$\symLeftVar \in \fctDom{\mu_\mathsf{in}}$}
			\STATE Create a new empty multiset $M = \tupleD{\Omega,\fctsymCard}$ with $\Omega = \emptyset$ and $\fctDom{\fctsymCard} = \emptyset$ \label{line:Proxy:Star2:FirstStatementOfFirstSubCase}
			\STATE $X$ := \textit{ExecALPW1}$(\mu_\mathsf{in}(\symLeftVar), \symPPE)$ \label{line:Proxy:Star2:CallExecALPW1}
			\FORALL {$x \in X$} \label{line:Proxy:Star2:BeginForLoop}
				\IF {$\symRightVar \notin \fctDom{\mu_\mathsf{in}}$ \OR $\mu_\mathsf{in}(\symRightVar) = x$}
					\STATE Create a new solution mapping $\mu$ such that \enumA~$\fctDom{\mu} = \lbrace \symLeftVar,\symRightVar \rbrace$, \par \enumB~$\mu(\symLeftVar) = \mu_\mathsf{in}(\symLeftVar)$, and \enumC~$\mu(\symRightVar) = x$
					\STATE Add $\mu$ to $\Omega$
					\STATE Adjust $\fctsymCard$ such that $\fctCard{\mu} = 1$
				\ENDIF
			\ENDFOR \label{line:Proxy:Star2:EndForLoop}
			\RETURN $M$ \label{line:Proxy:Star2:Return1}
		\ELSE
			\STATE Create PP pattern $\symPP' = \tupleD{\symRightVar, (\symPPEi)^*, \symLeftVar}$ \label{line:Proxy:Star2:ConstructEquivPattern}
			\STATE $M$ := \textit{EvalCtxBased}$\bigl( \symPP'\!, \mu_\mathsf{in} \bigr)$ \label{line:Proxy:Star2:Recursion}
			\RETURN $M$ \label{line:Proxy:Star2:Return2}
		\ENDIF
	\ENDIF
\end{algorithmic}
\end{small}
\vspace{1ex}

\noindent
The algorithm distinguishes whether $\symLeftVar \in \fctDom{\mu_\mathsf{in}}$ or $\symLeftVar \notin \fctDom{\mu_\mathsf{in}}$. In the former case, Algorithm~\ref{algo:Proxy} executes lines \ref{line:Proxy:Star2:FirstStatementOfFirstSubCase}-\ref{line:Proxy:Star2:Return1}, which are similar to the fragment of Algorithm~\ref{algo:Proxy} that covers the previous Case~4 (cf.~lines \ref{line:Proxy:Star1:InitializeOutput}-\ref{line:Proxy:Star1:Return} before), and the proof that executing lines \ref{line:Proxy:Star2:FirstStatementOfFirstSubCase}-\ref{line:Proxy:Star2:Return1} has the desired properties for PP pattern $\tupleD{\symLeftVar, (\symPPE)^*, \symRightVar}$ is also similar to the discussion of Case~4. Hence, we omit repeating this discussion and focus on the second sub-case, $\symLeftVar \notin \fctDom{\mu_\mathsf{in}}$ (which is covered by lines \ref{line:Proxy:Star2:ConstructEquivPattern}-\ref{line:Proxy:Star2:Return2}). As a basis for discussing this case we need the following two lemmas. We prove these lemmas after completing the proof of Lemma~\ref{lem:Proxy} (cf.~page~\pageref{sec:ProofOfLemma:OneVariableIsBound} for the proof of Lemma~\ref{lem:OneVariableIsBound} and page~\pageref{sec:ProofOfLemma:EquivalenceOfSwappingStarAndInverse} for the proof of Lemma~\ref{lem:EquivalenceOfSwappingStarAndInverse}).

\begin{lemma} \label{lem:OneVariableIsBound}
	Let $\symPP = \tupleD{\symLeftVar, \symPPE, \symRightVar}$ be a PP pattern such that $\symLeftVar, \symRightVar \in \symAllVariables$, and let $\symWebSafeVar \subseteq \symAllVariables$ be a set of variables.
	If $\fctCondWebSafe{\symPP}{\symWebSafeVar} = \fctVars{\symPP}$, then $\symLeftVar \in \symWebSafeVar$ or $\symRightVar \in \symWebSafeVar$.
\end{lemma}

\begin{lemma} \label{lem:EquivalenceOfSwappingStarAndInverse}
	For any PP expression $\symPPE$ and any pair of variables $\symLeftVar, \symRightVar \in \symAllVariables$, the two PP patterns $\symPP = \tupleD{\symLeftVar, (\symPPE)^*, \symRightVar}$ and $\symPP' = \tupleD{\symRightVar, (\symPPEi)^*, \symLeftVar}$ are semantically equivalent under context-based semantics; i.e., $\EvalCtx{\symPP}{\symWoD} = \EvalCtx{\symPP'}{\symWoD}$ holds for any \Web~$\symWoD$\!.
\end{lemma}

\noindent
Due to the fact that $\fctsymWebSafe\bigl( \symPattern \,\big|\, \fctDom{\mu_\mathsf{in}} \bigr) = \fctVars{\symPattern}$, we can use Lemma~\ref{lem:OneVariableIsBound} to show that, if $\symLeftVar \notin \fctDom{\mu_\mathsf{in}}$, then $\symRightVar \in \fctDom{\mu_\mathsf{in}}$. Therefore, the recursive call in line~\ref{line:Proxy:Star2:Recursion} (which swaps the subject and the object) will result in executing an instance of Algorithm~\ref{algo:Proxy} that meets the first sub-case (i.e., the recursive call in line~\ref{line:Proxy:Star2:Recursion} performs lines \ref{line:Proxy:Star2:FirstStatementOfFirstSubCase}-\ref{line:Proxy:Star2:Return1}).

Moreover, the fact that $\fctsymWebSafe\bigl( \symPattern \,\big|\, \fctDom{\mu_\mathsf{in}} \bigr) = \fctVars{\symPattern}$ can also be used to show that $\fctsymWebSafe\bigl( \symPP' \,\big|\, \fctDom{\mu_\mathsf{in}} \bigr) = \fctVars{\symPP'}$ where $\symPP'\! = \tupleD{\symRightVar, (\symPPEi)^*, \symLeftVar}$ is the PP pattern~created in line~\ref{line:Proxy:Star2:ConstructEquivPattern}. Then, by induction we can assume that, for any finite \Web~$\symWoD$\!, the recursive call in line~\ref{line:Proxy:Star2:Recursion} looks up a finite number of IRIs only and returns $\EvalCtx{\symPP' \,|\, \mu_\mathsf{in}}{\symWoD}$. As a consequence, we can use Lemma~\ref{lem:EquivalenceOfSwappingStarAndInverse} and Definition~\ref{def:ContextWithInput} to show that Algorithm~\ref{algo:Proxy} has the desired properties for graph pattern $\symPattern$ with $\symLeftVar \notin \fctDom{\mu_\mathsf{in}}$.

\subsubsection{Case 6:}
Suppose $\symPattern$ is a PP pattern $\tupleD{\symLeftVar, (\symPPE)^*, \symRightConst}$ s.t.~$\symLeftVar \in \symAllVariables$ and $\symRightConst \in ( \symAllURIs \cup \symAllLiterals )$.

The fragment of Algorithm~\ref{algo:Proxy} that covers this case is given as follows:

\vspace{1ex}
\begin{small}
\begin{algorithmic}[1]
\setcounter{ALC@line}{46}
	\IF {$\symPattern$ is of the form $\tupleD{\symLeftVar, (\symPPE)^*, \symRightConst}$ such that $\symLeftVar \in \symAllVariables$ and $\symRightConst \in ( \symAllURIs \cup \symAllLiterals )$}
		\STATE Create PP pattern $\symPP' = \tupleD{\symRightConst, (\symPPEi)^*, \symLeftVar}$ \label{line:Proxy:Star3:ConstructEquivPattern}
		\STATE $M$ := \textit{EvalCtxBased}$\bigl( \symPP'\!, \mu_\mathsf{in} \bigr)$
		\RETURN $M$
	\ENDIF
\end{algorithmic}
\end{small}
\vspace{1ex}

\noindent
Let $\symPP'$ be the PP pattern created in line~\ref{line:Proxy:Star3:ConstructEquivPattern}; i.e., $\symPP' = \tupleD{\symRightConst, (\symPPEi)^*, \symLeftVar}$. To show that, for any finite \Web\ $\symWoD$\!, Algorithm~\ref{algo:Proxy} computes $\EvalCtx{\symPattern \,|\, \mu_\mathsf{in}\,}{\symWoD}$ by looking up a finite number of IRIs only, it suffices to prove the following two claims:
\begin{itemize}
	\itemsep2mm
	\item[] \textit{Claim 1}: $\EvalCtx{\symPattern \,|\, \mu_\mathsf{in}\,}{\symWoD} = \EvalCtx{\symPP' \,|\, \mu_\mathsf{in}\,}{\symWoD}$ for any \Web\ $\symWoD$\!.
	\item[] \textit{Claim 2}: $\fctsymWebSafe\bigl( \symPP' \,\big|\, \fctDom{\mu_\mathsf{in}} \bigr) = \fctVars{\symPP'}$.
\end{itemize}

\noindent
Then, by induction it follows that Algorithm~\ref{algo:Proxy} has the desired properties for pattern $\symPattern$.

To verify the first claim we recall that, for any \Web\ $\symWoD$\!, $\EvalCtx{\symPattern}{\symWoD} \!=\! \EvalCtx{\symPP'}{\symWoD}$~(cf.~Definition~\ref{def:Context}). Therefore, by Definition~\ref{def:ContextWithInput}, Claim 1 follows trivially.

It remains to prove Claim 2. By Definition~\ref{def:cond-web-safeness}, we have:
\begin{align*}
	\fctsymWebSafe\bigl( \symPattern \,\big|\, \fctDom{\mu_\mathsf{in}} \bigr) &= \fctsymWebSafe\bigl( \symPP' \,\big|\, \fctDom{\mu_\mathsf{in}} \bigr) .
\intertext{By using the fact that $\fctsymWebSafe\bigl( \symPattern \,\big|\, \fctDom{\mu_\mathsf{in}} \bigr) = \fctVars{\symPattern}$, we obtain:}
	\fctVars{\symPattern} &= \fctsymWebSafe\bigl( \symPP' \,\big|\, \fctDom{\mu_\mathsf{in}} \bigr) ,
\intertext{and, due to $\fctVars{\symPattern} = \fctVars{\symPP'}$,}
	\fctVars{\symPP'} &= \fctsymWebSafe\bigl( \symPP' \,\big|\, \fctDom{\mu_\mathsf{in}} \bigr) .
\end{align*}

\subsubsection{Case 7:}
Suppose $\symPattern$ is a PP pattern $\tupleD{\symLeftConst, (\symPPE)^*, \symRightConst}$ such that $\symLeftConst, \symRightConst \in ( \symAllURIs \cup \symAllLiterals )$.

The fragment of Algorithm~\ref{algo:Proxy} that covers this case is given as follows:

\vspace{1ex}
\begin{small}
\begin{algorithmic}[1]
\setcounter{ALC@line}{50}
	\IF {$\symPattern$ is of the form $\tupleD{\symLeftConst, (\symPPE)^*, \symRightConst}$ such that $\symLeftConst \in ( \symAllURIs \cup \symAllLiterals )$ and $\symRightConst \in ( \symAllURIs \cup \symAllLiterals )$}
		\STATE $X$ := \textit{ExecALPW1}$(\symLeftConst, \symPPE)$ \label{line:Proxy:Star4:CallExecALPW1}
		\FORALL {$x \in X$} \label{line:Proxy:Star4:BeginForLoop}
			\IF {$x = \symRightConst$}
				\RETURN a new multiset $\tupleD{\Omega,\fctsymCard}$ with $\Omega = \lbrace \muEmpty \rbrace$ and $\fctsymCard = \fctsymCardGen{\Omega}$
			\ENDIF
		\ENDFOR \label{line:Proxy:Star4:EndForLoop}
		\RETURN a new empty multiset $M = \tupleD{\Omega,\fctsymCard}$ with $\Omega = \emptyset$ and $\fctDom{\fctsymCard} = \emptyset$
	\ENDIF
\end{algorithmic}
\end{small}
\vspace{1ex}

\noindent
This fragment of the algorithm leverages the fact that the definition of query result $\EvalCtx{\tupleD{\symLeftConst, (\symPPE)^*, \symRightConst}}{\symWoD}$ (cf.~Figure~\ref{fig:stc-sem}) can be rewritten as follows:
$$\EvalCtx{\tupleD{\symLeftConst, (\symPPE)^*, \symRightConst}}{\symWoD} =
		\Big\langle\,
			\begin{cases} \lbrace \muEmpty \rbrace & \text{if } \symRightConst \in \mathtt{ALWP1}(\symLeftConst,\symPPE,\symWoD) , \\ ~~\emptyset & \text{else} \end{cases} \textbf{ , }\,
			\fctsymCardGen{\Omega}
		\,\Big\rangle
.$$
Then, the discussion of this case resembles the discussion of Case~4 above.

\subsubsection{Case 8:}
Suppose $\symPattern$ is $(\symPattern_1 \OpAND \symPattern_2)$.

As a basis for discussing this case, we first show that
\begin{equation} \label{eq:proof:AND:1}
	\fctsymWebSafe\bigl( \symPattern_1 \,|\, \fctDom{\mu_\mathsf{in}} \bigr) = \fctVars{\symPattern_1}
	\quad \text{ or } \quad
	\fctsymWebSafe\bigl( \symPattern_2 \,|\, \fctDom{\mu_\mathsf{in}} \bigr) = \fctVars{\symPattern_2} .
\end{equation}
Thereafter, we use this fact to show that Algorithm~\ref{algo:Proxy} has the desired properties for $\symPattern = (\symPattern_1 \OpAND \symPattern_2)$.

To show (\ref{eq:proof:AND:1}), we use proof by contradiction. That is, we assume
\begin{equation*}
	\fctsymWebSafe\bigl( \symPattern_1 \,|\, \fctDom{\mu_\mathsf{in}} \bigr) \neq \fctVars{\symPattern_1}
	\quad \text{ and } \quad
	\fctsymWebSafe\bigl( \symPattern_2 \,|\, \fctDom{\mu_\mathsf{in}} \bigr) \neq \fctVars{\symPattern_2} .
\end{equation*}
Then, by Definition~\ref{def:cond-web-safeness}, $\fctsymWebSafe\bigl( \symPattern \,|\, \fctDom{\mu_\mathsf{in}} \bigr) = \emptyset$. Since $\fctsymWebSafe\bigl( \symPattern \,|\, \fctDom{\mu_\mathsf{in}} \bigr) = \fctVars{\symPattern}$, we have $\fctVars{\symPattern} = \emptyset$ and, thus,
\begin{equation} \label{eq:proof:AND:Contradiction}
	\fctVars{\symPattern_1} = \emptyset
	\quad \text{ and } \quad
	\fctVars{\symPattern_2} = \emptyset .
\end{equation}
Since $\fctsymWebSafe\bigl( \symPattern' \,|\, \fctDom{\mu_\mathsf{in}} \bigr) \subseteq \fctVars{\symPattern'}$ holds for any graph pattern $\symPattern'$ (cf.~Definition~\ref{def:cond-web-safeness}), we have $\fctsymWebSafe\bigl( \symPattern_1 \,|\, \fctDom{\mu_\mathsf{in}} \bigr) \subseteq \fctVars{\symPattern_1}$ and $\fctsymWebSafe\bigl( \symPattern_2 \,|\, \fctDom{\mu_\mathsf{in}} \bigr) \subseteq \fctVars{\symPattern_2}$. With (\ref{eq:proof:AND:Contradiction}), we obtain
\begin{equation*}
	\fctsymWebSafe\bigl( \symPattern_1 \,|\, \fctDom{\mu_\mathsf{in}} \bigr) = \emptyset
	\quad \text{ and } \quad
	\fctsymWebSafe\bigl( \symPattern_2 \,|\, \fctDom{\mu_\mathsf{in}} \bigr) = \emptyset .
\end{equation*}
Hence, again with (\ref{eq:proof:AND:Contradiction}),
\begin{equation*}
	\fctsymWebSafe\bigl( \symPattern_1 \,|\, \fctDom{\mu_\mathsf{in}} \bigr) = \fctVars{\symPattern_1}
	\quad \text{ and } \quad
	\fctsymWebSafe\bigl( \symPattern_2 \,|\, \fctDom{\mu_\mathsf{in}} \bigr) = \fctVars{\symPattern_2}  ,	 	
\end{equation*}
which contradicts our assumption and, thus, shows that (\ref{eq:proof:AND:1}) holds.

We now show that, for any finite \Web\ $\symWoD$\!, Algorithm~\ref{algo:Proxy} computes
	query result
$\EvalCtx{(\symPattern_1 \OpAND \symPattern_2) \,|\, \mu_\mathsf{in}\,}{\symWoD}$ by looking up a finite number of IRIs only.
The fragment of Algorithm~\ref{algo:Proxy} that covers this case is given as follows.

\vspace{1ex}
\begin{small}
\begin{algorithmic}[1]
\setcounter{ALC@line}{56}
	\IF {$\symPattern$ is of the form $(\symPattern_1 \OpAND \symPattern_2)$}
		\STATE \textbf{if} $\fctsymWebSafe\bigl( \symPattern_1 \,|\, \fctDom{\mu_\mathsf{in}} \bigr) = \fctVars{\symPattern_1}$ \textbf{then} $i$ := 1; $j$ := 2 \textbf{else} $i$ := 2; $j$ := 1
		\STATE Create a new empty multiset $M = \tupleD{\Omega,\fctsymCard}$ with $\Omega = \emptyset$ and $\fctDom{\fctsymCard} = \emptyset$
		\STATE $\tupleD{\Omega^{\symPattern_i},\fctsymCard^{\symPattern_i}}$ := \textit{EvalCtxBased}$( \symPattern_i, \mu_\mathsf{in})$
		\FORALL {$\mu \in \Omega^{\symPattern_i}$}
			\STATE $\tupleD{\Omega^{\mu},\fctsymCard^{\mu}}$ := \textit{EvalCtxBased}$( \symPattern_j, \mu_\mathsf{in} \cup \mu )$
			\FORALL {$\mu' \in \Omega^{\mu}$}
				\STATE $\mu^*$ := $\mu \cup \mu'$
				\STATE $k$ := $\fctsymCard^{\symPattern_i}\!(\mu) \cdot \fctsymCard^\mu\!(\mu')$
				\IF {$\mu^*\! \in \Omega$}
					\STATE $\mathit{old}$ := $\fctCard{\mu^*}$
					\STATE Adjust $\fctsymCard$ such that $\fctCard{\mu^*} = k + \mathit{old}$
				\ELSE
					\STATE Adjust $\fctsymCard$ such that $\fctCard{\mu^*} = k$
					\STATE Add $\mu^*$ to $\Omega$
				\ENDIF
			\ENDFOR
		\ENDFOR
		\RETURN $M$
	\ENDIF
\end{algorithmic}
\end{small}
\vspace{1ex}

\noindent
The algorithm first determines whether $\fctsymWebSafe\bigl( \symPattern_1 \,|\, \fctDom{\mu_\mathsf{in}} \bigr) = \fctVars{\symPattern_1}$ (which is decidable by using Definition~\ref{def:cond-web-safeness} recursively). If $\fctsymWebSafe\bigl( \symPattern_1 \,|\, \fctDom{\mu_\mathsf{in}} \bigr) = \fctVars{\symPattern_1}$, the algorithm lets $i = 1$ and $j = 2$; if $\fctCondWebSafe{\symPattern_1}{\emptyset} \neq \fctVars{\symPattern_1}$, $i = 2$ and $j = 1$. Due to (\ref{eq:proof:AND:1}), it holds that $\fctsymWebSafe\bigl( \symPattern_i \,|\, \fctDom{\mu_\mathsf{in}} \bigr) = \fctVars{\symPattern_i}$. Therefore, by induction we can assume that, when Algorithm~\ref{algo:Proxy} calls itself in line~\ref{line:Proxy:AND:Recursion1}, the recursive execution looks up a finite number of IRIs only and for the result $\tupleD{\Omega^{\symPattern_i},\fctsymCard^{\symPattern_i}}$ it holds that $\tupleD{\Omega^{\symPattern_i},\fctsymCard^{\symPattern_i}} = \EvalCtx{\symPattern_i \,|\, \mu_\mathsf{in}\,}{\symWoD}$.

Next, the algorithm iterates over all solution mappings $\mu \in \Omega^{\symPattern_i}$. We claim that
\begin{equation} \label{eq:proof:AND:2}
	\forall \mu \in \Omega^{\symPattern_i} :
	\fctsymWebSafe\bigl( \symPattern_j \,\big|\, \fctDom{\mu_\mathsf{in}} \cup \fctDom{\mu} \bigr) = \fctVars{\symPattern_j} .
\end{equation}

\noindent
Note, if (\ref{eq:proof:AND:2}) holds, by induction we can assume that, for each solution mapping $\mu \in \Omega^{\symPattern_i}$\!, the recursive call in line~\ref{line:Proxy:AND:Recursion2} looks up a finite number of IRIs only and for the result $\tupleD{\Omega^{\mu},\fctsymCard^{\mu}}$ it holds that $\tupleD{\Omega^{\mu},\fctsymCard^{\mu}} = \EvalCtx{\symPattern_j \,|\, \mu_\mathsf{in} \cup \mu\,}{\symWoD}$.

Hence, before we continue the discussion of the algorithm, we prove the claim: Let $\mu$ be an arbitrary solution mapping with $\mu \in \Omega^{\symPattern_i}$\!. W.l.o.g., it suffices to show that $\fctsymWebSafe\bigl( \symPattern_j \,\big|\, \fctDom{\mu_\mathsf{in}} \cup \fctDom{\mu} \bigr) = \fctVars{\symPattern_j}$ holds, for which we use the fact that $\fctsymWebSafe\bigl( \symPattern \,|\, \fctDom{\mu_\mathsf{in}} \bigr) = \fctVars{\symPattern}$ holds. In particular, since $\fctsymWebSafe\bigl( \symPattern_i \,|\, \fctDom{\mu_\mathsf{in}} \bigr) = \fctVars{\symPattern_i}$ holds as well, we note that $\fctsymWebSafe\bigl( \symPattern \,|\, \fctDom{\mu_\mathsf{in}} \bigr) = \fctVars{\symPattern}$ holds only because at least one of the following conditions is satisfied (cf.~Definition~\ref{def:cond-web-safeness}):
$\fctsymWebSafe\bigl( \symPattern_j \,|\, \fctDom{\mu_\mathsf{in}} \bigr) = \fctVars{\symPattern_j}$,
$\fctsymWebSafe\bigl( \symPattern_i \,|\, \fctDom{\mu_\mathsf{in}} \cup \fctCVars{\symPattern_i} \bigr) = \fctVars{\symPattern_i}$, or $\fctVars{\symPattern} = \emptyset$.
We now show that each of these conditions entails $\fctsymWebSafe\bigl( \symPattern_j \,\big|\, \fctDom{\mu_\mathsf{in}} \cup \fctDom{\mu} \bigr) = \fctVars{\symPattern_j}$.
\begin{enumerate}
	\item
		If $\fctsymWebSafe\bigl( \symPattern_j \,|\, \fctDom{\mu_\mathsf{in}} \bigr) = \fctVars{\symPattern_j}$, then $\fctsymWebSafe\bigl( \symPattern_j \,\big|\, \fctDom{\mu_\mathsf{in}} \cup \fctDom{\mu} \bigr) = \fctVars{\symPattern_j}$ follows by using Fact~\ref{fact:MonotonicityOfWSV}.
	\item
		If $\fctsymWebSafe\bigl( \symPattern_i \,|\, \fctDom{\mu_\mathsf{in}} \cup \fctCVars{\symPattern_i} \bigr) = \fctVars{\symPattern_i}$, then, due to $\mu \in \EvalCtx{\symPattern_i \,|\, \mu_\mathsf{in}\,}{\symWoD}$ and, thus, $\fctCVars{\symPattern_i} \subseteq \fctDom{\mu}$, we obtain $\fctsymWebSafe\bigl( \symPattern_j \,\big|\, \fctDom{\mu_\mathsf{in}} \cup \fctDom{\mu} \bigr) = \fctVars{\symPattern_j}$ by Fact~\ref{fact:MonotonicityOfWSV}.
	\item
		If $\fctVars{\symPattern} = \emptyset$, then $\fctsymWebSafe\bigl( \symPattern_j \,\big|\, \fctDom{\mu_\mathsf{in}} \cup \fctDom{\mu} \bigr) = \fctVars{\symPattern_j}$ is a trivial consequence of $\fctVars{\symPattern_j} \subseteq \fctVars{\symPattern}$ and $\fctsymWebSafe\bigl( \symPattern_j \,|\, \fctDom{\mu_\mathsf{in}} \bigr) \subseteq \fctVars{\symPattern_j}$.
\end{enumerate}

\noindent
Hence, we verified the correctness of (\ref{eq:proof:AND:2}) and now come back to Algorithm~\ref{algo:Proxy}. As mentioned before, after computing $\EvalCtx{\symPattern_i \,|\, \mu_\mathsf{in}\,}{\symWoD} = \tupleD{\Omega^{\symPattern_i},\fctsymCard^{\symPattern_i}}$ (in line~\ref{line:Proxy:AND:Recursion1}), for each $\mu \in \Omega^{\symPattern_i}$\!, the recursive call in line~\ref{line:Proxy:AND:Recursion2} computes $\EvalCtx{\symPattern_j \,|\, \mu_\mathsf{in} \cup \mu\,}{\symWoD} = \tupleD{\Omega^{\mu},\fctsymCard^{\mu}}$ by looking up a finite number of IRIs only. Then, the algorithm populates a new, initially empty multiset $M$ incrementally as follows.

For each pair of a solution mapping $\mu \in \Omega^{\symPattern_i}$ and a corresponding solution mapping $\mu' \in \Omega^\mu$\!, the algorithm generates a joined solution mapping $\mu^* \!= \mu \cup \mu'$ (which is possible because, due to $\mu' \in \Omega^\mu$\!, $\mu$ and $\mu'$ are compatible) and adds $\mu^*$ exactly $k$~times to multiset $M$, where $k = \fctsymCard^{\symPattern_i}(\mu) \cdot \fctsymCard^\mu(\mu')$. Let $M^*$ denote the resulting, fully populated version of multiset $M$ (i.e., after populating it incrementally based on all $\mu' \in \Omega^\mu$ for all $\mu \in \Omega^{\symPattern_i}$). It is easily seen that $M^*$ is the expected result of the $\mu_\mathsf{in}$-restricted evaluation of graph pattern $(\symPattern_1 \OpAND \symPattern_2)$ over \Web\ $\symWoD$ (i.e., $M^* = \EvalCtx{(\symPattern_1 \OpAND \symPattern_2) \,|\, \mu_\mathsf{in}\,}{\symWoD}$). Hence, the algorithm returns $M^*$\!. Since each of the recursive calls looks up a finite number of IRIs and the intermediate result $\EvalCtx{\symPattern_i \,|\, \mu_\mathsf{in}\,}{\symWoD}$ is finite (because of the finiteness of the queried \Web\ $\symWoD$\!), the number of IRIs looked up during the
	computation of $\EvalCtx{(\symPattern_1 \OpAND \symPattern_2) \,|\, \mu_\mathsf{in}\,}{\symWoD}$
	%execution of $A$
is finite.

\subsubsection{Case 9:}
Suppose $\symPattern$ is $(\symPattern_1 \OpUNION \symPattern_2)$.

We have to show that, for any finite \Web\ $\symWoD$\!, Algorithm~\ref{algo:Proxy} computes query result $\EvalCtx{(\symPattern_1 \OpUNION \symPattern_2) \,|\, \mu_\mathsf{in}\,}{\symWoD}$ by looking up a finite number of IRIs only. The corresponding fragment of Algorithm~\ref{algo:Proxy} for this case is given as follows.

\vspace{1ex}
\begin{small}
\begin{algorithmic}[1]
\setcounter{ALC@line}{72}
	\IF {$\symPattern$ is of the form $(\symPattern_1 \OpUNION \symPattern_2)$}
		\STATE $M_1$ := \textit{EvalCtxBased}$\bigl( \symPattern_1, \mu_\mathsf{in} \bigr)$ \label{line:Proxy:UNION:Recursion1}
		\STATE $M_2$ := \textit{EvalCtxBased}$\bigl( \symPattern_2, \mu_\mathsf{in} \bigr)$ \label{line:Proxy:UNION:Recursion2}
		\STATE $M$ := $M_1 \multicup M_2$ \, (this multiset union is defined in Section~\ref{subsec:preliminaries} and
		\par \hspace{23mm} can be computed by using a standard algorithm) \label{line:Proxy:UNION:ComputeMultisetUnion}
		\RETURN $M$
	\ENDIF
\end{algorithmic}
\end{small}
\vspace{1ex}

\noindent
As a basis for discussing this case we emphasize that
\begin{equation}
	\fctsymWebSafe\bigl( \symPattern_1 \,\big|\, \fctDom{\mu_\mathsf{in}} \bigr) = \fctVars{\symPattern_1}
	\quad \text{ and } \quad
	\fctsymWebSafe\bigl( \symPattern_2 \,\big|\, \fctDom{\mu_\mathsf{in}} \bigr) = \fctVars{\symPattern_2} ,
\end{equation}
which follows from \enumA~Definition~\ref{def:cond-web-safeness}, \enumB~$\fctVars{\symPattern} = \fctVars{\symPattern_1} \cup \fctVars{\symPattern_2}$, and \enumC~the fact that $\fctsymWebSafe\bigl( \symPattern \,\big|\, \fctDom{\mu_\mathsf{in}} \bigr) = \fctVars{\symPattern}$.
Therefore, by induction we can assume that each of the two recursive calls in line \ref{line:Proxy:UNION:Recursion1} and~\ref{line:Proxy:UNION:Recursion2} looks up a finite number of IRIs in the queried \Web\ $\symWoD$\!, and for the results $M_1$ and $M_2$ it holds that $M_1 = \EvalCtx{\symPattern_1 \,|\, \mu_\mathsf{in} \,}{\symWoD}$ and $M_2 = \EvalCtx{\symPattern_2 \,|\, \mu_\mathsf{in} \,}{\symWoD}$. Then, it is easily seen that $M = M_1 \multicup M_2$ is the expected result of the $\mu_\mathsf{in}$-restricted evaluation of graph pattern $(\symPattern_1 \OpUNION \symPattern_2)$ over \Web~$\symWoD$ (i.e., $M = \EvalCtx{(\symPattern_1 \OpUNION \symPattern_2) \,|\, \mu_\mathsf{in}\,}{\symWoD}$) and the number of IRIs looked up during the computation of this result is finite.

\subsubsection{Case 10:}
Suppose $\symPattern$ is $(\symPattern_1 \OpOPT \symPattern_2)$.

The corresponding fragment of Algorithm~\ref{algo:Proxy} for this case is given as follows.

\vspace{1ex}
\begin{small}
\begin{algorithmic}[1]
\setcounter{ALC@line}{77}
	\IF {$\symPattern$ is of the form $(\symPattern_1 \OpOPT \symPattern_2)$}
		\STATE Create a new empty multiset $M_\mathsf{out} \!=\! \tupleD{\Omega_\mathsf{out},\fctsymCard_\mathsf{out}}$ with $\Omega_\mathsf{out} \!=\! \emptyset$ and $\fctDom{\fctsymCard_\mathsf{out}} \!=\! \emptyset$
		\STATE $\tupleD{\Omega^{\symPattern_1},\fctsymCard^{\symPattern_1}}$ := \textit{EvalCtxBased}$( \symPattern_1, \mu_\mathsf{in})$
		\FORALL {$\mu \in \Omega^{\symPattern_1}$}
			\STATE $\tupleD{\Omega^{\mu},\fctsymCard^{\mu}}$ := \textit{EvalCtxBased}$( \symPattern_2, \mu )$
			\IF {$\Omega^{\mu} = \emptyset$}
				\IF {$\mu \in \Omega_\mathsf{out}$}
					\STATE $\mathit{old}$ := $\fctsymCard_\mathsf{out}(\mu)$
					\STATE Adjust $\fctsymCard_\mathsf{out}$ such that $\fctsymCard_\mathsf{out}(\mu) = \mathit{old} + 1$
				\ELSE
					\STATE Adjust $\fctsymCard_\mathsf{out}$ such that $\fctsymCard_\mathsf{out}(\mu) = 1$
					\STATE Add $\mu$ to $\Omega_\mathsf{out}$
				\ENDIF
			\ELSE
				\FORALL {$\mu' \in \Omega^{\mu}$}
					\IF {$\mu'$ and $\mu_\mathsf{in}$ are compatible}
						\STATE $\mu^*$ := $\mu \cup \mu'$
						\STATE $k$ := $\fctsymCard^{\symPattern_1}\!(\mu) \cdot \fctsymCard^\mu\!(\mu')$
						\IF {$\mu^*\! \in \Omega_\mathsf{out}$}
							\STATE $\mathit{old}$ := $\fctsymCard_\mathsf{out}( \mu^* )$
							\STATE Adjust $\fctsymCard_\mathsf{out}$ such that $\fctsymCard_\mathsf{out}( \mu^* ) = k + \mathit{old}$
						\ELSE
							\STATE Adjust $\fctsymCard_\mathsf{out}$ such that $\fctsymCard_\mathsf{out}( \mu^* ) = k$
							\STATE Add $\mu^*$ to $\Omega_\mathsf{out}$
						\ENDIF
					\ENDIF
				\ENDFOR
			\ENDIF
		\ENDFOR
		\RETURN $M_\mathsf{out}$
	\ENDIF
\end{algorithmic}
\end{small}
\vspace{1ex}

\noindent
We omit the discussion of this case because it is very similar to the discussion of case 5 for patterns of the form $(\symPattern_1 \OpAND \symPattern_2)$.
\qed

\section{Proof of Lemma~\ref{lem:OneVariableIsBound}} \label{sec:ProofOfLemma:OneVariableIsBound}
Suppose it holds that
\begin{equation} \label{eq:proof2:Antecedent}
	\fctsymWebSafe\bigl( \tupleD{\symLeftVar, \symPPE, \symRightVar} \,|\, \symWebSafeVar \bigr) = \fctsymVars\bigl( \tupleD{\symLeftVar, \symPPE, \symRightVar} \bigr) .
\end{equation}
We have to show that $\symLeftVar \in \symWebSafeVar$ or $\symRightVar \in \symWebSafeVar$ holds. For this proof we use an induction on the possible structure of PP expression $\symPPE$.

\subsection{Base Case}
Suppose $\symPPE$ is either an IRI $\symURI \in \symAllURIs$ or of the form $!( \symURI_1 \,|\, ... \,|\, \symURI_n)$ with $\symURI_1, ... , \symURI_n \in \symAllURIs$.
By using (\ref{eq:proof2:Antecedent}) and the fact that $\fctsymVars\bigl( \tupleD{\symLeftVar, \symPPE, \symRightVar} \bigr) \neq \emptyset$, we have $\symLeftVar \in \symWebSafeVar$ (cf.~Definition~\ref{def:cond-web-safeness}).

\subsection{Induction Step}
For the induction step we distinguish four cases (which correspond to the last four cases in the grammar of PP expressions as given in Section~\ref{subsec:preliminaries}).

\subsubsection{Case 1:}
Suppose $\symPPE$ is of the form $\invPPE{\symPPE_x}$ where $\symPPE_x$ is an arbitrary PP expression.
We claim that
\begin{equation} \label{eq:proof2:Case1}
	\fctsymWebSafe\bigl( \tupleD{\symRightVar, \symPPE_x, \symLeftVar} \,|\, \symWebSafeVar \bigr) = \fctsymVars\bigl( \tupleD{\symRightVar, \symPPE_x, \symLeftVar} \bigr) .
\end{equation}
If (\ref{eq:proof2:Case1}) holds, then $\symLeftVar \!\in\! \symWebSafeVar$ or $\symRightVar \!\in\! \symWebSafeVar$ holds by induction. Hence, it remains to show~(\ref{eq:proof2:Case1}).

By Definition~\ref{def:cond-web-safeness}, we have:
\begin{align*}
	\fctsymWebSafe\bigl( \tupleD{\symLeftVar, \invPPE{\symPPE_x}, \symRightVar} \,|\, \symWebSafeVar \bigr) &= \fctsymWebSafe\bigl( \tupleD{\symRightVar, \symPPE_x, \symLeftVar} \,|\, \symWebSafeVar \bigr) .
\intertext{By using $\fctsymWebSafe\bigl( \tupleD{\symLeftVar, \invPPE{\symPPE_x}, \symRightVar} \,|\, \symWebSafeVar \bigr) = \fctsymVars\bigl( \tupleD{\symLeftVar, \invPPE{\symPPE_x}, \symRightVar} \bigr)$ (cf.~(\ref{eq:proof2:Antecedent}) above), we obtain:}
	\fctsymVars\bigl( \tupleD{\symLeftVar, \invPPE{\symPPE_x}, \symRightVar} \bigr) &= \fctsymWebSafe\bigl( \tupleD{\symRightVar, \symPPE_x, \symLeftVar} \,|\, \symWebSafeVar \bigr) .
\end{align*}
Then, with $\fctsymVars\bigl( \tupleD{\symLeftVar, \invPPE{\symPPE_x}, \symRightVar} \bigr) = \fctsymVars\bigl( \tupleD{\symRightVar, \symPPE_x, \symLeftVar} \bigr)$, we
	can verify the correctness of
	%have
(\ref{eq:proof2:Case1}).

\subsubsection{Case 2:}
Suppose $\symPPE$ is of the form $(\symPPE_x)^*$ where $\symPPE_x$ is an arbitrary PP~expression.
By using an argument similar to the argument used for the previous case, we can show that $\fctsymWebSafe\bigl( \tupleD{\symRightVar, \symPPE_x, \symLeftVar} \,|\, \symWebSafeVar \bigr) = \fctsymVars\bigl( \tupleD{\symRightVar, \symPPE_x, \symLeftVar} \bigr) .$ Then, $\symLeftVar \in \symWebSafeVar$ or $\symRightVar \in \symWebSafeVar$ holds by induction.

\subsubsection{Case 3:}
Suppose $\symPPE$ is of the form $(\symPPE_1|\symPPE_2)$ where $\symPPE_1$ and $\symPPE_2$ are arbitrary PP expressions.
We claim that:
\begin{equation} \label{eq:proof2:Case3}
	\forall i \in \lbrace 1,2 \rbrace : 
	\fctsymWebSafe\bigl( \tupleD{\symLeftVar, \symPPE_i, \symRightVar} \,|\, \symWebSafeVar \bigr) = \fctsymVars\bigl( \tupleD{\symLeftVar, \symPPE_i, \symRightVar} \bigr)
	.
\end{equation}
If (\ref{eq:proof2:Case3}) holds, then $\symLeftVar \!\in\! \symWebSafeVar$ or $\symRightVar \!\in\! \symWebSafeVar$ holds by induction. Hence, it remains to show~(\ref{eq:proof2:Case3}).

By Definition~\ref{def:cond-web-safeness}, we have:
\begin{align*}
	\fctsymWebSafe\bigl( \tupleD{\symLeftVar, (\symPPE_1|\symPPE_2), \symRightVar} \,|\, \symWebSafeVar \bigr) &= \bigcap_{ i \in \lbrace 1,2 \rbrace } \fctsymWebSafe\bigl( \tupleD{\symLeftVar, \symPPE_i, \symRightVar} \,|\, \symWebSafeVar \bigr) .
\intertext{By using $\fctsymWebSafe\bigl( \tupleD{\symLeftVar, (\symPPE_1|\symPPE_2), \symRightVar} \,|\, \symWebSafeVar \bigr) = \fctsymVars\bigl( \tupleD{\symLeftVar, (\symPPE_1|\symPPE_2), \symRightVar} \bigr)$ (cf.~(\ref{eq:proof2:Antecedent}) above), we obtain:}
	\fctsymVars\bigl( \tupleD{\symLeftVar, (\symPPE_1|\symPPE_2), \symRightVar} \bigr) &= \bigcap_{ i \in \lbrace 1,2 \rbrace } \fctsymWebSafe\bigl( \tupleD{\symLeftVar, \symPPE_i, \symRightVar} \,|\, \symWebSafeVar \bigr) .
\end{align*}
Then, with $\fctsymVars\bigl( \tupleD{\symLeftVar, (\symPPE_1|\symPPE_2), \symRightVar} \bigr) \!=\! \fctsymVars\bigl( \tupleD{\symLeftVar, \symPPE_i, \symRightVar} \bigr)$ for all $i \!\in\! \lbrace 1,2 \rbrace$, we
	can verify the correctness of
	%have
(\ref{eq:proof2:Case3}).

\subsubsection{Case 4:}
Suppose $\symPPE$ is of the form $\symPPE_1/\symPPE_2$ where $\symPPE_1$ and $\symPPE_2$ are arbitrary PP expressions. In this case, by Definition~\ref{def:cond-web-safeness}, we have:
\begin{align*}
	\fctsymWebSafe\bigl( \tupleD{\symLeftVar, \symPPE_1/\symPPE_2, \symRightVar} \,|\, \symWebSafeVar \bigr) &= \fctsymWebSafe\bigl( \symPattern' \,|\, \symWebSafeVar \bigr) \setminus \lbrace ?v \rbrace ,
\intertext{where $\symPattern' = \bigl( \tupleD{\symLeftVar,\symPPE_1,?v} \OpAND \tupleD{?v,\symPPE_2,\symRightVar} \bigr)$ and $?v \in \symAllVariables$ is an arbitrary variable such that $?v \notin \bigl( \symWebSafeVar \cup \lbrace \symLeftVar,\symRightVar \rbrace \bigr)$ and $?v \in \fctCondWebSafe{\symPattern'}{\symWebSafeVar}$. By using the fact that $\fctsymWebSafe\bigl( \tupleD{\symLeftVar, \symPPE_1/\symPPE_2, \symRightVar} \,|\, \symWebSafeVar \bigr) = \fctsymVars\bigl( \tupleD{\symLeftVar, \symPPE_1/\symPPE_2, \symRightVar} \bigr)$ (cf.~(\ref{eq:proof2:Antecedent}) above), we obtain:}
	\fctsymVars\bigl( \tupleD{\symLeftVar, \symPPE_1/\symPPE_2, \symRightVar} \bigr) &= \fctsymWebSafe\bigl( \symPattern' \,|\, \symWebSafeVar \bigr) \setminus \lbrace ?v \rbrace .
\end{align*}
Consequently, $\fctsymWebSafe\bigl( \symPattern' \,|\, \symWebSafeVar \bigr) \neq \emptyset$. Therefore, by Definition~\ref{def:cond-web-safeness}, either
\begin{enumerate}
	\item $\fctCondWebSafe{\symPattern_1'}{\symWebSafeVar} = \fctVars{\symPattern_1'}$ and $\fctCondWebSafe{\symPattern_2'}{\symWebSafeVar} = \fctVars{\symPattern_2'}$, or
	\item $\fctCondWebSafe{\symPattern_1'}{\symWebSafeVar} = \fctVars{\symPattern_1'}$ and $\fctCondWebSafe{\symPattern_2'}{\symWebSafeVar \cup \fctCVars{\symPattern_1'}} = \fctVars{\symPattern_2'}$, or
	\item $\fctCondWebSafe{\symPattern_2'}{\symWebSafeVar} = \fctVars{\symPattern_2'}$ and $\fctCondWebSafe{\symPattern_1'}{\symWebSafeVar \cup \fctCVars{\symPattern_2'}} = \fctVars{\symPattern_1'}$,
\end{enumerate}
where $\symPattern_1' = \tupleD{\symLeftVar,\symPPE_1,?v}$ and $\symPattern_2' = \tupleD{?v,\symPPE_2,\symRightVar}$ (i.e., $( \symPattern_1' \OpAND \symPattern_2' ) = \symPattern'$). W.l.o.g., we discuss the first of these three alternatives only (the discussion of the other two would be almost identical).

Then, due to $\fctCondWebSafe{\symPattern_1'}{\symWebSafeVar} = \fctVars{\symPattern_1'}$, by induction we can assume that $\symLeftVar \in \symWebSafeVar$ or $?v \in \symWebSafeVar$. However, we can rule out the latter because $?v \notin \bigl( \symWebSafeVar \cup \lbrace \symLeftVar,\symRightVar \rbrace \bigr)$ (see above). Hence, $\symLeftVar \in \symWebSafeVar$. In a similar manner it is possible to also show $\symRightVar \in \symWebSafeVar$ by using $\fctCondWebSafe{\symPattern_2'}{\symWebSafeVar} = \fctVars{\symPattern_2'}$.
\qed

\section{Proof of Lemma~\ref{lem:EquivalenceOfSwappingStarAndInverse}} \label{sec:ProofOfLemma:EquivalenceOfSwappingStarAndInverse}
Let $\symPP = \tupleD{\symLeftVar, (\symPPE)^*, \symRightVar}$ and $\symPP' = \tupleD{\symRightVar, (\symPPEi)^*, \symLeftVar}$ be two PP patterns such that $\symPPE$ is an arbitrary PP expression and 
$\symLeftVar$ and $\symRightVar$ are two variables~(i.e., $\symLeftVar, \symRightVar \in \symAllVariables$). Furthermore, let $\symWoD$ be an arbitrary \Web. We have to show that $\EvalCtx{\symPP}{\symWoD} \subseteq \EvalCtx{\symPP'}{\symWoD}$ (Claim 1) and $\EvalCtx{\symPP}{\symWoD} \supseteq \EvalCtx{\symPP'}{\symWoD}$ (Claim 2) hold.

\bigskip
\noindent
\textbf{Proof of Claim 1:}
Let $\mu^*$ be an arbitrary solution mapping such that $\mu^* \in \EvalCtx{\symPP}{\symWoD}$. W.l.o.g., we show that $\EvalCtx{\symPP}{\symWoD} \subseteq \EvalCtx{\symPP'}{\symWoD}$ by showing that $\mu^* \in \EvalCtx{\symPP'}{\symWoD}$.
To this end, by Definition~\ref{def:Context}, we have to show that $\mu^*$ satisfies the following three conditions:
\begin{itemize}
	\item[]\textit{Condition 1:} $\fctDom{\mu^*} = \lbrace \symLeftVar,\symRightVar \rbrace$,
	\item[]\textit{Condition 2:} $\mu^*(\symRightVar) \in \fctTerms{\symWoD}$, and
	\item[]\textit{Condition 3:} $\mu^*(\symLeftVar) \in \mathtt{ALWP1}(\mu(\symRightVar),\symPPEi,\symWoD)$.
\end{itemize}
On the other hand, since $\mu^* \in \EvalCtx{\symPP}{\symWoD}$, $\mu^*$ has the following three properties:
\begin{itemize}
	\item[]\textit{Property 1:} $\fctDom{\mu^*} = \lbrace \symLeftVar,\symRightVar \rbrace$,
	\item[]\textit{Property 2:} $\mu^*(\symLeftVar) \in \fctTerms{\symWoD}$, and
	\item[]\textit{Property 3:} $\mu^*(\symRightVar) \in \mathtt{ALWP1}(\mu(\symLeftVar),\symPPE,\symWoD)$.
\end{itemize}
Hence, $\mu^*$ satisfies Condition~1. To see that $\mu^*$ also satisfies Condition 2 and 3, consider Property~3. Due to this property, there exists a sequence of solution mappings $\mu_1 , ...\, , \mu_n$ and two variables $?x,?y \in \symAllVariables$ such that \enumA~$\fctDom{\mu_i} = \lbrace ?x,?y \rbrace$ for all $i \in \lbrace 0, ...\, , n \rbrace$, \enumB~$\mu_1(?x) = \mu^*(\symLeftVar)$, \enumC~$\mu_n(?y) = \mu^*(\symRightVar)$, and \enumD~$\mu_i \in \EvalCtx{\tupleD{?x,\symPPE,?y}}{\symWoD}$ for all $i \in \lbrace 0, ...\, , n \rbrace$. Due to the latter, $\mu_i(?y) \in \fctTerms{\symWoD}$  for all $i \in \lbrace 0, ...\, , n \rbrace$. Thus, with $\mu_n(?y) = \mu^*(\symRightVar)$, we have $\mu^*(\symRightVar) \in \fctTerms{\symWoD}$; i.e., $\mu^*$ satisfies Condition~2.

Moreover, by Definition~\ref{def:Context}, $\EvalCtx{\tupleD{?x,\symPPE,?y}}{\symWoD} = \EvalCtx{\tupleD{?y,\symPPEi,?x}}{\symWoD}$ and, thus, $\mu_i \in \EvalCtx{\tupleD{?y,\symPPEi,?x}}{\symWoD}$ for all $i \in \lbrace 0, ...\, , n \rbrace$. Therefore, the sequence of solution mappings $\mu_1 , ...\, , \mu_n$ can also be used to show that $\mu_1(?x) \!\in\! \mathtt{ALWP1}(\mu_n(?y),\text{\footnotesize$\symPPEi$},\symWoD)$. Due to this fact and due to $\mu_1(?x) = \mu^*(\symLeftVar)$ and $\mu_n(?y) = \mu^*(\symRightVar)$, we can verify that $\mu^*$ satisfies Condition~3.

\bigskip
\noindent
\textbf{Proof of Claim 2:}
Let $\mu^*$ be an arbitrary solution mapping such that $\mu^* \in \EvalCtx{\symPP'}{\symWoD}$. W.l.o.g., we show that $\EvalCtx{\symPP}{\symWoD} \supseteq \EvalCtx{\symPP'}{\symWoD}$ by showing that $\mu^* \in \EvalCtx{\symPP}{\symWoD}$.
To this end, by Definition~\ref{def:Context}, we have to show that $\mu^*$ satisfies the following three conditions:
\begin{itemize}
	\item[]\textit{Condition 1:} $\fctDom{\mu^*} = \lbrace \symLeftVar,\symRightVar \rbrace$,
	\item[]\textit{Condition 2:} $\mu^*(\symLeftVar) \in \fctTerms{\symWoD}$, and
	\item[]\textit{Condition 3:} $\mu^*(\symRightVar) \in \mathtt{ALWP1}(\mu(\symLeftVar),\symPPE,\symWoD)$.
\end{itemize}
On the other hand, since $\mu^* \in \EvalCtx{\symPP'}{\symWoD}$, $\mu^*$ has the following three properties:
\begin{itemize}
	\item[]\textit{Property 1:} $\fctDom{\mu^*} = \lbrace \symLeftVar,\symRightVar \rbrace$,
	\item[]\textit{Property 2:} $\mu^*(\symRightVar) \in \fctTerms{\symWoD}$, and
	\item[]\textit{Property 3:} $\mu^*(\symLeftVar) \in \mathtt{ALWP1}(\mu(\symRightVar),\symPPEi,\symWoD)$.
\end{itemize}
Due to the symmetry of these conditions and properties to the conditions and properties in the discussion of Claim~1, it is easily seen that Claim 2 can be proved by using an argument that is reverse to the argument used for proving Claim 1.
\qed
}

% % %
\end{document}